\documentclass[a4paper,10pt,openright]{article}
\usepackage[british]{babel}
\usepackage{amsmath,amssymb,stmaryrd,xcolor,latexsym,theorem,hyperref,tikz}
\numberwithin{equation}{section}
\usepackage{verbatim}

\newcommand{\assign}{:=}
\newcommand{\backassign}{=:}
\newcommand{\barsuchthat}{|}
\newcommand{\comma}{{,}}
\newcommand{\mathGamma}{\Gamma}
\newcommand{\mathLaplace}{\Delta}
\newcommand{\mathd}{\mathrm{d}}
\newcommand{\nobracket}{}

\newcommand{\tmem}[1]{{\em #1\/}}

\newcommand{\tmop}[1]{\ensuremath{\operatorname{#1}}}
\newcommand{\tmtextbf}[1]{\text{{\bfseries{#1}}}}
\newcommand{\tmtextit}[1]{\text{{\itshape{#1}}}}
\newcommand{\tmverbatim}[1]{\text{{\ttfamily{#1}}}}

\newenvironment{proof}{\noindent\textbf{Proof\ }}{\hspace*{\fill}$\Box$\medskip}
\newcounter{nnacknowledgments}

{\theorembodyfont{\rmfamily}\newtheorem{acknowledgments*}[nnacknowledgments]{Acknowledgments}}
\newtheorem{theorem}{Theorem}[section]
\newtheorem{corollary}[theorem]{Corollary}
\newtheorem{definition}[theorem]{Definition}
\newtheorem{example}[theorem]{Example}
\newtheorem{lemma}[theorem]{Lemma}
\newcounter{nnnotation}

\newtheorem{notation*}[nnnotation]{Notation}
\newtheorem{proposition}[theorem]{Proposition}
\newtheorem{remark}[theorem]{Remark}

\newcommand{\tmkeywords}{\textbf{Keywords:} }

\definecolor{AlColor}{RGB}{77,0,154}
\definecolor{ClaColor}{RGB}{255,155,74}

\begin{document}

\title{On the stochastic Sine-Gordon model:\\ an interacting field theory approach}

\author{
	Alberto Bonicelli\thanks{AB:
		Dipartimento di Fisica,
		Universit\`a degli Studi di Pavia \& INFN, Sezione di Pavia, \& Indam, Sezione di Pavia 
		Via Bassi 6,
		I-27100 Pavia,
		Italia;
		\mbox{alberto.bonicelli01@universitadipavia.it}}
	\and
	Claudio Dappiaggi\thanks{CD:
		Dipartimento di Fisica,
		Universit\`a degli Studi di Pavia \& INFN, Sezione di Pavia, \& Indam, Sezione di Pavia
		Via Bassi 6,
		I-27100 Pavia,
		Italia;
		\mbox{claudio.dappiaggi@unipv.it}}
	\and
	Paolo Rinaldi\thanks{PR: Institute for Applied Mathematics, Universit\"at Bonn, 
		Endenicher Allee 60,
		D-53115 Bonn,
		Germany;
		\mbox{rinaldi@iam.uni-bonn.de}
}}

\maketitle

\begin{abstract}
  We investigate the massive Sine-Gordon model in the finite ultraviolet regime on the two-dimensional Minkowski spacetime $(\mathbb{R}^2,\eta)$ with an additive Gaussian white noise. In particular we construct the expectation value and the correlation functions of a solution of the underlying stochastic partial differential equation (SPDE) as a power series in the coupling constant, proving ultimately uniform convergence. This result is obtained combining an approach first devised in \cite{DDRZ21} to study SPDEs at a perturbative level with the one discussed in \cite{Bahns-Rejzner_Sine-Gordon} to construct the quantum sine-Gordon model using techniques proper of the perturbative, algebraic approach to quantum field theory (pAQFT). At a formal level the relevant expectation values are realized as the evaluation of suitably constructed functionals over $C^\infty(\mathbb{R}^2)$. In turn, these are elements of a distinguished algebra whose product is a deformation of the pointwise one, by means of a kernel which is a linear combination of two components. The first encompasses the information of the Feynmann propagator built out of an underlying Hadamard, quantum state, while the second encodes the correlation codified by the Gaussian white noise. In our analysis, first of all we extend the results obtained in \cite{BPR23-Sine-Gordon-Massive,Bahns-Rejzner_Sine-Gordon} proving the existence of a convergent modified version of the S-matrix and of an interacting field as elements of the underlying algebra of functionals. Subsequently we show that it is possible to remove the contribution due to the Feynmann propagator by taking a suitable $\hbar\to 0^+$-limit, hence obtaining the sought expectation value of the solution and of the correlation functions of the SPDE associated to the stochastic Sine-Gordon model. 
\end{abstract}


\tmkeywords{Stochastic Sine-Gordon model, Perturbative algebraic quantum field theory, Microlocal analysis}

{\tableofcontents}

\section{Introduction }\label{Sec_intro}

The investigation of nonlinear stochastic partial differential equations (SPDEs) represents one of the most thriving branches of research in mathematics, mostly thanks to the formulation of different successful frameworks aimed at studying their underlying solution space. Regularity structures \cite{Hairer:2014sma} or paracontrolled calculus \cite{Gubinelli_et_al} have proven to be two complementary, albeit rather different, approaches which have allowed to prove existence and uniqueness of the solutions of a large class of nonlinear elliptic or parabolic SPDEs. At the same time these equations are remarkably important in many physical models, stemming from interface dynamics to stochastic quantization. Yet, in order to build a solid and longstanding bridge between the probabilistic and analytic approach to SPDEs and the physical models, it is mandatory to be able to provide as much explicit information as possible on the underlying solutions and on their correlation functions. With this goal in mind, in \cite{DDRZ21} it has been proposed a new method for the construction of the solutions and of the correlation functions of nonlinear SPDEs, largely inspired by the algebraic approach to quantum field theory. On the one hand this framework has the advantage of allowing to encompass the renormalization procedure and freedoms which are often a key ingredient in the solution theory of a nonlinear SPDE, without resorting to any specific $\epsilon$-regularization scheme. On the other hand, it allows to establish an algorithmic procedure to construct both the solutions and the $n$-point correlation functions, though as a formal power series in the coupling constant which rules the nonlinear term in the equation of motion.

It is important to highlight that the algebraic approach, devised in \cite{DDRZ21} has the additional net advantage to be applicable also to nonlinear SPDEs which do not lie in the subcritical regime, a necessary prerequisite instead when applying the theory of regularity structures or of paracontrolled calculus, see, {\it e.g.}, \cite{BCDR23} and \cite{BDR23}. One might be tempted to perceive subcriticality essentially as yielding a constraint on the nonlinear potential ruling the underlying dynamics, but this viewpoint is correct provided that the fundamental solutions associated to the operator ruling the linear contribution to the underlying SPDE is regularizing. This is indeed the case when considering second order elliptic or parabolic differential operators with smooth coefficients, as it occurs in the vast majority of the models in the literature. On the contrary this feature is no longer present when one considers hyperbolic SPDEs, for example the stochastic, nonlinear wave equation. In this case one has often to resort to a case by case analysis, see for example \cite{Gubinelli1,Gubinelli2} in order to prove existence and uniqueness of the underlying solutions. 

For this reason, in comparison to the elliptic and parabolic scenarios, hyperbolic SPDEs have been analyzed less in depth. In this paper we shall focus our attention on a specific instance of this class of equations, which is known as the stochastic massive sine-Gordon model in the so called finite ultraviolet regime. Its parabolic counterpart has been studied in \cite{Hairer:2014sma}, while the hyperbolic scenario in two space dimensions has been investigated recently in \cite{Oh:2020}. If one focuses instead the attention on one space dimension, the sine-Gordon equation without an additive Gaussian, white noise as a source, is a very important and thoroughly studied model in quantum field theory, especially due to the underlying integrability properties and to its connections with the Thirring model, which is at the heart of a phenomenon known as Bosonization \cite{Col75,Thirring}. In particular, within the framework of the perturbative, algebraic approach to quantum field theory, it has gained a lot of attention in the past decade since it represents one of the few notable examples where it is possible to prove convergence of the perturbative series defining the S-matrix of the model, see \cite{BPR23-Sine-Gordon-Massive,Bahns-Rejzner_Sine-Gordon}. This result prompts naturally the question whether the techniques used in these papers can be adapted to be applicable also to the analysis of the stochastic sine-Gordon equation on the two-dimensional Minkowski spacetime. In this work we shall prove that this is indeed the case and this allows us to obtain two notable, connected results. On the one hand we are able to establish for the first time the convergence of the perturbative series which lies at the heart of the algebraic approach to the construction of the solution and of the correlation functions of the stochastic sine-Gordon equation. On the other hand, this opens the path to the possibility of combining the recent analysis in \cite{BCDR23} on the stochastic Thirring model to establish a stochastic counterpart of the phenomenon of Bosonization using techniques inspired by the algebraic approach to quantum field theory.   

More in detail, our approach to the stochastic sine-Gordon equation can be divided in two main steps. In the first one, we follow the rationale of \cite{BPR23-Sine-Gordon-Massive,Bahns-Rejzner_Sine-Gordon}, namely we consider a suitable algebra of functionals defined on $C^\infty(\mathbb{R}^2)$, where $\mathbb{R}^2$ plays here the r\^{o}le of the underlying Minkowski spacetime. At the beginning we consider a commutative pointwise product and, subsequently, in the spirit of the perturbative, algebraic approach to quantum field theory, we deform it by means of a suitable kernel which encompasses the information both of an underlying Feynmann propagator, the building block of the time-ordered product in an interacting quantum theory, and of the correlation function of the Gaussian process codified by the underlying white noise. It is important to stress that this entails a deviation from the approach in \cite{BPR23-Sine-Gordon-Massive} since, in this paper, no stochastic effect has been considered. Yet we are able to show that all convergence results for the $S$-matrix and for the interacting quantum field can be generalized to this scenario, although they hold true in an arbitrary, but fixed compact region of the two-dimensional Minkowski spacetime. The outcome is a model which mixes both a quantum and a stochastic behaviour. Yet the effect of the former can be sharply disentangled from the latter since the action of the Feynman propagator is always tagged by the presence of a multiplicative constant, namely $\hbar$. This feature leads to the second step of our approach in which we investigate the limit as $\hbar$ tends to $0$ of all relevant functionals and we prove that such limit is always well-defined. In this way we are also able to show that we are actually constructing explicitly the functionals encoding the information on the expectation value of the interacting solution of the stochastic sine-Gordon equation and on the associated $n$-point correlation functions. Observe that taking the expectation value corresponds in the algebraic approach to considering the evaluation of the corresponding functional on the zero configuration. 

\vskip .2cm

\noindent The paper is organized as follows: in Section \ref{Sec:interacting-AQFT} we review the building blocks of the algebraic approach to an interacting quantum field theory, in particular the S-matrix and the Bogoliubov map in Section \ref{Sec: time-ordered stuff}. The goal of Section \ref{Sec:microlocal-approach-to-spdes} is instead to present succinctly the content of \cite{DDRZ21} and the interplay with microlocal analysis. The specialization of the structures outlined in the first sections to the specific case of the sine-Gordon model is the content of Section \ref{Sec: Functionals and sine-Gordon}, while in Section \ref{Sec:Strategy} we present the strategy that we plan to follow to construct the solutions and the $n$-point correlation functions of the stochastic sine-Gordon equation. Section \ref{Sec:interplay} contains the first of our main results since we introduce the notion of the $Q-S$-matrix first as a formal power series in the underlying coupling constant. Subsequently in Section \ref{Sec: Convergence of the Q-S matrix} we prove uniform convergence in Theorem \ref{Thm:estimates-Q-S-matrix-I} and in Corollary \ref{Cor: convergence S matrix}. This allows us in turn to establish also convergence both of the interacting field in Section \ref{Sec:conv-interacting-field} and of the $n$-point correlation functions in Section \ref{Sec:convergence-corr-functions}. The main result of our work is discussed instead in Section \ref{Sec:classical-limit}, namely in Theorem \ref{Thm: hbar limit}, we prove that the limit as $\hbar\to 0^+$ of all relevant functionals exist and this allows us to establish the existence of the expectation value of the solutions of the stochastic sine-Gordon equation as well as of the associated $n$-point correlation functions, as suitably convergent power series in the underlying coupling constant.

\subsection{Notation and Conventions}\label{Sec: Notation}
In this short section we introduce the stochastic sine-Gordon equation and we take the chance to fix the notation and the conventions that we use in this work. Throughout the paper we denote by $\mathbb{M}$, a generic $d$-dimensional globally hyperbolic spacetime. With $\mathcal{E} (\mathbb{M}) \assign C^{\infty} (\mathbb{M})$, $\mathcal{D} (\mathbb{M}) \assign C_0^{\infty} (\mathbb{M})$ we indicate respectively the space of smooth field configurations and of test functions. In addition $\mathcal{D}'(\mathbb{M})$ is the space of distributions over $\mathbb{M}$, dual to $\mathcal{D}(M)$, while $\mathcal{E}^\prime(\mathbb{M})$ is the space of compactly supported distributions dual to $\mathcal{E}(\mathbb{M})$. We will be particularly interested in the case where the r\^{o}le of $\mathbb{M}$ is played by the two-dimensional Minkowski spacetime $\mathbb{R}^2$ which is endowed with the standard Minkowski metric $\eta$ of signature $(+,-)$. On top of it we consider the {\em stochastic sine-Gordon equation}
\begin{equation}\label{Eq: sine-Gordon equation_intro}
	(\Box + m^2) \hat{\psi} + \lambda g a \sin(a\hat{\psi}) = \hat{\xi}, 
\end{equation}
where $\Box=\partial^2_t-\partial^2_x$ is the d'Alembert wave operator, $\hat{\psi}$ is a random valued distribution, whereas $a\in \mathbb{R}$ is chosen according to the finite ultraviolet regime, namely $a^2<4\pi/\hbar$, while $g \in \mathcal{D} (\mathbb{R}^2)$. In addition $\hat{\xi}$ denotes a space-time white noise, namely a Gaussian centered random distribution whose
two-point correlation function is, at the level of integral kernel,
\begin{equation}\label{Eq: covariance white noise_intro}
	\mathbb{E} [\hat{\xi} (z) \hat{\xi} (z^\prime)] = \delta (z - z^\prime),
\end{equation}
where we adopt the notation $z=(t,x)$, $t$ being the time coordinate.

\section{Setting}

The goal of this section is both to fix the notation and the conventions adopted in this work, and to give a succinct overview of the key definitions and results  concerning the algebraic approach to interacting quantum fields (AQFT). For more information, see \cite{BPR23-Sine-Gordon-Massive, BFDY15, DDR20, Rejzner-pAQFT}. Subsequently we sketch the key ideas at the heart of an approach to stochastic PDEs inspired by AQFT and analyzed in \cite{BCDR23,BDR23,DDRZ21}. These frameworks will serve as the foundation for the analysis of the stochastic sine-Gordon model.

\subsection{Interacting algebraic quantum field
	theory}\label{Sec:interacting-AQFT}

Algebraic Quantum Field Theory (AQFT) is a two step approach to quantization, which is tailored to be applicable to as many models as possible, regardless whether they are defined on a Lorentzian or an Euclidean manifold. At first, given a physical system, one needs to construct a suitable $*$-algebra of observables, say $\mathcal{A}$ which encompasses all structural properties, ranging from dynamics to causality or to the canonical commutation or anti-commutation relations.
 Subsequently, on top of $\mathcal{A}$, one must identify an {\em algebraic state}, that is a linear, normalized and positive functional $\omega:\mathcal{A}\to\mathbb{C}$. 
 As a consequence of the renown GNS theorem one can recover from the pair $(\mathcal{A},\omega)$ the standard probabilistic interpretation, proper of quantum theories. 
 In the past decade, especially in connection to interacting quantum field theories, it has become clear that an advantageous way to construct a concrete $*$-algebra $\mathcal{A}$ consists of identifying it as a collection of suitable functionals on the underlying space of smooth field configurations. 
 This procedure allows, on the one hand, to encompass the above mentioned structural properties in terms of a deformation of the pointwise 
product among functionals, while, on the other hand, it facilitates the possibility of including specific constraints on the existing singular structures, a feature which is of paramount relevance to deal with renormalization in the algebraic setting -- see, {\it e.g.}, \cite{BFDY15, Rejzner-pAQFT}.

We denote by $\mathcal{F}(\mathbb{M})$ the space of complex-valued, continuous, linear functionals over $\mathcal{E}(\mathbb{M})$. It is worth recalling that $\mathcal{F}(\mathbb{M})$ comes with a natural notion of functional derivative, which allows to identify a class of distinguished functionals, namely the polynomial ones. 
\begin{definition}\label{Def: functional derivatives}
	Given $F\in\mathcal{F}(\mathbb{M})$, we call $F^{(k)}$, $k\geq 1$, its $k$-th order functional derivative, namely $F^{(k)}\in\mathcal{E}'(\underbrace{\mathbb{M}\times\ldots\times \mathbb{M}}_{k}; \mathcal{F}(\mathbb{M}))$ such that
	\begin{align}\label{Eq: functional ferivatives}
		F^{(k)}(\eta_1\otimes\ldots\otimes\eta_k;\eta):=\frac{\partial^{k}}{\partial s_1\ldots\partial s_k} F(\eta+s_1\eta_1+\ldots+s_k\eta_k)\Big\vert_{s_1=\ldots={s}_{k}=0},
	\end{align}
	for all $\eta,\eta_1,\ldots,\eta_k\in\mathcal{E}(\mathbb{M})$. Accordingly, we define the \textbf{directional derivative} along $\varphi\in\mathcal{E}(\mathbb{M})$ as
	\begin{align*}
		\delta_\varphi:\mathcal{F}(\mathbb{M})\rightarrow\mathcal{F}(\mathbb{M}),\qquad[\delta_\varphi F](\eta):=F^{(1)}(\varphi;\eta).
	\end{align*}
	A functional $F\in\mathcal{F}(\mathbb{M})$ is said to be \textbf{polynomial}, $F\in\mathcal{F}_{\textrm{Pol}}(\mathbb{M})$, if there exists $n\in\mathbb{N}_0$ such that $F^{(k)}=0$ for all $\{k\geq n\}$. In addition we define the (spacetime) {\bf support} of a functional $F$ as 
	\begin{equation}\label{Eq: Support of a Functional}
		\mathrm{supp}(F):=\{x\in\mathbb{M}\;|\;\forall U\in\mathcal{N}_x\;\exists\varphi,\psi\in\mathcal{E}(\mathbb{M})\;\textrm{with}\;\mathrm{supp}(\varphi)\subset U,\;\textrm{such that}\;F(\varphi+\psi)\neq F(\varphi)\},
	\end{equation}
where $\mathcal{N}_x$ denotes the family of all open subsets of $\mathbb{M}$ such that $x\in\mathcal{N}_x$.
\end{definition}

In the applications both to interacting quantum field theories and to non-linear stochastic partial differential equations, we will be forced to consider either products among the derivatives of suitable functionals or their composition with specific propagators. 
Some of these operations are \textit{a priori} ill-defined and one needs to resort to techniques proper of microlocal analysis to overcome these hurdles. In order to keep the length of this work at bay we shall assume that the reader is familiar with the basic concepts of this framework and we refer to \cite{Hormander-83} for all information or to \cite[Appendix B]{DDRZ21} for a succinct summary of the key ingredients. In view of these considerations, we look for a restricted class of functionals, characterized by a constraint on the singular structure of the functional derivatives.

\begin{definition}\label{Def:functionals}
	Denoting with $\mathcal{F}(\mathbb{M})$ the space of continuous, complex-valued functionals on $\mathcal{E}(\mathbb{M})$,  we define
	\begin{itemize}
		\item the \tmtextbf{microcausal functionals} as
		\[ \mathcal{F}_{\mu c}(\mathbb{M}) \assign \left\{ F\in\mathcal{F}(\mathbb{M})\; \barsuchthat\; F^{(n)} \in \mathcal{E}'
		(\mathbb{M}^n), \quad \tmop{WF} (F^{(n)}) \cap \Bigl[ \bigcup_{p\in\mathbb{M}} (\bar{V}_p^+ \cup
		\bar{V}_p^-)\Bigr]= \emptyset, \quad n \in \mathbb{N} \right\}, \]      
		where $F^{(n)}$ is the $n$-th functional derivative of $F$ as per Definition \ref{Def: functional derivatives}, while
		$\bar{V}_p^+$ and $\bar{V}_p^-$ are, respectively, the sets of future-pointing
		and past-pointing covectors in $T^{\ast}_p \mathbb{M}$; 
		
		\item the \tmtextbf{regular functionals} as
		\[ \mathcal{F}_{\tmop{reg}}(\mathbb{M}) \assign \left\{ F \in \mathcal{F}_{\mu c}(\mathbb{M})
		\;\barsuchthat\; F^{(n)} \in \mathcal{D}(\mathbb{M}^n)\hookrightarrow\mathcal{E}^\prime(\mathbb{M}^n), \quad n \in
		\mathbb{N} \right\} ; \]
		\item the \tmtextbf{local functionals} as
		\[ \mathcal{F}_{\tmop{loc}}(\mathbb{M}) \assign \left\{ F \in \mathcal{F}_{\mu c}(\mathbb{M})
		\;\barsuchthat\; F^{(1)} \in \mathcal{E}(\mathbb{M})\hookrightarrow\mathcal{D}^\prime(\mathbb{M}), \quad \tmop{supp}
		(F^{(n)}) \subset \tmop{Diag}_n \subset \mathbb{M}^n, \quad n \in
		\mathbb{N} \right\}, \]
		where $\tmop{Diag}_n$ denotes the total diagonal of $\mathbb{M}^n$, namely
		\[ \tmop{Diag}_n \assign \{ (x, \ldots, x) \in \mathbb{M}^n\; \barsuchthat\; x
		\in \mathbb{M} \} . \]
	\end{itemize}
\end{definition}

\begin{remark}\label{Rem: Polynomial Functionals}
	With reference to Definition \ref{Def: functional derivatives}, when we need to consider in addition only polynomial functionals, we shall employ the symbol $\mathcal{F}_{\mu c/reg/loc}^{p}(\mathbb{M})\doteq \mathcal{F}_{\mu c/reg/loc}(\mathbb{M})\cap\mathcal{F}_{\textrm{Pol}}(\mathbb{M})$.
\end{remark}

\begin{example}\label{Ex: basic functionals}
	In order to introduce functionals which will play a prominent r\^ole in our construction, let us delve in two simple, yet informative examples. The first one is the so-called \tmtextbf{smeared linear field}: Given $f \in
	\mathcal{D} (\mathbb{M})$, we set
	\[ \Phi_f : \varphi \mapsto \Phi_f (\varphi) \assign \int_{\mathbb{M}}
	\mathd \mu_x\, f (x) \varphi (x), \]
	with $\mathd \mu_x$ is the metric induced measure. A direct computation of the first functional derivative, see Equation \eqref{Eq: functional ferivatives}, entails that its integral kernel reads
	\begin{equation*}
		\Phi_f^{(1)}(x,y)=f(x)\delta_{Diag_2}(x,y),
	\end{equation*}
	where $\delta_{Diag_2}\in\mathcal{D}'(\mathbb{M}\times\mathbb{M})$ acts as
	\begin{equation}
		\delta_{Diag_2}(h):=\int_{\mathbb{M}}\mathd \mu_x\, h (x,x),\qquad \forall h\in\mathcal{D}(\mathbb{M}\times\mathbb{M}),
	\end{equation}
	while all higher order derivatives vanish. Hence the smeared linear field is a polynomial local functional, $\Phi_f\in\mathcal{F}^{\,p}_{loc}(\mathbb{M})$ for all $f\in\mathcal{D}(\mathbb{M})$.
	As a second example of local functional we consider the
	\textbf{smeared vertex operator}
	\begin{equation}\label{Eq: vertex functional}
		V_{a, f} : \varphi \mapsto V_{a, f} (\varphi) \assign \int_{\mathbb{M}}
		\mathd \mu_x f (x) e^{ia \varphi (x)}, \quad a \in \mathbb{R},\;\textrm{and}\, f
		\in \mathcal{D} (\mathbb{M}),
	\end{equation}
	where locality is once more a by-product of Definition \ref{Def: functional derivatives}. We observe that Equation \eqref{Eq: vertex functional} does not identify a polynomial functional, but it is of primary relevance since it encodes the information of the interaction term in the Lagrangian of the sine-Gordon model, {\it cf.} Equation \eqref{Eq: sine-Gordon equation_intro}, 
	\[ V_g \assign \frac{V_{a, g} + V_{- a, g}}{2}, \]
	with $g \in \mathcal{D} (\mathbb{M})$ and $a\in\mathbb{R}_+$. 
\end{example}

Starting from Definition \ref{Def:functionals}, we can endow $\mathcal{F}_{\mu c}(\mathbb{M})$ with the structure of a commutative $*$-algebra denoted by $\mathcal{A}_{\mu c}(\mathbb{M})\doteq (\mathcal{F}_{\mu c}(\mathbb{M}), \cdot, \ast)$ and constituted by the following data:
\begin{itemize}
	\item a $*$-operation $\ast:\mathcal{F}_{\mu c}(\mathbb{M})\to\mathcal{F}_{\mu c}(\mathbb{M})$ such that, for all $F\in\mathcal{F}_{\mu c}(\mathbb{M})$
	\[ F^{\ast} (\varphi) \assign \overline{F (\varphi)}. \]
	\item a product $\cdot:\mathcal{F}_{\mu c}(\mathbb{M})\times\mathcal{F}_{\mu c}(\mathbb{M})\to\mathcal{F}_{\mu c}(\mathbb{M})$ such that, for all $F,G\in\mathcal{F}_{\mu c}(\mathbb{M})$,
	\begin{equation}\label{Eq: tensor product}
		F \cdot G =\mathcal{M} (F \otimes G), 
	\end{equation}
	where $\mathcal{M}$ denotes the pullback on $\mathcal{F}_{\mu c} \otimes
	\mathcal{F}_{\mu c}$ via the diagonal map
	\begin{align*}
		\iota :\, &\mathcal{E} (\mathbb{M})\rightarrow \mathcal{E} (\mathbb{M}) \times \mathcal{E} (\mathbb{M})\\
		&\iota (\varphi) \assign (\varphi, \varphi).
	\end{align*}
\end{itemize}
Observe that we call $\cdot$, as per Equation \eqref{Eq: tensor product}, the pointwise product between functionals since, for all 
$\varphi \in \mathcal{E} (\mathbb{M})$,
\begin{equation}\label{Eq: pointwise product}
	(F \cdot G) (\varphi) = F (\varphi) G (\varphi).
\end{equation}

\begin{definition}\label{Def: topology on microcausal functionals}
	Let $\mathcal{A}_{\tmop{cl}}(\mathbb{M})$ be the algebra of microcausal functionals. We say that a family $\{F_n \}_{n \in \mathbb{N}}$, $F_n \in
	\mathcal{A}_{\tmop{cl}}(\mathbb{M})$ converges to a functional $F \in
	\mathcal{A}_{\tmop{cl}}(\mathbb{M})$ for $n \rightarrow \infty$ if, for all $\ell \in
	\mathbb{N}$ and for any $\varphi \in \mathcal{E} (\mathbb{M})$ it holds that
	$F_n^{(\ell)} (\varphi)$ converges to $F^{(\ell)} (\varphi)$ as $n \rightarrow \infty$
	in the weak $\ast$-topology of $\mathcal{E}^\prime (\mathbb{M}^\ell)$. 
\end{definition}

We observe that in the above definition the subscript $\tmop{cl}$ stands for classic since the product is the classical pointwise one.


Up to this point, all algebras that we have considered do not carry any specific information either on an underlying dynamics or on a quantization scheme. In order to encode eventually these data, the algebraic approach calls for considering a deformation of the product $\cdot$ introduced above. This is codified by means of a formal deformation parameter which is denoted by $\hbar$ and by a bidistribution $K\in\mathcal{D}^\prime(\mathbb{M}\times\mathbb{M})$ whose explicit form depends on the case in hand. More precisely one switches from $\mathcal{A}_{\mu c}(\mathbb{M})$ to $(\mathcal{F}_{\mu c}(\mathbb{M}), \star_{\hbar K}, \ast)$ such that, for all $F,G\in\mathcal{F}_{\mu c}(\mathbb{M})$,
\begin{equation}\label{Eq:deform-quantisation}
	F \star_{\hbar K} G =\mathcal{M} \circ e^{D_{\hbar K}}  [F \otimes G],
	\quad D_{\hbar K} \assign \left\langle \hbar K, \frac{\delta}{\delta
		\varphi} \otimes \frac{\delta}{\delta \varphi} \right\rangle \assign
	\int_{\mathbb{M}^2} \mathd \mu_x \mathd \mu_y \hbar K (x, y)
	\frac{\delta}{\delta \varphi (x)} \otimes \frac{\delta}{\delta \varphi (y)},
\end{equation}
where $K (x, y)$ is the formal integral kernel of $K$. Henceforth we shall also refer to this class of products as
{\tmem{exponential products}}.

\begin{remark}\label{Rem:deformation-map}
	Observe that the exponential product in Equation~\eqref{Eq:deform-quantisation} can be conveniently rewritten as
	\begin{equation}\label{Eq: algebra homomorphism}
		F \star_{\hbar K} G = \Gamma_{\hbar K}  [\Gamma_{\hbar K}^{- 1} (F)
		\Gamma_{\hbar K}^{- 1} (G)],
	\end{equation}
	where we introduced the deformation map $\Gamma_{\hbar K}:\mathcal{F}_{\mu c}(\mathbb{M})\rightarrow \mathcal{F}_{\mu c}(\mathbb{M})$ defined as
	\begin{equation}\label{Eq: star product}
		\Gamma_{\hbar K} = e^{\frac{1}{2} \mathcal{D}_{\hbar K}},
		\qquad \mathcal{D}_{\hbar K} = \left\langle \hbar K,
		\frac{\delta^2}{\delta \varphi^2} \right\rangle = \int_{\mathbb{M}^2}
		\mathd \mu_x \mathd \mu_y \hbar K (x, y) \frac{\delta^2}{\delta \varphi
			(x) \delta \varphi (y)} .
	\end{equation}
\end{remark}

Observe that Equation \eqref{Eq:deform-quantisation} is a priori only a formal expression on account of two potential, distinct issues:
\begin{enumerate}
	\item since $K\in\mathcal{D}^\prime(\mathbb{M}\times\mathbb{M})$, the action of $D_{\hbar K}$ and of its powers might be ill-defined. This is subordinated to the singular structure both of $K$ and of the functional derivatives of the underlying functionals. 
	As we shall see, in our investigation, being the background of low dimension, this will not be an issue. More in general, this is handled by resorting, if necessary, to a renormalization procedure.
	\item the action $\mathcal{M} \circ e^{D_{\hbar K}}$ yields a priori only a formal power series in $\hbar$ unless one proves convergence with respect to a suitable topology or considers polynomial functionals which entail that only a finite number of non vanishing contributions exist.  
\end{enumerate}

\begin{remark}\label{Rem: Advanced and Retarded Propgators}
	In the preceding discussion, particularly in defining $\mathcal{F}(\mathbb{M})$ and its distinguished subspaces as per Definition \ref{Def:functionals}, we have only considered kinematic configurations $\varphi\in\mathcal{E}(\mathbb{M})$ and no information on an underlying dynamics has been assumed. Yet, in many concrete models, among which the sine-Gordon, one assumes that the field abides by suitable equations of motions of the form 
	\begin{equation}\label{Eq: nonlinear dynamics}
		P\varphi+V^{(1)}[\varphi]=0,
	\end{equation}
	where $V:\mathbb{R}\to\mathbb{R}$ is a non linear potential while $P$ is a normally hyperbolic operator, see \cite{Baer}. In the following and in view of the application to the Lorentzian sine-Grordon model, we shall always assume that an underlying dynamics of the form of Equation \eqref{Eq: nonlinear dynamics} has been chosen. For definiteness a reader can think of $P$ as being the Klein-Gordon operator, namely $P=\Box-m^2$ where $\Box$ is the d'Alembert wave operator built out of the underlying metric, although such assumption is not strictly necessary as far as the content of this section is concerned. 
	
	This entails that, being $\mathbb{M}$ globally hyperbolic, there exists unique advanced and retarded fundamental solutions $\Delta^{A/R}:\mathcal{D}(\mathbb{M})\to\mathcal{E}(\mathbb{M})$ such that $\textrm{supp}(\Delta^{A/R}(f))\subseteq J^\mp(\textrm{supp}(f))$ for all $f\in\mathcal{D}(\mathbb{M})$. 
	
	As discussed, {\it e.g.} in \cite{BGP}, these propagators are the building block for implementing in a covariant way the canonical commutation relations (CCRs). More precisely, with reference to Equation \eqref{Eq:deform-quantisation}, we can make the identification $K =\frac{i}{2} \Delta$, where $\Delta \assign \Delta^R - \Delta^A$ again is the {\tmem{causal propagator}}. As a matter of fact, one can observe that, considering the smeared linear fields as per Example \ref{Ex: basic functionals}, it holds that $\forall f_1, f_2 \in \mathcal{D} (\mathbb{M})$
	\begin{equation}\label{Eq:CCR}
		[\Phi_{f_1}, \Phi_{f_2}]_{\star_{i 2^{- 1} \hbar \Delta}} = \Phi_{f_1}
		\star_{i 2^{- 1} \hbar \Delta} \Phi_{f_2} - \Phi_{f_2} \star_{i 2^{- 1}
			\hbar \Delta} \Phi_{f_1} = i\hbar\{\Phi_{f_1}, \Phi_{f_2}\}= i \hbar \langle f_1, \Delta f_2 \rangle, . 
	\end{equation}
where we have also introduced the symbol $\{,\}$ denoting the classical {\bf Poisson brackets}. 
\end{remark}


Observe that, if we stick with the choice of $K=\frac{i}{2}\Delta$ as deformation kernel, the product $\star_{i 2^{- 1} \hbar \Delta}$ as per Equation \eqref{Eq:deform-quantisation} is well-defined provided that we consider functionals lying in $\mathcal{F}_{\tmop{reg}}^p(\mathbb{M})$ as per Remark \ref{Rem: Polynomial Functionals}. Yet, such class is too small to encompass all relevant observables proper of a quantum field theory and, therefore, one needs to account for more singular functionals. Observe that \cite{Radzikowski:1996ei}
\begin{equation}\label{Eq: WF Delta}
	\textrm{WF}(\Delta)=\{(x,k_x,y,-k_y)\in T^*\left(\mathbb{M}\times\mathbb{M}\right)\setminus\{0\}\;|\;(x,k_x)\sim (y,k_y)\},
\end{equation}
where $\sim$ entails that the point $x$ and $y$ are connected by a lightlike geodesic $\gamma$ such that $k_x$ is coparallel to $\gamma$ at $x$ while $k_y$ is obtained by means of a parallel transport of $k_x$ along $\gamma$. This implies that, in general, given $F,G\in\mathcal{F}_{\mu c}(\mathbb{M})$, $F\star_{i 2^{- 1} \hbar \Delta} G$ is ill-defined at a microlocal level. This hurdle can be circumvented by means of a {\tmem{normal
		ordering procedure}} \cite{DDR20,Hollands-Wald-01}. Restricting to regular functionals, it amounts to considering, at the level of formal power series in $\hbar$,
\begin{equation}\label{Eq: normal ordering}
	:F:_{\omega} \assign \Gamma^{- 1}_{\hbar\omega} F \in
	\mathcal{F}_{\tmop{reg}} \llbracket \hbar \rrbracket, \qquad \qquad F \in
	\mathcal{F}_{\tmop{reg}},
\end{equation}
where $\Gamma_{\hbar \omega}$ is defined as in Equation \eqref{Eq: star product}. Here $\omega \in \mathcal{D}' (\mathbb{M}\times\mathbb{M})$ denotes the so-called
{\tmem{Hadamard parametrix}}, characterized by $(P\otimes\mathbb{I})\omega=(\mathbb{I}\otimes P)\omega\in C^\infty(\mathbb{M}\times\mathbb{M})$ and whose anti-symmetric part is $i 2^{- 1} \Delta$ . Moreover, it satisfies the {\tmem{microlocal spectrum condition}}, namely
\begin{equation}
	\tmop{WF} (\omega) = \{(x, y ; k_x, k_y) \in T^{\ast} \mathbb{M}^2 \setminus
	\{0\}\;\barsuchthat\; (x, k_x) \sim (y, k_y), k_x \vartriangleright 0\},
	\label{Eq:msc}
\end{equation}
where $\sim$ has the same meaning as in Equation \eqref{Eq: WF Delta}, while $k_x
\vartriangleright 0$ entails that the covector $k_x$ is future pointing. This condition on the wavefront set of the Hadamard parametric entails that Equation \eqref{Eq: normal ordering} is well-defined also for microcausal functionals $\mathcal{F}_{\mu c}(\mathbb{M})$ as in Definition \ref{Def:functionals}. This feature is thoroughly examined in \cite{Fredenhagen:2012sb} and therefore we omit entering into the technical details. To conclude this succinct overview, we highlight the following two comments:
\begin{enumerate}
	\item In view of the preceding comment we shall denoted by $\mathcal{F}_{\mu c}\llbracket\hbar\rrbracket$ the normal ordered space of microcausal functionals and the ensuing algebra as the triple $(\mathcal{F}_{\mu c}\llbracket\hbar\rrbracket,\star_{\hbar\omega},\ast)$
	\item most notably the product $\star_{\hbar\omega}$ does not affect the canonical commutation relations, since, in view of Example \ref{Ex: basic functionals}, for any $f_1,f_2\in\mathcal{D}(\mathbb{M})$, it holds that
	\[ [\Phi_{f_1}, \Phi_{f_2}]_{\star_{\hbar\omega}} = i \hbar \langle f_1,
	\Delta f_2 \rangle. \]
\end{enumerate}

\subsubsection{S-matrix and Bogoliubov map}\label{Sec: time-ordered stuff}
Keeping in mind Equation \eqref{Eq: nonlinear dynamics}, the algebras of functionals defined in the previous section do not allow to account for the non linear contribution encoded by $V^{(1)}[\varphi]$. To encompass this information a key structure in the perturbative approach to quantum field theory is the $S$-matrix and the associated Bogoliubov map. In order to keep the length of this paper at bay, we shall not discuss the whole framework in detail, leaving an interested reader to \cite{BFDY15,Dutsch-Fredenhagen-AQFT-perturb-th-loop-exp,Rejzner-pAQFT} for a thorough analysis. We shall limit ourselves to introducing all the relevant structures necessary for our analysis to be self-consistent.

The starting point is the {\em time ordered product} which is here defined in terms of a {\tmem{time ordering map}} $\mathcal{T}$ acting on the space of multi-local functionals, which are tensor products of elements lying in $\mathcal{F}_{\tmop{loc}}(\mathbb{M})$ as in Definition \ref{Def:functionals}. More precisely $\mathcal{T}$ is constructed out of a family of multi-linear maps
\[ \mathcal{T}_n : \mathcal{F}_{\tmop{loc}}^{\otimes n}(\mathbb{M}) \rightarrow
   \mathcal{F}_{\mu c}(\mathbb{M}), \]
which satisfy the constraints $\mathcal{T}_0 = 1$ and $\mathcal{T}_1 = \tmop{id}$. The link between $\mathcal{T}$ and $\mathcal{T}_n$ is codified by the identity
\begin{equation}
	\mathcal{T} \left(
	\prod_{j = 1}^n F_j \right) =\mathcal{T}_n \left( \bigotimes_{j = 1}^n F_j
	\right).
\end{equation} 
In addition, one requires the maps $\mathcal{T}_n$ to be such that $\mathcal{T}$ is symmetric and to satisfy a {\tmem{causal
factorization property}}. This can be stated as follows: consider $\{F_i \}_{i=1,\ldots, n}, \{G_j \}_{j=1,\ldots, m} \subset
\mathcal{F}_{\tmop{loc}}(\mathbb{M})$ two arbitrary families of local functionals such that $F_i \gtrsim G_j$ for any $i, j$,w where the symbol $\gtrsim$ entails that $\tmop{supp} (F_i) \cap J^- (\tmop{supp} (G_j)) = \emptyset$, where the support of a functional is as per Equation \eqref{Eq: Support of a Functional}. It descends that
\begin{equation}\label{Eq:factoriz-time-ordering}
  \mathcal{T} \left( \bigotimes_i F_i  \bigotimes_j G_j \right) =\mathcal{T}
  \left( \bigotimes_i F_i \right) \star \mathcal{T} \left( \bigotimes_j G_j
  \right), 
\end{equation}
where $\star$ is defined in Equation \eqref{Eq:deform-quantisation}. We stress that, whenever the hypotheses at the heart of Equation \eqref{Eq:factoriz-time-ordering} are not met, one needs to devise a suitable extension criterion for $\mathcal{T}$, which requires in turn a renormalization procedure. In this work we abide by the Epstein-Glaser inductive procedure {\cite{Hollands-Wald-02,DDR20, Epstein-Glaser-73}}, although, as highlighted in {\cite{BPR23-Sine-Gordon-Massive, Bahns-Rejzner_Sine-Gordon}}, when one considers two-dimensional self-interacting, scalar field theories, renormalization is unnecessary if one restricts the attention to local functionals not containing derivatives of the field. 
As a consequence, for these classes of local functionals the causal factorization property suffices to fully determine the map $\mathcal{T}$.

As discussed in {\cite{BPR23-Sine-Gordon-Massive,Rejzner-pAQFT}}, if we work with the algebra $(\mathcal{F}_{\mu c}\llbracket\hbar\rrbracket,\star_{\hbar\omega},\ast)$, an explicit realization of the time-ordering map is completely characterized by the identity
\begin{equation}\label{Eq: time-ordered product}
	 \mathcal{T}^{\hbar \omega_F}  (F_1 \otimes \ldots \otimes F_n)\assign
	  F_1 \star_{\hbar \Delta_F} \ldots \star_{\hbar \Delta_F} F_n \assign
	\mathcal{M} \circ e^{\sum_{\ell < j}^n D^{\ell j}_{\hbar \Delta_F}}  (F_1
	\otimes \ldots \otimes F_n), 
\end{equation}
where $F_i\in\mathcal{F}_{\mu c}\llbracket\hbar\rrbracket$, for all $i=1,\dots,n$ while $\Delta_F\in\mathcal{D}'(\mathbb{M}\times\mathbb{M})$ denotes the {\tmem{Feynman parametrix}}, linked to $\omega\in\mathcal{D}'(\mathbb{M}\times\mathbb{M})$ and to the advanced propagator $\Delta^A\in\mathcal{D}'(\mathbb{M}\times\mathbb{M})$ associated to $P$ as per Remark \ref{Rem: Advanced and Retarded Propgators} via the defining relation
\begin{equation}\label{Eq: feynman propagator}
	\Delta_F = \omega + i \Delta^A .
\end{equation}
In addition in Equation \eqref{Eq: time-ordered product} $\mathcal{M}$ is defined as in Equation \eqref{Eq: tensor product} while 
\[ D^{\ell j}_{\hbar\omega_F}\assign \left\langle \hbar\omega_F, \frac{\delta^2}{\delta \varphi_{\ell}
   \delta \varphi_j} \right\rangle, \]
which is manifestly symmetric under exchange of $i$ and $\ell$. We have all the data necessary to define two key ingredients in the perturbative investigation of a self-interacting, scalar field theory whose dynamics is ruled by Equation \eqref{Eq: nonlinear dynamics}:
\begin{enumerate}
	\item the $S$-matrix as
\begin{equation} \label{Eq:S-matrix}
  S (\lambda V) \assign\exp_{\cdot_{\hbar\Delta_F}}\left(\frac{i}{\hbar}\lambda V\right)\assign \sum_{n \geqslant 0} \frac{1}{n!} \left( \frac{i
  \lambda}{\hbar} \right)^n \underbrace{V \star_{\Delta_F} \ldots \star_{\Delta_F} V}_{n} =
  \sum_{n \geqslant 0} \frac{1}{n!} \left( \frac{i \lambda}{\hbar} \right)^n
  \mathcal{T}^{\hbar \Delta_F} (V^{\otimes n}), 
\end{equation}
where here with $V\equiv V[\varphi]$ we denote the interaction potential lying in $\mathcal{F}_{\tmop{loc}}(\mathbb{M})$. \tmtextit{A priori}, the $S$-matrix is a formal power series in the coupling constant $\lambda$ and a Laurent series in $\hbar$. 
\item the interacting classical field, which is a perturbative solution of Equation~\eqref{Eq: nonlinear dynamics} with
vanishing initial conditions, written as a formal power series in $\lambda$ with
coefficients lying in $\mathcal{F}_{\tmop{loc}}[\mathbb{M}]$:
\begin{equation}\label{Eq:first-classical-obs}
	r_{\lambda V_g} (\varphi) (x) = \sum_{n \geqslant 0} \lambda^n \int_{t_1
		\leqslant \ldots \leqslant t_n \leqslant t} \mathd \mu_{x_1} \ldots \mathd
	\mu_{x_n} g (x_1) \ldots g (x_n) \{ V  (x_1), \{ V (x_2), \ldots \{ V (x_n),
	\varphi (x) \} \ldots \} \}, 
\end{equation}
where $d\mu_x$ is the metric induced measure, the curly brackets are the classical Poisson brackets as per Remark \ref{Rem: Advanced and Retarded Propgators}, while $g\in\mathcal{D}(\mathbb{M})$ is a cut-off function which is introduced to avoid infrared divergences. Observe that $r_{\lambda V_g}$ is also referred to as classical M\"oller map {\cite{Hawkins-Rejzner-star-product,Rejzner-pAQFT}}.
\end{enumerate}

\begin{example}
The formal expression in Equation \eqref{Eq:first-classical-obs} is analogous to the standard notion of
perturbative solution of a nonlinear partial differential equation adopted, {\tmem{e.g.}}, in {\cite{DDRZ21}}. As an example, consider
  \[ V_g (\varphi) = \frac{1}{4} \int_{\mathbb{M}} \mathd \mu_x\, g (x)
     \varphi^4 (x) . \]
The perturbative solution is defined in terms of the following power series in the coupling constant $\lambda$
  \[ \Phi \llbracket \lambda \rrbracket (\varphi) (x) = \sum_{n \geqslant 0}
     \lambda^n F_n (\varphi) (x), \]
  where
  \[ F_0 (\varphi)(x) = \Phi (\varphi) (x) = \varphi (x), \qquad \qquad F_n
     (\varphi) (x) \assign - \sum_{j_1 + j_2 + j_3 = n - 1} \mathLaplace^R
     \ast [F_{j_1} (\varphi) F_{j_2} (\varphi) F_{j_3} (\varphi)] (x), \]
    where $\ast$ here denotes the convolution while $\mathLaplace^R$ denotes the retarded fundamental solution associated to $P$ as in Remark \ref{Rem: Advanced and Retarded Propgators}.  By direct inspection, it can be seen that, order by order in $\lambda$, it holds that
  \[ \Phi \llbracket \lambda \rrbracket (\varphi) (x) = r_{\lambda V_g}
     (\varphi) (x) . \]
\end{example}

The information carried by the S-matrix and by the classical M\o ller map can be brought together extending Equation \eqref{Eq:first-classical-obs} to the whole algebra of observables of the theory. More precisely, given $F \in \mathcal{F}^{\otimes m}_{\tmop{loc}}(\mathbb{M})$ and, adopting the notation $Y = (y_1, \ldots, y_m)\in\mathbb{M}^m$, the associated interacting classical observable can be written as
\begin{equation}
  r_{\lambda V_g} (F) (Y) = \sum_{n \geqslant 0} \lambda^n \int_{t_1 \leqslant
  \ldots \leqslant t_n \leqslant t} \mathd \mu_{x_1} \ldots \mathd \mu_{x_n} g
  (x_1) \ldots g (x_n) \{ V  (x_1), \{ V (x_2), \ldots \{ V (x_n), F (Y) \}
  \ldots \} \} . \label{Eq:first-classical-generic-obs}
\end{equation}
Having established the algebraic formulation of a classical interacting field theory, an analogous procedure carries over to the quantum scenario, with the notable difference that one needs to work with $(\mathcal{F}_{\mu c}\llbracket\hbar\rrbracket,\star_{\hbar H},\ast)$. To this end we introduce the {\tmem{Bogoliubov map}} $R_{\lambda V}$ associated with the interaction $V\in\mathcal{F}_{loc}(\mathbb{M})$, which maps any observable of the free field theory whose dynamics is ruled by $P$ into its interacting quantum counterpart \cite{Rejzner-pAQFT}. The action on any $F \in
\mathcal{F}_{\tmop{loc}}^{\otimes m}(\mathbb{M})$ reads as 
\begin{equation}
  R_{\lambda V} (F) \assign - i \hbar \frac{d}{d \alpha} S (\lambda
  V)^{\star_{\hbar\omega} - 1} \star_{\hbar\omega} S (\lambda V + \alpha F)
  |_{\alpha = 0} = S (\lambda V)^{\star_{\hbar\omega} - 1} \star_{\hbar\omega}  (S (\lambda V) \star_{\omega_F} F), \label{Eq:Bogoliubov map}
\end{equation}
where $S (\lambda V)^{\star_{\hbar \omega} - 1}$ denotes the inverse (in the
sense of formal power series) of $S (\lambda V)$ with respect to the product
$\star_{\hbar\omega}$. 

\begin{remark}\label{Rem:inverse-S-matrix-expansion}
For later convenience we need to give an explicit expression of $S (\lambda V)^{\star_{\hbar \omega} - 1}$. This requires the anti-Feynman parametrix 
\begin{equation}\label{Eq: anti-feynman}
 	\Delta_{\tmop{AF}} = \omega - i \Delta^R,
\end{equation}
where $\Delta^R$ is the retarded fundamental solution associated to $P$ as per Remark \ref{Rem: Advanced and Retarded Propgators}. Consequently it holds that
  \[ S (\lambda V)^{\star_{\hbar \omega} - 1} \assign \sum_{n \geqslant 0}
     \frac{1}{n!} \left( -\frac{i \lambda}{\hbar} \right)^n \underbrace{V \star_{\hbar
     \Delta_{\tmop{AF}}} \ldots \star_{\hbar \Delta_{\tmop{AF}}} V}_{n} = \sum_{n
     \geqslant 0} \frac{1}{n!} \left(-\frac{i \lambda}{\hbar} \right)^n
     \mathcal{T}^{\hbar \Delta_{\tmop{AF}}} (V^{\otimes n}).\]
\end{remark}
As a consequence, Equation~\eqref{Eq:Bogoliubov map} can be written as a
formal power series in the coupling constant $\lambda$:
\[ R_{\lambda V} (F) = \sum_{n \geqslant 0} \frac{\lambda^n}{n!} R_{n, m}
   (V^{\otimes n}, F), \]
where $R_{n,m} (V^{\otimes n}, F)$ are the so-called {\tmem{retarded products}}
\[ R_{n, m} (V^{\otimes n}, F) = \left( \frac{i}{\hbar} \right)^n \sum_{\ell =
   0}^n \binom{n}{\ell} (- 1)^{\ell} \mathcal{T}_{\ell}^{\hbar
   \Delta_{\tmop{AF}}} (V \otimes \ldots \otimes V) \cdot_{\hbar \omega}
   \mathcal{T}_{n - \ell, m}^{\hbar \Delta_F} (V \otimes \ldots \otimes V
   \otimes F). \]
The notation $\mathcal{T}_{n - \ell, m}^{\hbar\Delta_F} (V \otimes
\ldots \otimes V \otimes F)$ keeps track of the fact that its argument
contain $n - \ell$ copies of the interaction $V$ and that $F$ is a
multi-local functional of order $m \in \mathbb{N}$. Hence, on account of Equation \eqref{Eq: time-ordered product}
\[ \mathcal{T}_{n - \ell, m}^{\hbar\Delta_F} (V \otimes \ldots \otimes V
   \otimes F) = V \cdot_{\hbar\Delta_F} \ldots \cdot_{\hbar\Delta_F} V
   \cdot_{\hbar\Delta_F} F. \]
\begin{remark}\label{Rem:supp-retarded-products}
An important feature of the retarded products, which underpins their name, concerns their support properties. More precisely, $R_{n,m}(V^{\otimes n}, F)$ vanishes if al least one of the first $n$ arguments is not supported in the past light cone of one of the last $m$ ones {\cite{Dutsch-Fredenhagen-AQFT-perturb-th-loop-exp,Epstein-Glaser-73}}.
\end{remark}

\subsection{Microlocal approach to SPDEs}\label{Sec:microlocal-approach-to-spdes}

In this section we give a succinct overview of a recent approach to the construction at a perturbative level of both the expectation value of  solutions and the correlation functions of a class of nonlinear stochastic partial differential equations (SPDEs), first introduced in \cite{DDRZ21} for scalar theories and later extended in \cite{BDR23} to the analysis of the nonlinear stochastic Schr\"odinger equation and in \cite{BCDR23} to that of spinors. This framework can be seen as a transliteration to the analysis of SPDEs of the algebraic approach to quantum field theory outlined in Section \ref{Sec:interacting-AQFT} and, as such, it has the main advantage of allowing to encode all renormalization ambiguities intrinsically without resorting to any specific $\epsilon$-regularization scheme. The reason lies in the possibility of adapting to the case in hand the microlocal approach to Epstein-Glaser renormalization, see {\it e.g.}, \cite{DDR20, Hollands-Wald-02, Keller:2009nv}. 

It is worth emphasizing that, in \cite{BCDR23, BDR23, DDRZ21}, the focus has always been on the analysis of elliptic or parabolic SPDEs and therefore we feel worth revisiting the content of these works when the linear part of the dynamics is ruled by an hyperbolic partial differential operator. For this reason, in this section, we start by considering a simple toy model, namely a linear, scalar, SPDE on $\mathbb{R}^d$ endowed with the flat Minkowski metric of signature $(-,\underbrace{+,\dots,+}_{d-1})$:
\begin{equation}\label{Eq: linear stochastic equation}
	(\Box - m^2)  \hat{\psi} = \chi \hat{\xi},
\end{equation}
with vanishing initial condition. Here $\Box$ is the d'Alembert wave operator, while $m^2> 0$ is a fixed parameter, which can be interpreted in concrete models as playing the r\^{o}le of a mass term. Furthermore $\hat{\xi}$ denotes a space-time white noise as per Equation \eqref{Eq: covariance white noise_intro}. The last ingredient $\chi : \mathbb{M} \rightarrow \mathbb{R}$ is the smooth cutoff function
\[ \chi (t, x) = \tilde{\chi} (t) 1 (x), \] 
where $\tilde{\chi} : \mathbb{R} \rightarrow \mathbb{R}$ is a positive smooth function such that there exists $T \in \mathbb{R}$ for which
\begin{equation}\label{Eq: time cutoff}
	\tilde{\chi} (t) = \left\{ \begin{array}{l}
		0 \qquad \tmop{if} \quad t < T,\\
		1 \qquad \tmop{if} \quad t \geqslant T + 1.
	\end{array} \right.
\end{equation}
Observe that $\Box+m^2$ plays the r\^{o}le of the operator $P$ in Remark \ref{Rem: Advanced and Retarded Propgators} and, in particular, we can associate to it unique advanced and retarded fundamental solutions, still denoted by $\Delta^A$ and $\Delta^R$ respectively. The latter allows to solve Equation \eqref{Eq: linear stochastic equation} forward in time as, with vanishing initial condition for simplicity,
\[ \hat{\phi} = \Delta^R  (\chi\hat{\xi}). \]
We can infer that the solution $\hat{\phi}$ is a Gaussian random distribution with vanishing mean, while, given $f, f^\prime \in \mathcal{D} (\mathbb{R}^d)$, the covariance reads
\begin{align}\label{Eq:first-app-Q}
\nonumber
    Q (f, f^\prime) \assign & \int_{\mathbb{R}^d\times\mathbb{R}^d} \mathd \mu_z \mathd \mu_{z^\prime}
    \mathbb{E} [\hat{\phi} (z) \hat{\phi} (z^\prime)] f (z) f (z^\prime)\\
    \nonumber
    = & \int_{\mathbb{R}^d\times\mathbb{R}^d} \mathd \mu_z \mathd \mu_{z^\prime} 
    \int_{\mathbb{R}^d} \mathd z_1 \chi^2 (t_1) \Delta^R  (z - z_1) \Delta^A  (z_1 -
    z^\prime) f (z) f (z^\prime)\\
     = & \int_{\mathbb{R}^d\times\mathbb{R}^d} \mathd \mu_z \mathd \mu_{z^\prime} 
    \int_{\mathbb{R}^d} \mathd z_1 \chi^2 (t_1) \Delta^R  (z - z_1) \Delta^R  (z^\prime
    - z_1) f (z) f (z^\prime),
\end{align}
where, in the last equality, we used the following structural properties of the advanced and retarded fundamental solutions associated to a second order, hyperbolic, partial differential operator on a globally hyperbolic manifold $\mathbb{M}$, $\langle h, \Delta^R h^\prime\rangle=\langle\Delta^A h, h^\prime\rangle$ for all $f,f^\prime\in\mathcal{D}(\mathbb{M})$, see \cite{BGP}.

Observe that Equation \eqref{Eq:first-app-Q} entails that the covariance $Q\in\mathcal{D}^\prime(\mathbb{R}^d\times\mathbb{R}^d)$ can be written as 
\begin{equation} \label{Eq:def-of-Q}
	Q \assign \Delta^R \circ_{\chi} \Delta^A,
\end{equation}
where the symbol $\circ$ denotes the composition of distributions, while the subscript $\chi$ keeps track of the cut-off function. A direct application of \cite[Thm. 8.2.14]{Hormander-83} entails that 
\begin{equation}\label{Eq: WF Q}
	\tmop{WF} (Q) \subseteq \tmop{WF} (\Delta^R)=\{(z,k_z,z^\prime,k_{z^\prime})\in T^*\left(\mathbb{R}^d\times\mathbb{R}^d\right)\setminus\{0\}\;|\;(z,k_z)\sim_+(z^\prime,-k_{z^\prime})\}\cup\,\textrm{WF}(\delta_2), 
\end{equation}
where $\sim_+$ entails that $z$ and $z^\prime$ are connected by a lightlike geodesic $\gamma$, $-k_y$ is the parallel transport of $k_x$ along $\gamma$ and $z\in J^+(z^\prime)$. In addition $\delta_2$ is the Dirac delta along the diagonal of $\mathbb{R}^d\times\mathbb{R}^d$ and, for the sake of completeness, we recall that
$$\mathrm{WF}(\delta_2)=\{(z,k_z,z,-k_z)\in T^*(\mathbb{R}^d\times\mathbb{R}^d)\setminus\{0\}\}.$$
In addition the covariance $Q$ enjoys the property of being positive, a feature which will be used extensively in the following sections.
\begin{lemma}\label{Lem:positivity-of-Q}
	The bi-distribution $Q
	\in \mathcal{D}' (\mathbb{R}^d\times\mathbb{R}^d)$ as in Equation \eqref{Eq:def-of-Q} is positive, namely, for all $f \in
	\mathcal{D} (\mathbb{R}^d)$,
	\[ Q (f, f) \geqslant 0. \]
\end{lemma}

\begin{proof}
	Working at the level of integral kernels, it holds that
	\begin{align*}
		Q (f, f)  = & \int_{\mathbb{R}^{3d}} \mathd \mu_{z_1} \mathd \mu_{z_2}
		\mathd \mu_{z_3} \chi^2 (z_2) \Delta^R  (z_1 - z_2) \Delta^A  (z_2 -
		z_3) f (z_1) f (z_3)\\
		= & \int_{\mathbb{R}^d} \mathd \mu_{z_2} \chi^2 (z_2) \left(
		\int_{\mathbb{R}^d} \mathd \mu_{z_1} \Delta^R (z_1 - z_2) f (z_1) \right)
		\left( \int_{\mathbb{R}^d} \mathd \mu_{z_3} \Delta^A (z_2 - z_3) f (z_3)
		\right)\\
		= & \int_{\mathbb{R}^d} \mathd \mu_{z_2} \chi^2 (z_2) \left(
		\int_{\mathbb{R}^d} \mathd \mu_{z_1} \Delta^A (z_2 - z_1) f (z_1) \right)
		\left( \int_{\mathbb{R}^d} \mathd \mu_{z_3} \Delta^A (z_2 - z_3) f (z_3)
		\right)\\
		= & \int_{\mathbb{R}^d} \mathd \mu_{z_2} \chi^2 (z_2)  (\Delta^A f)
		(z_2)  (\Delta^A f) (z_2)\\
		= & (\chi (\Delta^A f), \chi (\Delta^A f))_{L^2 (\mathbb{R}^d)}\geq 0,
	\end{align*}
	where, in the third line, we used the identity $\Delta^R (z_1-z_2) = \Delta^A  (z_2- z_1)$.
\end{proof}


Equation \eqref{Eq:first-app-Q} gives a complete control on the solution theory of Equation \eqref{Eq: linear stochastic equation} and, in the following, we reformulate its content using the rationale at the heart of the approach to SPDEs inspired by algebraic quantum field theory, adapting to the case in hand the procedure of~{\cite{DDRZ21}}. More precisely we define the {\tmem{functional-valued distribution}} $F :
\mathcal{E} (\mathbb{R}^d) \times \mathcal{D} (\mathbb{R}^d)\times\mathcal{D}(\mathbb{R}^d)
\rightarrow \mathbb{C}$ on the space of smooth configurations 
\begin{equation}\label{Eq: Example of Functional}
 \mathcal{E} (\mathbb{R}^d) \times \mathcal{D} (\mathbb{R})^d\times\mathcal{D} (\mathbb{R})^d \ni
   (\varphi, f, f^\prime) \mapsto F [\varphi] (f, f^\prime) = \int_{\mathbb{R}^d\times\mathbb{R}^d} \mathd
   \mu_z \mathd \mu_{z^\prime} \varphi (z) \varphi (z^\prime) f (z) f (z^\prime) \in
   \mathbb{C}.
\end{equation}
Observe that, on the one hand the functional $F$, despite being regular, can also be interpreted as an element lying in the space of microcausal functional as per Definition \ref{Def:functionals} which is endowed with an algebra structure in terms of the pointwise product as per Equation \eqref{Eq: pointwise product}. On the other hand, this rationale allows us to follow Equation \eqref{Eq:deform-quantisation} as well as Remark \ref{Rem:deformation-map} to introduce a deformed product encompassing the information carried by the underlying white noise $\hat{\xi}$ by means of the map 
\begin{equation}\label{Eq:GammaQ}
  \Gamma_Q = e^{\frac{1}{2} \mathcal{D}_Q}, \qquad
  \mathcal{D}_Q = \left\langle Q, \frac{\delta^2}{\delta \varphi^2}
  \right\rangle = \int_{\mathbb{R}^d\times\mathbb{R}^d} \mathd \mu_{z} \mathd \mu_{z^\prime} Q (z,
  z^\prime) \frac{\delta^2}{\delta \varphi (z) \delta \varphi (z^\prime)} .
\end{equation}
Applying Equation \eqref{Eq:GammaQ} to Equation \eqref{Eq: Example of Functional} yields
\begin{equation}\label{Eq: almost 2-point}
	(\Gamma_Q F) [\varphi] (f, f^\prime) = \int_{\mathbb{R}^d\times\mathbb{R}^d} \mathd \mu_z \mathd
	\mu_{z^\prime} \varphi (z) \varphi (z^\prime) f (z) f (z^\prime) + \int_{\mathbb{R}^d\times\mathbb{R}^d}
	\mathd \mu_z \mathd \mu_{z^\prime} Q (z, z^\prime) f (z) f (z^\prime), 
\end{equation}
where $Q (z, z^\prime)$ denotes the integral kernel of $Q\in\mathcal{D}'(\mathbb{R}^d\times\mathbb{R}^d)$ introduced in Equation~\eqref{Eq:first-app-Q}. As one can infer by direct inspection, the action of $\Gamma_Q$ is that of encoding the information of the correlation function on the underlying stochastic process. Yet, a last bit of information is missing, namely the expectation value $\mathbb{E}$ corresponds to evaluating Equation \eqref{Eq: almost 2-point} at the configuration $\varphi = 0$. In other words
\[ \tmop{ev}_0 (\Gamma_Q F) (f, f^\prime) \assign (\Gamma_Q F) [0] (f, f^\prime) =
   \int_{\mathbb{R}^d\times\mathbb{R}^d} \mathd \mu_z \mathd \mu_{z^\prime} Q (z, z^\prime) f (z) f (z^\prime)
   = Q (f, f^\prime) . \]
Observe that the procedure, just highlighted, strongly enjoys from $F$ being regular and, if one tries to apply it to a local, non-linear polynomial functional, such as
\[ \Phi^2 [\varphi] (f) = \int_{\mathbb{R}^d} \mathd \mu_z \varphi^2 (z) f (z),
\]
it does not carry over slavishly. A formal application of Equation \eqref{Eq:GammaQ} yields
\[ (\Gamma_Q \Phi^2) [\varphi] (f) = \Phi^2 [\varphi] (f) + \int_{\mathbb{R}^d}
   \mathd \mu_z Q (z, z) f (z) , \]
which is a priori ill-defined since, on account of Equation~\eqref{Eq:first-app-Q}, this would be equivalent to testing $Q$
against a distribution of the form $f (z) \delta (z - z^\prime)$. This yields
\begin{equation}\label{Eq: ill-defined Q}
	\int_{\mathbb{R}^d} \mathd \mu_z Q (z, z) f (z) = \int_{\mathbb{R}^D} \mathd
	\mu_z  \int_{\mathbb{R}^d} \mathd \mu_{z^\prime} \chi^2 (t^\prime)  [\Delta^R (z - z^\prime)]^2
	f (z) = (\chi \Delta^R)^2  (f \otimes 1).
\end{equation}
Yet the last expression in Equation \eqref{Eq: ill-defined Q} encompasses the square of $\chi\Delta^R$ which is not a well-defined distribution due to the singular structure of $\mathLaplace^R$. Bypassing this hurdle requires a renormalization procedure and this feature cannot be ascribed to the hyperbolic nature of the problem in hand, since it occurs also when analyzing elliptic or parabolic models, see {\it e.g.} \cite{BCDR23,BDR23,DDRZ21}. 

Yet it is important to stress that, in lower dimensional scenarios, typically when $d\leq 2$, the singularities encoded in the fundamental solutions of the Klein-Gordon operator are rather mild, being of logarithmic type, and, in this case, a renormalization procedure is not necessary. This is exactly the scenario that we shall consider when working with the sine-Gordon model.

\subsubsection{On the covariance $Q$ in two-dimensional Minkowski spacetime}\label{Sec: about Q}

In the following we shall investigate some notable structural properties of the bi-distribution $Q$ in the case when the spacetime dimension is $d=2$. It is also convenient to consider in this scenario the massless d'Alembert wave operator $\Box$ and, repeating in this case the analysis of Section \ref{Sec:microlocal-approach-to-spdes}, we denote the counterpart of $Q$ as $Q_0 = \Delta^R_0 \circ_{\chi} \Delta^A_0$, $\Delta^{R/A}_0$ being the advanced and retarded fundamental solutions associated to $\Box$. Denoting with $z=(t,x)$ an arbitrary point of $\mathbb{R}^2$, the integral kernel of $\Delta^R_0$ reads \cite{Bahns-Rejzner_Sine-Gordon}
\begin{equation}\label{Eq: Prop in 2D}
\Delta^R_0 (z) = - \frac{1}{2} \theta (t - |x|) = - \frac{1}{2} \theta (t -
   x) \theta (t + x),
\end{equation}
where $\theta$ denotes the Heaviside step function. The counterpart in the
massive case reads instead {\cite{BPR23-Sine-Gordon-Massive}}
\[ \Delta^R (z) = \Delta^R_0 (z) + \left( 1 - I_0  \left( m \sqrt{z^2} \right)
   \right) \Delta^R_0 (z) = \left( 2 - I_0  \left( m \sqrt{z^2} \right)
   \right) \Delta^R_0 (z), \]
where $I_0$ is the modified Bessel function of first kind, while $z^2 = -
t^2 + x^2$ is the Minkowskian square. Since $\textrm{supp}(\Delta^R_0(z))\subseteq J^+(0)$, we consider only the region $-
t^2 + x^2 \leqslant 0$ and, thereon, $I_0  \left( m \sqrt{z^2}
\right) = I_0  \left(i m \sqrt{|z^2 |} \right)$. Since $I_0  \left( m\sqrt{z^2 } \right)$ is a real, smooth, damped
oscillating function in $z$ with $\sup_z \left| I_0  \left( m \sqrt{z^2}
\right) \right| = 1$, it holds that
\begin{align}\label{Eq: conto Q}
     Q (z, z^\prime)  = & \int_{\mathbb{R}^2} \mathd \mu_{\tilde{z}} \chi^2
     (\tilde{z}) \Delta^R (z - \tilde{z}) \Delta^R (z^\prime - \tilde{z})\nonumber\\
      = & \int_{\mathbb{R}^2} \mathd \mu_{\tilde{z}} \chi^2 (\tilde{z})
     \Delta_0^R (z - \tilde{z}) \Delta_0^R (z^\prime - \tilde{z}) \left( 2 - I_0
     \left( im \sqrt{| (z - \tilde{z})^2 |} \right) \right) \left( 2 - I_0
     \left( im \sqrt{| (z^\prime - \tilde{z})^2 |} \right) \right) .
   \end{align}
Since $\chi^2 (\tilde{z})$ has past-compact support according to Equation \eqref{Eq: time cutoff} and since  Equation \eqref{Eq: Prop in 2D} entails that
\[ \Delta_0^R (z - \tilde{z}) \Delta_0^R (z^\prime - \tilde{z}) = 1_{\{J^{+} (z)
   \cap J^- (z^\prime)\}} (\tilde{z}), \]
where $1_{\{J^{- 1} (z) \cap J^- (z_1)\}} (z')$ denotes the characteristic
function on the subset $J^{- 1} (z) \cap J^- (z^\prime) \subset \mathbb{R}^2$, we can conclude that the integral in Equation \eqref{Eq: conto Q} is convergent. In addition, $Q (z, z^\prime)$ is a continuous and positive function, which is translation invariant along the space direction. These data can be summarized in the following lemma

\begin{lemma}\label{Lem: On the regularity of Q}
 On the two-dimensional Minkowski spacetime the covariance $Q \in \mathcal{D}' (\mathbb{R}^2\times\mathbb{R}^2)$ introduced in Equation \eqref{Eq:def-of-Q} is positive and in turn it is generated by a continuous and positive function $Q(z,z^\prime)$. Consequently, the coincidence limit $Q(z,z)$ is a smooth positive function, namely $Q (z, z) \in \mathcal{E} (\mathbb{R}^2)$ and $Q (z, z) \geqslant 0$ for all $z \in \mathbb{R}^2$.
\end{lemma}

\subsection{The algebra of functionals for the sine-Gordon model}\label{Sec: Functionals and sine-Gordon}
In Section \ref{Sec:interacting-AQFT} we have emphasized how the analysis of an interacting field theory requires in the language of algebraic quantum field theory the identification of a distinguished algebra of functionals and in Section \ref{Sec:microlocal-approach-to-spdes} we have emphasized how the same structures come into play when studying a nonlinear stochastic partial differential equation along the lines of \cite{DDRZ21}. Hence, in the following we introduce the set of functionals which is necessary to investigate the sine-Gordon model as per Section \ref{Sec: Notation}.

The rationale that we follow is strongly inspired by to the one discussed in {\cite[Sec. 1.4]{BPR23-Sine-Gordon-Massive}} and, for this very reason, we shall refer to this reference for the proof of many statements, limiting ourselves to highlighting the main differences with our setting.

\begin{definition}\label{Def: enlarged class functionals}
  We denote with $\mathcal{F}^V(\mathbb{R}^2)$ the vector space of functionals generated by elements
  of the form
  \[ F_{a, n, m} (\zeta) = \int_{\mathbb{M}^{n+m}} \mathd \mu_{X, Y} e^{i
     \sum_{j = 1}^n a_j \varphi (x_j)} e^{- i \sum_{1 \leqslant \ell < j
     \leqslant n} a_{\ell} a_j (Q (x_{\ell}, x_j) + \hbar H_0 (x_{\ell},
     x_j))} \zeta (X, Y) \prod_{k = 1}^m \varphi (y_k), \]
  where $\mathd \mu_{X, Y} = \mathd \mu_{x_1} \ldots \mathd \mu_{y_m}$, $X =
  (x_1, \ldots, x_n)$, $Y = (y_1, \ldots, y_m)$, $a \subset (- \eta, \eta)^n$
  is a multi-index with $\eta < 4 \pi \hbar^{- 1}$. Furthermore $\zeta$ is a distribution generated by
  bounded and compactly supported function such that
  \begin{equation}
    \tmop{WF} (\zeta) \subset W_{m + n} \assign \left\{ (x_1, \ldots, x_{n +
    m} ; k_1, \ldots, k_{n + m}) \in T^{\ast} \mathbb{R}^{d(n + m)} \setminus
    \emptyset\; \barsuchthat\; \sum_{\ell = 1}^{n + m} k_{\ell} = 0 \right\} .
    \label{Eq:WF-condition}
  \end{equation}
\end{definition}
Observe that, in
Equation~\eqref{Eq:WF-condition}, the condition
\[ \sum_{\ell = 1}^{n + m} k_{\ell} = 0, \]
surmises that all covectors have been parallel transported at a common base point.

\begin{remark}
  \label{Rem:clarifying-WF} 
  For future convenience, it is worth highlighting that 
  a comparison between Equations \eqref{Eq:WF-condition} and Equation \eqref{Eq: WF Q} entails that
  \begin{equation}\label{Eq: estimate WF retarded}
  	\textrm{WF}(Q)\subseteq\tmop{WF} (\mathLaplace^R) \subset W_2 .
  \end{equation}
\end{remark}

\begin{proposition}\label{Prop: closedness F^V}
  Denoting with $\mathcal{F}_{\mu c}(\mathbb{R}^2)$ the set of microcausal functionals on $\mathbb{R}^2$ as per Definition \ref{Def:functionals} and with $\mathcal{F}^V(\mathbb{R}^2)$ that introduced in Definition \ref{Def: enlarged class functionals}, it holds that
  \begin{enumerate}
    \item \tmverbatim{}$\mathcal{F}^V \subset \mathcal{F}_{\mu c}$;
    
    \item $\mathcal{F}^V(\mathbb{R}^2)$ identifies a $*$-subalgebra of $(\mathcal{F}_{\mu c}(\mathbb{R}^2), \star_{\hbar K}, \ast)$
    if $K$ is one among the following bi-distributions $Q + \hbar
    \omega, \tmop{Re} (Q + \hbar \omega), Q + \hbar \Delta_F, \tmop{Re} (Q +
    \hbar \Delta_F)$ where $\omega$ is the two-point function of the Poincar\'e invariant ground state on the two-dimensional Minkowski spacetime and $\Delta_F$ the associated Feynmann propagator.
  \end{enumerate}
In particular we set
\begin{equation}\label{Eq: algebra 2}
	\mathcal{A}_{\hbar\omega + Q}^V(\mathbb{R}^2) \assign (\mathcal{F}^V(\mathbb{R}^2), \cdot_{\hbar \omega
		+ Q}, \ast) ,
\end{equation}
\end{proposition}
\begin{proof}
  The proof follows the same lines of {\cite[Prop. 1.2, 1.3]{BPR23-Sine-Gordon-Massive}}, the only difference being that the generators in Definition \ref{Def: enlarged class functionals} encompass also $Q\in\mathcal{D}^\prime(\mathbb{R}^2\times\mathbb{R}^2)$ defined in Equation \eqref{Eq:first-app-Q}. Nonetheless, the line of reasoning extends slavishly to our scenario on account of Equation \eqref{Eq: estimate WF retarded}.
\end{proof}

\subsection{The model and the strategy}

\label{Sec:Strategy}

We reckon that it is desirable to outline succinctly the goal of our analysis as well as the strategy employed to reach it, before delving in all the technical details. Henceforth we shall fix the background $\mathbb{M}$ to be the two-dimensional Minkowski spacetime $(\mathbb{R}^2,\eta)$ and, as already mentioned in Section \ref{Sec: Notation}, on top of it we consider the {\em stochastic sine-Gordon equation} with coupling constant $\lambda$
\begin{equation}\label{Eq: sine-Gordon equation}
	(\Box + m^2) \hat{\psi} + \lambda g a \sin(a\hat{\psi}) = \chi \hat{\xi}, 
\end{equation}
where $\Box=\partial^2_t-\partial^2_x$ is the d'Alembert wave operator and where we assume vanishing initial condition. In addition, $a\in \mathbb{R}$, while $g \in \mathcal{D} (\mathbb{R}^2)$. Finally the spacetime with noise $\hat{\xi}$ and $\chi\in\mathcal{E}(\mathbb{R}^2)$ enjoy the properties discussed in Section~\ref{Sec:microlocal-approach-to-spdes}.


Following and actually combining the rationale of Section~\ref{Sec:microlocal-approach-to-spdes} and in particular of Section \ref{Sec:interacting-AQFT}, we start from Equation \eqref{Eq: sine-Gordon equation} and we construct the classical interacting field
\begin{equation}\label{Eq: r-map}
 r_{\lambda V_g} (\varphi) (x) = \sum_{n \geqslant 0} \lambda^n \int_{t_1
   \leqslant \ldots \leqslant t_n \leqslant t} \mathd \mu_{x_1} \ldots \mathd
   \mu_{x_n} g (x_1) \ldots g (x_n) \{ V  (x_1), \{ V (x_2), \ldots \{ V
   (x_n), \varphi (x) \} \ldots \} \}.
\end{equation}
This is an application of Equation \eqref{Eq:first-classical-obs} in which the r\^{o}le of $V(x)$ is played by $\cos(a\varphi(x))$ and, heuristically, this amounts to (perturbatively) constructing a solution of the classical PDE $(\Box+m^2)\varphi+\lambda g a\sin(a\varphi)=0$. Subsequently, as seen in Section~\ref{Sec:microlocal-approach-to-spdes}, we can start from Equation \eqref{Eq: r-map} encoding the stochastic properties of the noise by acting with the map $\Gamma_Q$ defined in
Equation~\eqref{Eq:GammaQ}, namely
\begin{equation}\label{Eq: expectation classical Moller}
	 \Gamma_Q [r_{\lambda V_g} (\varphi)] (x) = \sum_{n \geqslant 0} \lambda^n
	\Gamma_Q \int_{t_1 \leqslant \ldots \leqslant t_n \leqslant t} \mathd \mu_X
	g (x_1) \ldots g (x_n) \{ V  (x_1), \{ V (x_2), \ldots \{ V (x_n), \varphi
	(x) \} \ldots \} \} .
\end{equation}
In agreement with the formulation of \cite{DDRZ21}, evaluating $\Gamma_Q [r_{\lambda V_g} (\varphi)] (x)$ at $\varphi = 0$ yields the expectation value of the solution of Equation \eqref{Eq: sine-Gordon equation} as a formal power series in
$\lambda$. Yet, a natural question pertains the convergence of the formal power series in Equation \eqref{Eq: expectation classical Moller}. To answer this question, we adopt an unconventional strategy which consists of building $\Gamma_Q [r_{\lambda V_g} (\varphi)]$ as the classical limit of $\Gamma_Q [R_{\lambda V_g} (\varphi)]$, where $R_{\lambda V_g} (\varphi)$ is the interacting quantum field defined in Equation \eqref{Eq:Bogoliubov map} via the Bogoliubov map for the specific choice of the observable $F\equiv \Phi$. In other words, we shall prove that
\begin{equation}\label{Eq: expected relation}
	\Gamma_Q [r_{\lambda V_g} (\varphi)] = \lim_{\hbar \rightarrow 0^+}
	\Gamma_Q [R_{\lambda V_g} (\varphi)] .
\end{equation}
The net advantage of this unusual strategy is that, adapting to the case in hand techniques introduced in \cite{BPR23-Sine-Gordon-Massive,Frolich-CMP-76-2-d-stat-mech}, we are able to prove absolute convergence of the formal power series defining $\Gamma_Q [R_{\lambda V_g} (\varphi)]$ for any field configuration $\varphi\in\mathcal{E}(\mathbb{M})$. The subsequent investigation of the classical limit procedure will be discussed in Section~\ref{Sec:classical-limit}. As a consequence, the non-perturbative expectation value of the solution can be written as
\[ \Gamma_Q [r_{\lambda V_g} (\varphi)]_{\varphi = 0} = \lim_{\hbar
   \rightarrow 0^+} \Gamma_Q [R_{\lambda V_g} (\varphi)]_{\varphi = 0} \,. \]
The same line of reasoning shall be shown to be applicable also to the analysis of the $n$-point correlation functions of the
solution of Equation \eqref{Eq: sine-Gordon equation}, see Section~\ref{Sec:convergence-corr-functions}.

\section{Interplay between quantum and stochastic}\label{Sec:interplay}
As outlined in Section~\ref{Sec:Strategy}, we shall investigate the construction of the expectation value and of the correlations of the solutions of Equation \eqref{Eq: sine-Gordon equation} following a two-steps procedure.

\noindent In the following we address the analysis of the first part. This consists of an unconventional combination of the frameworks outlined in Section \ref{Sec: time-ordered stuff} and \ref{Sec:microlocal-approach-to-spdes}. More precisely we start by considering $F \in \mathcal{F}_{\tmop{loc}}(\mathbb{R}^2)$ as per Definition \ref{Def: enlarged class functionals} and we associate to it its quantum counterpart via the Bogoliubov map as per Equation \eqref{Eq:Bogoliubov map}:
\[ R_{\lambda V} (F) = S (\lambda V)^{\star_{\hbar \omega} - 1} \star_{\hbar
   \omega}  (S (\lambda V) \cdot_{\hbar \Delta_F} F), \]
where here we are implicitly thinking about $V=\cos(a\varphi)$ and where $\omega$ and $\Delta_F$ are in the previous sections. On account of Remark~\ref{Rem:deformation-map}, the last expression can be written as
\[ R_{\lambda V} (F) = \Gamma_{\hbar \omega}  \Bigl[\Gamma_{\hbar \omega}^{- 1} ((S
   (\lambda V))^{\star_{\hbar \omega} - 1}) \Gamma_{\hbar \omega}^{- 1}
   (\Gamma_{\hbar \Delta_F} [\Gamma_{\hbar \Delta_F}^{- 1} (S (\lambda V))
   \Gamma_{\hbar \Delta_F}^{- 1} (F)])\Bigr] . \]
Subsequently, by applying the rationale outlined in Section \ref{Sec:microlocal-approach-to-spdes}, we can codify the information on the correlations of the Gaussian white noise in Equation \eqref{Eq: sine-Gordon equation} by applying the map $\Gamma_Q$ as per Equation \eqref{Eq:GammaQ}:
\begin{align*}
     \Gamma_Q [R_{\lambda V} (F)]  = & \Gamma_Q \Gamma_{\hbar \omega}
     \Bigl[\Gamma_{\hbar \omega}^{- 1} ((S (\lambda V))^{\star_{\hbar \omega} - 1})
     \Gamma_{\hbar \omega}^{- 1} (\Gamma_{\hbar \Delta_F} [\Gamma_{\hbar
     \Delta_F}^{- 1} (S (\lambda V)) \Gamma_{\hbar \Delta_F}^{- 1} (F)])\Bigr]\\
      = & \Gamma_{Q + \hbar \omega} \Bigl[\Gamma_{\hbar \omega}^{- 1}
     (\Gamma_Q^{- 1} \Gamma_Q (S (\lambda V))^{\star_{\hbar \omega} - 1})
     \Gamma_{\hbar \omega}^{- 1} \Gamma_Q^{- 1} \Gamma_Q (\Gamma_{\hbar
     \Delta_F} [\Gamma_{\hbar \Delta_F}^{- 1} (S (\lambda V)) \Gamma_{\hbar
     \Delta_F}^{- 1} (F)])\Bigr]\\
      = & \Gamma_{Q + \hbar \omega} \Bigl[\Gamma_{Q + \hbar \omega}^{- 1}
     (\Gamma_Q (S (\lambda V))^{\star_{\hbar \omega} - 1}) \Gamma_{Q + \hbar
     \omega}^{- 1} (\Gamma_{Q + \hbar \Delta_F} [\Gamma_{\hbar \Delta_F}^{- 1}
     (S (\lambda V)) \Gamma_{\hbar \Delta_F}^{- 1} (F)])\Bigr]\\
     = & \Gamma_{Q + \hbar \omega}\Bigl[\Gamma_{Q + \hbar \omega}^{- 1}
     (\Gamma_Q (S (\lambda V))^{\star_{\hbar \omega} - 1}) \Gamma_{Q + \hbar
     \omega}^{- 1} (\Gamma_{Q + \hbar \Delta_F} [\Gamma_{\hbar \Delta_F}^{- 1}
     \Gamma_Q^{- 1} \Gamma_Q (S (\lambda V)) \Gamma_{\hbar \Delta_F}^{- 1}
     \Gamma_Q^{- 1} \Gamma_Q (F)])\Bigr]\\
     = & \Gamma_{Q + \hbar \omega} \Bigl[\Gamma_{Q + \hbar \omega}^{- 1}
     (\Gamma_Q (S (\lambda V))^{\star_{\hbar \omega} - 1}) \Gamma_{Q + \hbar
     \omega}^{- 1} (\Gamma_{Q + \hbar \Delta_F} [\Gamma_{Q + \hbar
     \Delta_F}^{- 1} \Gamma_Q (S (\lambda V)) \Gamma_{Q + \hbar \Delta_F}^{-
     1} \Gamma_Q (F)])\Bigr] .
   \end{align*}
Exploiting once more  Equation~\ref{Eq: algebra homomorphism}, $\Gamma_Q [R_{\lambda V} (F)]$ can be rewritten as
\begin{equation}
  \Gamma_Q [R_{\lambda V} (F)] = \Gamma_Q ((S (\lambda V))^{\star_{\hbar
  \omega} - 1}) \cdot_{Q + \hbar \omega} [(\Gamma_Q (S (\lambda V)) \cdot_{Q +
  \hbar \Delta_F} \Gamma_Q (F))] . \label{Eq:Q-Bogoliubov}
\end{equation}
Giving a closer look to this expression and focusing on $\Gamma_Q  (S (\lambda V))$, on account of Equation~\eqref{Eq:S-matrix} we can decompose this term as 
\begin{equation}\label{Eq: Q-S matrix expansion}
\Gamma_Q  (S (\lambda V)) = \sum_{n \geqslant 0} \frac{1}{n!} \left(
\frac{i \lambda}{\hbar} \right)^n \Gamma_Q \mathcal{T}_n^{\hbar \Delta_F}
(V^{\otimes n}) \,,
\end{equation}
where, in view of Equation \eqref{Eq: time-ordered product},
\[ \mathcal{T}_n^{\hbar \Delta_F} (V \otimes \ldots \otimes V) = \Gamma_{\hbar
   \Delta_F}  [\Gamma_{\hbar \Delta_F}^{- 1} (V) \ldots \Gamma_{\hbar
   \Delta_F}^{- 1} (V)]\,. \]
The action of $\Gamma_{Q}$ yields
\begin{align}\label{Eq: conto}
     \Gamma_Q \mathcal{T}_n^{\hbar \Delta_F} (V \otimes \ldots \otimes V)  =
     & \Gamma_Q \Gamma_{\hbar \Delta_F} [\Gamma_{\hbar \Delta_F}^{- 1} (V)
     \ldots \Gamma_{\hbar \Delta_F}^{- 1} (V)]\nonumber\\
      = & \Gamma_Q \Gamma_{\hbar \Delta_F} [\Gamma_{\hbar \Delta_F}^{- 1}
     \Gamma_Q^{- 1} \Gamma_Q (V) \ldots \Gamma_{\hbar \Delta_F}^{- 1}
     \Gamma_Q^{- 1} \Gamma_Q (V)]\nonumber\\
      = & \Gamma_{Q + \hbar \Delta_F} [\Gamma_{Q + \hbar \Delta_F}^{- 1}
     \Gamma_Q (V) \ldots \Gamma_{Q + \hbar \Delta_F}^{- 1} \Gamma_Q (V)]\nonumber\\
      = & \Gamma_Q (V) \cdot_{Q + \hbar \Delta_F} \ldots \cdot_{Q + \hbar
     \Delta_F} \Gamma_Q (V)\nonumber\\
      = & \mathcal{T}^{\hbar \Delta_F + Q}_n (\Gamma_Q (V) \otimes
     \ldots \otimes \Gamma_Q (V)) .
   \end{align} 
To summarize we have constructed the {\em  $Q - S$-matrix} as
\begin{equation}\label{Eq:Q-S-matrix}
  \Gamma_Q (S (\lambda V)) = \sum_{n \geqslant 0} \frac{1}{n!} \left( \frac{i
  \lambda}{\hbar} \right)^n \mathcal{T}_n^{\hbar \Delta_F + Q} (\Gamma_Q (V)
  \otimes \ldots \otimes \Gamma_Q (V)).
\end{equation}
Bearing in mind that our primary goal is to prove the convergence of the series defining an interacting field associated to the sine-Gordon model, it is convenient to give an analogous characterization of the  $\star_{\hbar \omega}$-inverse of $S$ in terms of anti-time ordered products, namely
\begin{align*}
     S (\lambda V)^{\star_{\hbar \omega} - 1} \assign & \sum_{n \geqslant 0}
     \frac{1}{n!} \left( \frac{- i \lambda}{\hbar} \right)^n V \cdot_{\hbar
     \Delta_{\tmop{AF}}} \ldots \cdot_{\hbar \Delta_{\tmop{AF}}} V\\
      = & \sum_{n \geqslant 0} \frac{1}{n!} \left( \frac{- i \lambda}{\hbar}
     \right)^n \mathcal{T}_n^{\hbar \Delta_{\tmop{AF}}} (V \otimes \ldots
     \otimes V),
   \end{align*} 
where $\Delta_{\tmop{AF}}$ denotes the anti-Feynman propagator, see Equation \eqref{Eq: anti-feynman}. As a matter of fact, mirroring the computations in Equation \eqref{Eq: conto}, it descends that we can rewrite the $\star_{\hbar \omega}$-inverse of the $Q - S$-matrix
\begin{equation}
  \Gamma_Q (S (\lambda V)^{\star_{\hbar \omega} - 1}) = \sum_{n \geqslant 0}
  \frac{1}{n!} \left( \frac{- i \lambda}{\hbar} \right)^n \mathcal{T}_n^{\hbar
  \Delta_{\tmop{AF}} + Q} (\Gamma_Q (V) \otimes \ldots \otimes \Gamma_Q (V)),
  \label{Eq:inverse-Q-S-matrix}
\end{equation}
where
\[ \mathcal{T}_n^{\hbar \Delta_{\tmop{AF}} + Q}  (\Gamma_Q (V) \otimes \ldots
   \otimes \Gamma_Q (V)) = \Gamma_Q (V) \cdot_{Q + \hbar
   \Delta_{\tmop{AF}}} \ldots \cdot_{Q + \hbar \Delta_{\tmop{AF}}} \Gamma_Q
   (V) \,. \]

\subsection{Convergence of the $Q - S$-matrix}\label{Sec: Convergence of the Q-S matrix}
In the preceding section the main result has been the definition in Equation \eqref{Eq:Q-S-matrix} of the $Q - S$-matrix $\Gamma_Q  (S (\lambda V))$ as a formal power series expansion in the coupling constant $\lambda$. In this section we prove absolute convergence of the series for the two-dimensional sine-Gordon model. This amount to choosing
\begin{equation}
  V_g (\varphi) \assign \int_{\mathbb{R}^2} d \mu_z \cos (a \varphi (z)) g (z) =
  \frac{1}{2}  (V_{a, g} + V_{- a, g}), \qquad \qquad V_{a, g} (\varphi)
  \assign \int_{\mathbb{R}^2} \mathd \mu_z e^{ia \varphi (z)} g (z) .
  \label{Eq:setting}
\end{equation}
Henceforth, for simplicity and without loss of generality, we shall assume that $g\geq 0$. Our strategy consists of extending to the case in hand the methods devised in \cite{BPR23-Sine-Gordon-Massive}, so to account also for the contribution of the correlations codified by $Q$, see Equation \eqref{Eq:def-of-Q}. To this end a few preliminary considerations are necessary. First of all, in view of Equation \eqref{Eq:setting}, Equation \eqref{Eq:Q-S-matrix} takes the form
\begin{equation}\label{Eq: non ho idee}
	\Gamma_Q (S (\lambda V_g)) = \sum_{n \geqslant 0} \frac{1}{n!} \left( \frac{i
		\lambda}{2 \hbar} \right)^n  \sum_{k = 0}^n \binom{n}{k}
	\mathcal{T}_n^{\hbar \Delta_F + Q} ((\Gamma_Q (V_{a, g}))^{\otimes k}
	\otimes (\Gamma_Q (V_{- a, g}))^{\otimes (n - k)}) . 
\end{equation}
Focusing on the single elements appearing in the tensor products and recalling Equation \eqref{Eq:GammaQ}, it descends 
\begin{align*}
     \Gamma_Q (V_{\pm a, g}) (\varphi)  = & \exp \left\{ \frac{1}{2} 
     \left\langle Q (z, y), \frac{\delta^2}{\delta \varphi (z) \delta \varphi
     (y)} \right\rangle \right\}  \int_{\mathbb{R}^2} \mathd \mu_x e^{\pm ia
     \varphi (x)} g (x)\\
      = & \sum_{k \geqslant 0} \frac{(\pm ia)^{2 k}}{k! 2^k}
     \int_{\mathbb{R}^{2(2 k + 1)}} \mathd \mu_x \prod_{i=1}^{\ell}\mathd\mu_{z_i}\mathd\mu_{y_i} \prod_{\ell = 1} \delta (x - z_{\ell}) \delta (x - y_{\ell}) Q
     (z_{\ell}, y_{\ell}) e^{\pm ia \varphi (x)} g (x)\\
      = & \sum_{k \geqslant 0} \frac{1}{k! 2^k}  \int_{\mathbb{R}^2} \mathd
     \mu_x e^{\pm ia \varphi (x)} (- a^2 Q (x, x))^k g (x)\\
      = & \int_{\mathbb{R}^2} \mathd \mu_x g (x) e^{\pm ia \varphi (x)} 
     \sum_{k \geqslant 0} \frac{(- \frac{a^2}{2} Q (x, x))^k}{k!}\\
      = & \int_{\mathbb{R}^2} \mathd \mu_x g (x) e^{\pm ia \varphi (x)} e^{-
     \frac{a^2}{2} Q (x, x)},
   \end{align*}
where the last identity is a byproduct of the regularity of $Q (x, x)$, see Lemma \ref{Lem: On the regularity of Q}, and of $g$ lying in $\mathcal{D}(\mathbb{R}^2)$. These entail that the series $\sum_{k \geqslant 0} g (x) \frac{(- \frac{a^2}{2} Q (x, x))^k}{k!}$ is uniformly convergent. In other words, introducing
\begin{equation}\label{Eq: g_Q}
g_Q (x) \assign g (x) e^{- \frac{a^2}{2} Q (x, x)} \in \mathcal{D} (\mathbb{M}),
\end{equation}
we have thus proved that
\begin{equation}\label{Eq: boh}
	\Gamma_Q (V_{\pm a, g}) = V_{\pm a, g_Q} ,
\end{equation}
which entails in turn that
\begin{equation}\label{Eq:action-of-gamma-Q-on-interaction}
  \Gamma_Q (V_g) = V_{g_Q}. 
\end{equation}
Furthermore, the $n$-th order coefficient of the formal power series in Equation \eqref{Eq: non ho idee} can be written as
\begin{align*}
	\mathcal{T}_n^{\hbar \Delta_F + Q} &(\Gamma_Q (V_{a_1, g}) \otimes \ldots
	\otimes \Gamma_Q (V_{a_n, g}))\\
	&=\int_{\mathbb{R}^{2n}} \mathd \mu_{x_1}
	\ldots \mathd \mu_{x_n} e^{i \sum_k a_k \varphi (x_k)} e^{- \sum_{1 \leq i
			< j \leq n} a_i a_j  (\hbar \Delta^F (x_i, x_j) + Q (x_i, x_j))} g_Q (x_1)
	\ldots g_Q (x_n),
\end{align*}
where $a_i$ can acquire the values $\pm a$ depending on the case under investigation. Observing that the positivity of $Q(x,x)$ as per Lemma \ref{Lem: On the regularity of Q} entails that $0 \leqslant g_Q (x)\leqslant g (x)$ and denoting by $H\assign\tmop{Re} (\Delta_F)$, we can estimate the absolute value of the time-ordered product as follows
\begin{align*}
     |\mathcal{T}_n^{\hbar \Delta_F + Q} (\Gamma_Q (V_{a_1, g}) \otimes \ldots
     \otimes \Gamma_Q (V_{a_n, g})) |  \leqslant & \int_{\mathbb{R}^{2n}} \mathd
     \mu_X e^{- \sum_{1 \leq i < j \leq n} a_i a_j  (\hbar H (x_i, x_j) + Q
     (x_i, x_j))} g_Q (x_1) \ldots g_Q (x_n) \\
      \leqslant & \int_{\mathbb{R}^{2n}} \mathd \mu_X e^{- \sum_{1 \leq i < j
     \leq n} a_i a_j  (\hbar H (x_i, x_j) + Q (x_i, x_j))} g (x_1) \ldots g
     (x_n)\\
      = & \tmop{ev}_0 (\mathcal{T}_n^{\hbar H + Q} (V_{a_1, g} \otimes
     \ldots \otimes V_{a_n, g})),
   \end{align*} 
where $\mathd\mu_X=\mathd\mu_{x_1}\ldots\mathd\mu_{x_n}$ and where $\tmop{ev}_0 (F) = F (0)$ is the evaluation map at the field configuration $\varphi = 0$. As a consequence, denoting with $[\Gamma_Q (S (\lambda V_g))]_n$ the $n$-th order coefficient in $\lambda$ of the expansion in Equation~\eqref{Eq:Q-S-matrix}, for any configuration $\varphi \in \mathcal{E} (\mathbb{R}^2)$, it holds that
\begin{equation}
  | [\Gamma_Q (S (\lambda V_g))]_n | \leqslant \frac{1}{n!} \left(
  \frac{\lambda}{\hbar} \right)^n \tmop{ev}_0 [\mathcal{T}_n^{\hbar H + Q}
  (V_g \otimes \ldots \otimes V_g)] \,. \label{Eq:n-th-order-bound}
\end{equation}


\begin{remark}
  Both $H\assign\tmop{Re} (\Delta_F)$ and $Q$ are symmetric bi-distribution. This
  entails that convergence of the $Q - S$-matrix, which is built out of a non-commutative product, is tantamount to that of the exponential series of $V_g$ computed
  with respect to the commutative product $\cdot_{\hbar H + Q}$, eventually evaluated at the zero configuration.
\end{remark}

The analysis in \cite{BPR23-Sine-Gordon-Massive} is mainly based on techniques of conditioning and inverse conditioning borrowed from Euclidean
quantum field theory \cite{Frolich-CMP-76-2-d-stat-mech}. They allow one to control the convergence of the $S$-matrix of a massive theory in
terms of the one associated to the massless counterpart. Before delving into the application of such methods to the case in hand, we report for the sake of completeness a notable result.
\begin{theorem}[Thm. 2.2, \cite{BPR23-Sine-Gordon-Massive}]
  \label{Thm:conditioning}Let $w_0, w_1 \in \mathcal{D}' (\mathbb{R}^2\times\mathbb{R}^2)$ be
  positive, real and symmetric. Let
  \[ w_0 - w_1 = P - N, \]
  where $P, N \in \mathcal{D}' (\mathbb{R}^2\times\mathbb{R}^2)$ are once more positive and
  symmetric. Moreover, assume both that $\exists \left. N\right|_{\tmop{Diag}_2}\in L^\infty(\mathbb{R}^2)$ where $\tmop{Diag}_2$ is the diagonal of $\mathbb{R}^2\times\mathbb{R}^2$.
  Hence, for $V_g$ as in Equation~\eqref{Eq:setting}, it holds that, setting $v_0 = w_0 + N$,
  \[ \tmop{ev}_0 [\exp_{w_1} (\lambda V_g)] \leqslant \tmop{ev}_0 [\exp_{v_0}
     (\lambda V_g)] \leqslant 2 \tmop{ev}_0 [\exp_{w_0} (2\lambda  e^{\frac{a^2}{2} K} V_g)]. \]
\end{theorem}

We review succinctly the strategy adopted in \cite{BPR23-Sine-Gordon-Massive} and, subsequently, we show how to adapt it to the case of our interest. In \cite{BPR23-Sine-Gordon-Massive} the authors consider the distributions in Theorem \ref{Thm:conditioning} as being
\[ w_0 = H_0 \assign \tmop{Re} (\Delta_{F, 0}), \qquad \qquad \qquad w_1 = H
   \assign \tmop{Re} (\Delta_F), \]
where $\Delta_{F, 0}$ denotes the Feynman propagator associated with
the massless theory while $\Delta_F$ is the one associated with the massive
counterpart. To avoid possible sources of confusion, we shall adopt the following convention to
distinguish a bi-distribution from its integral kernel: $\forall f,f^\prime\in\mathcal{D}(\mathbb{R}^2)$
\begin{equation}\label{Eq: integral kernels}
	 H_0 (f\otimes f^\prime) = \int_{\mathbb{R}^4} \mathd \mu_z \mathd \mu_{z_1}
	\mathcal{H}_0 (z-z_1) f (z) f^\prime (z_1), \quad H (f\otimes f^\prime) =
	\int_{\mathbb{R}^4} \mathd \mu_z \mathd \mu_{z_1} \mathcal{H} (z-z_1) f
	(z) f^\prime (z_1).
\end{equation}
It can be seen that $H \in \mathcal{D}^\prime(\mathbb{R}^2\times\mathbb{R}^2)$ is positive, translation invariant and symmetric being generated by $\mathcal{H}(x)=\frac{1}{2\pi}\mathrm{Re}\left(K_0(m\sqrt{x^2})\right)$, where $K_0$ is a modified Bessel function of second kind and where $x^2$ is the Lorentzian distance between the point $x\in\mathbb{R}^2$ and the origin, {\tmem{cf.}} {\cite[Prop. 2.4]{BPR23-Sine-Gordon-Massive}}. Considering instead the massless case, the integral kernel of $H_0$ reads
\begin{equation}
  \mathcal{H}_0 (x) = - \frac{1}{4 \pi} \log \left| \frac{x^2}{4 \mu^2}
  \right|, \label{Eq:H0-kernel}
\end{equation}
where $\mu$ is a positive, reference length scale. By direct inspection, $H_0\in \mathcal{D}^\prime(\mathbb{R}^2\times\mathbb{R}^2)$ is symmetric in its arguments, but it fails to be positive. Nonetheless, as proven in \cite[Prop. 2.3]{BPR23-Sine-Gordon-Massive}, if we consider the space-time diamond of size $\mu$ 
\begin{equation}\label{Eq: spacetime diamond}
	 D_{\mu} \assign \left\{ (t, s) \in \mathbb{R}^2 \barsuchthat - \mu < t - s <
	\mu, \;\textrm{and}\; - \mu < t + s < \mu \right\},
\end{equation}
then the restriction thereon of $H_0$ is positive. In other words, $H_0$ is both symmetric and positive if we consider test-functions lying in $\mathcal{D} (D_{\mu}^2)$. 
This lack of positivity of $H_0$ on the whole space-time requires 
to restrict the support of the interaction to the  space-time diamond defined in Equation \eqref{Eq: spacetime diamond}. This is tantamount to choosing the cutoff function $g$ in Equation~\eqref{Eq:setting} to be positive and such that $\mathrm{supp}(g)\subseteq D_{\mu}$. 

In order to apply Theorem~\ref{Thm:conditioning}, we must decompose the restriction of $H_0 - H$
to $\mathcal{D} (D_{\mu}^2)$ in the positive and negative components. To this end, let us consider a smooth, symmetric, positive and compactly supported function $\Omega \in \mathcal{D} (\mathbb{R}^2)$ such that, denoting by $\iota : \mathbb{R}^2\times\mathbb{R}^2\rightarrow \mathbb{R}^2$ the map defined as $\iota (x, y) = x - y$, the function $\left.\Omega\circ\iota\right|_{D_{2 \mu} \times D_{2 \mu}}=1$. Focusing on $(H_0 - H)\Omega\in\mathcal{D}^\prime(\mathbb{R}^2\times\mathbb{R}^2)$, we set
\[ W (f\otimes f^\prime) \assign \int_{\mathbb{R}^4} \mathd \mu_z \mathd \mu_{z_1}
   (\mathcal{H}_0 (z - z_1) -\mathcal{H} (z - z_1)) \Omega (z - z_1) f (z) f^\prime
   (z_1). \]
By direct inspection and on account of Equation \eqref{Eq: integral kernels}, $W(f\otimes f^\prime) = (H_0 - H)  (f\otimes f^\prime)$ for $f, f^\prime \in \mathcal{D} (D_{\mu})$. Being $\mathcal{H}_0 -\mathcal{H}$ locally integrable, $(\mathcal{H}_0
-\mathcal{H}) \Omega$ is a measurable function and thus $\exists C > 0$
such that
\[ |W (f\otimes f^\prime) | \leqslant C \|f\|_{L^2(\mathbb{R}^2)} \|f^\prime\|_{L^2(\mathbb{R}^2)}. \]
Therefore one can extend per continuity $W$ to a bounded quadratic form over $L^2 (D_{\mu})$. In turn, on account of the Riesz representation theorem
$W$ can be written in terms of a bounded linear operator over $L^2 (D_{\mu})$ which acts as a multiplication in Fourier
space, {\it i.e.}, there exists $\hat{\mathcal{W}}\in L^\infty(\mathbb{R}^2)$ such that, for all $f,f^\prime\in L^2(\mathbb{R}^2)$,
\[ W (f\otimes f^\prime) = \int_{\mathbb{R}^2} \mathd \mu_k \hat{\mathcal{W}} (k) \hat{f} (k)
   \hat{f}^\prime (k). \]
Observe that, combining the Riemann-Lesbegue theorem with $(\mathcal{H}_0 -\mathcal{H}) \Omega$ being symmetric and real-valued, one can show that $\hat{\mathcal{W}}$ is in addition continuous, it vanishes at infinity and it is real-valued.

With these premises, the decomposition in positive and negative parts of
$(\mathcal{H}_0 -\mathcal{H}) \Omega$ follows suit. As a matter of facts, it suffices to
decompose the real and continuous function $\hat{\mathcal{W}}$ in its positive
and negative parts, {\tmem{i.e.}},
\[ \hat{\mathcal{W}} = \hat{\mathcal{P}} - \hat{\mathcal{N}}, \]
where $\hat{\mathcal{P}} \geqslant 0$ and $\hat{\mathcal{N}} \geqslant 0$. In turn, it descends that for all $f,f^\prime\in\mathcal{D}(\mathbb{R}^2)$
\[ N (f, f^\prime) \assign \int_{\mathbb{R}^4} \mathd \mu_x \mathd \mu_y \mathcal{N}
   (x - y) f (x) f^\prime (y), \]
\[ P (f, f^\prime) \assign \int_{\mathbb{R}^4} \mathd \mu_x \mathd \mu_y \mathcal{P}
   (x - y) f (x) f^\prime(y). \]
We also observe that
\[ N (x, x) :=\mathcal{N} (0) = \int_{\mathbb{R}} \mathd \mu_k
   \hat{\mathcal{N}} (k) = \| \hat{\mathcal{N}} \|_{L^1} . \]
In addition, as discussed in \cite[Prop. 2.5]{BPR23-Sine-Gordon-Massive}, $\hat{\mathcal{W}} \in L^1 (\mathbb{R}^2)$, implying in turn that $\hat{\mathcal{N}}$ is integrable. Thus $\exists K>0$ such that
\[ N (x, x) = \| \hat{\mathcal{N}} \|_{L^1} \leqslant K. \]
It is noteworthy that this constant can be chosen to be $K = \| \hat{\mathcal{W}}
\|_{L^1}$.

The analysis sketched above is tailored to the investigation of an interacting quantum field theory, but, since we are interested in proving the convergence of the $Q - S$-matrix, we need to improve it. To start with, we underline that, in view of Lemma~\ref{Lem:positivity-of-Q} and restricting ourselves to the compact region $D_\mu$ as per Equation \ref{Eq: spacetime diamond}, both $Q + \hbar H_0$ and $Q + \hbar H$ are positive, symmetric and real bi-distributions as required by Theorem~\ref{Thm:conditioning}. In addition we observe that the difference between $Q + \hbar H_0$ and $Q + \hbar H$ is nothing but $H_0 - H$. This entails that the operators $N$ and $P$ appearing in Theorem \ref{Thm:conditioning} are left unmodified by the kernel $Q$.

\
Before delving into the detailed proof of convergence of the $Q-S$ matrix, we need a technical result which will play a pivotal r\^ole in the forthcoming discussion.
\begin{lemma}[Cauchy determinant]
	\label{Lem:Cauchy-determinant} Consider $\mu > 0$ and $g \in \mathcal{D}
	(D_{\mu})$, where $D_\mu$ is chosen as per Equation \eqref{Eq: spacetime diamond}. Moreover, denoting by $V_a^g$ the vertex functional introduced in Equation \eqref{Eq: vertex functional}, let us define
	\[ \begin{array}{lll}
		\mathcal{O} & \assign & \tmop{ev}_0 (V_a^g \cdot_{\hbar H_0 + Q_m}
		\ldots \cdot_{\hbar H_0 + Q_m} V_a^g \cdot_{\hbar H_0 + Q_m} V_{- a}^g
		\cdot_{\hbar H_0 + Q_m} \ldots \cdot_{\hbar H_0 + Q_m} V_{- a}^g)
	\end{array} \]
	where $a_i = a$ if $i \in \{1, \ldots n\}$, $a_i = - a$ if $i \in \{n +
	1, \ldots, 2 n\}$ and where $a$ is chosen such that $\alpha \assign
	\frac{a^2 \hbar}{4 \pi} < 1$. 
	Then, for any $p \in [1, \alpha^{- 1})$, there exist two constants $C_Q(\mu)$, depending on $\mu$ in Equation \eqref{Eq: spacetime diamond} and $\tilde{C}$
	such that
	\begin{equation}\label{Eq: estimate O}
		|\mathcal{O}| \leqslant (C_Q(\mu))^{2 n^2} (4 \mu^2)^{n \alpha} \|g\|^{2
			n}_{L^q} (\tilde{C}^{2 n} (n!)^2)^{1 / p},
	\end{equation}
	where $q$ satisfies $1 / q + 1 / p = 1$.
\end{lemma}
\begin{proof}
	We follow \cite[Thm. 2.7 ]{BPR23-Sine-Gordon-Massive}, though we need to account also for $Q$. First of all we observe that, on account of the specific
	form of $\mathcal{H}_0$, {\tmem{cf.}} Equation~\eqref{Eq:H0-kernel},
	\begin{align}\label{Eq: O}
		\mathcal{O} = & \tmop{ev}_0 \Bigl( \underbrace{V_a^g \cdot_{\hbar H_0
				+ Q} \ldots \cdot_{\hbar_0 + Q} V_a^g}_{n} \cdot_{\hbar H_0 + Q}
		\underbrace{V_{- a}^g \cdot_{\hbar H_0 + Q} \ldots \cdot_{\hbar H_0 +
				Q} V_{- a}^g}_n \Bigr)\nonumber\\
		= & \int_{\mathbb{R}^{4 n}} \mathd \mu_{z_1} \ldots \mathd \mu_{z_{2
				n}} e^{- \sum_{1 \leq i < j \leq 2 n} a_i a_j  (\hbar \mathcal{H}_0
			(z_i, z_j) + Q (z_i, z_j))} g (z_1) \ldots g (z_{2 n}) \nonumber\\
		= & (4 \mu^2)^{n \alpha} \int_{\mathbb{R}^{4 n}} \mathd \mu_{X, Y}
		e^{\sum_{1 \leq i, j \leq n} a^2 Q (x_i, y_j)} e^{- \sum_{0 \leq i < j
				\leq n} a^2  [Q (x_i, x_j) + Q (y_i, y_j)]} \times\nonumber\\
		& \qquad \qquad \qquad \qquad \qquad \times \left( \frac{\prod_{1
				\leq i < j \leq n} | (x_i - x_j)^2 || (y_i - y_j)^2 |}{\prod_{i = 1}^n
			\prod_{j = 1}^n | (x_i - y_j)^2 |} \right)^{\alpha} G (X, Y)
	\end{align}
	where we introduced the notation $X = (x_1, \ldots, x_n)$, $Y = (y_1,\ldots, y_n)$, $\mathd \mu_{X, Y} = \mathd \mu_{x_1} \ldots \mathd
	\mu_{y_n}$, and $G (X, Y) = g (x_1) \ldots g (y_n)$. Furthermore, we exploited the fact that for $i \leqslant n$ and $j > n$ it holds that $a_i a_j = - a^2$ while, in all the other allowed regimes, $a_i a_j = a^2$. Considering a point $z = (t, s) \in \mathbb{R}^2$ in Euclidean coordinates, we can represent it with respect to null counterparts by setting $z^u = t - s$ and $z^v = t + s$. Hence $z=(z^u, z^v)$ and the squared Lorentzian distance factorizes as $|z^2 | = |z^u | |z^v |$. By defining the $n \times n$ matrices
	\[ \mathcal{D}^v_{ij} = \frac{1}{(x_i^v - y_j^v)}, \qquad 
	\mathcal{D}^u_{ij} = \frac{1}{(x_i^u - y_j^u)}, \]
	one can write {\cite{BPR23-Sine-Gordon-Massive,Bahns-Rejzner_Sine-Gordon}}
	\[ \frac{\prod_{1 \leq i < j \leq n} | (x_i - x_j)^2 || (y_i - y_j)^2
		|}{\prod_{i = 1}^n \prod_{j = 1}^n | (x_i - y_j)^2 |} = | \det
	\mathcal{D}^v | | \det \mathcal{D}^u |, \]
	where
	\[ \det \mathcal{D}^b = \sum_{\pi} \prod_{i = 1}^n \frac{(- 1)^{| \pi
			|}}{(x_i^b - y_{\pi (i)}^b)}, \quad b \in \{u, v\}, \]
	and where $\pi$ runs over all partitions of $\{1, \ldots, n\}$. As a consequence Equation \eqref{Eq: O} can be rewritten in terms of Cauchy determinants as
	\begin{equation}
		\mathcal{O}= (4 \mu^2)^{n \alpha} \int_{\mathbb{R}^{4 n}} \mathd \mu_{X,
			Y} e^{\sum_{1 \leq i < j \leq n} a^2 Q (x_i, y_j)} e^{- \sum_{1 \leq i < j
				\leq n} a^2  [Q (x_i, x_j) + Q (y_i, y_j)]} | \det \mathcal{D}^v
		|^{\alpha} | \det \mathcal{D}^u |^{\alpha} G (X, Y) .
		\label{Eq:observable-to-estimate}
	\end{equation}
	In order to obtain a bound on the absolute value of $\mathcal{O}$, first of all we introduce
	\[ C_Q(\mu) \assign \left(\sup_{x, y \in D_{\mu}} e^{a^2 Q (x, y)}\right)^{\frac{1}{2}}, \]
	observing that, on account of the local boundedness of $Q$ inherited from its
	continuity, $C_Q(\mu) < \infty$ for every finite $\mu$. This entails that
	\[ e^{\sum_{1 \leq i, j \leq n} a^2 Q (x_i, y_j)} = \prod_{1 \leq i, j \leq
		n} e^{a^2 Q (x_i, y_j)} \leqslant (C^2_Q(\mu))^{n^2} \]
Combining this information with the estimate
	\[ e^{- a^2 Q (x, y)} \leqslant 1. \]
	we can rewrite Equation~\eqref{Eq:observable-to-estimate}
	to obtain
	\[ \mathcal{O} \leqslant (C_Q(\mu))^{2 n^2}  (4 \mu^2)^{n \alpha} 
	\int_{\mathbb{R}^{4 n}} \mathd \mu_{X, Y} | \det \mathcal{D}^v |^{\alpha}
	| \det \mathcal{D}^u |^{\alpha} G (X, Y) . \]
	In view of \cite[Thm. 3.4]{Frolich-CMP-76-2-d-stat-mech} it holds that, for any $p$ and $q$ such that $1 / q + 1 / p = 1$,
	\[ \mathcal{O} \leqslant (C_Q(\mu))^{2 n^2}  (4 \mu^2)^{n \alpha} \|G\|_{L^q
		(D_{\mu}^{2 n})} \|| \det \mathcal{D}^v |^{\alpha} | \det \mathcal{D}^u
	|^{\alpha} \|_{L^p (D_{\mu}^{2 n})}. \]
 Having assumed that $\alpha = \frac{a^2 \hbar}{4 \pi} < 1$, we choose $p \geqslant 1$ so that $\alpha p < 1$. This entails that the product
	\[ \prod_{i = 1}^n \frac{1}{|x_i^b - y_{\pi (i)}^b |^{\alpha p}}, \]
	is locally integrable and, thus, $\|| \det \mathcal{D}^v |^{\alpha} | \det
	\mathcal{D}^u |^{\alpha} \|_{L^p (D_{\mu}^{2 n})} \lesssim 1$. More
	precisely,
	\begin{align*}
		\mathcal{O}  \leqslant & (C_Q(\mu))^{2 n^2} (4 \mu^2)^{n \alpha} \|G\|_{L^q
			(D_{\mu}^{2 n})} \left( \sum_{\pi} \sum_{\pi'} \int_{\mathbb{R}^{4 n}}
		\mathd \mu_{X, Y} \prod_{i = 1}^n \frac{1}{|x_i^u - y_{\pi (i)}^u
			|^{\alpha p}}  \prod_{\ell = 1}^n \frac{1}{|x_{\ell}^v - y_{\pi'
				(\ell)}^v |^{\alpha p}} \right)^{\frac{1}{p}}\\
		\leqslant & (C_Q(\mu))^{2 n^2} (4 \mu^2)^{n \alpha} \|G\|_{L^q
			(D_{\mu}^{2 n})} (\tilde{C}^{2 n} (n!)^2)^{1 / p},
	\end{align*}
	where the factor $(n!)^2$ comes from counting the number of admissible permutations $\pi$ and $\pi'$, while $K$ is a constant such that
	\[ \int_{\mathbb{R}^{4 n}} \mathd \mu_{X, Y} \prod_{i = 1}^n \frac{1}{|x_i^u
		- y_{\pi (i)}^u |^{\alpha p}}  \prod_{\ell = 1}^n \frac{1}{|x_{\ell}^v -
		y_{\pi' (\ell)}^v |^{\alpha p}} \leqslant K^{2 n} . \]
	Finally, observing that
	\[ \|G\|_{L^q (D_{\mu}^{2 n})} = \|g\|_{L^q (D_{\mu})}^{2 n}, \]
	we conclude that
	\[ \mathcal{O} \leqslant (C_Q)^{2 n^2} (4 \mu^2)^{n \alpha} \|g\|_{L^q
		(D_{\mu})}^{2 n} (\tilde{C}^{2 n} (n!)^2)^{1 / p} . \]
\end{proof}

We are in a position to prove a key result on the convergence of the $Q - S$-matrix introduced in Equation \eqref{Eq:Q-S-matrix}.

\begin{theorem}  \label{Thm:estimates-Q-S-matrix-I}
Let $g \in \mathcal{D} (D_{\mu})$ be a
  positive test-function supported on the compact domain $D_\mu$ as per Equation \eqref{Eq: spacetime diamond} and let $[\Gamma_Q (S (\lambda V))]_n$ be the $n$-th perturbative coefficient of the $Q - S$-matrix introduced in Equation \eqref{Eq:Q-S-matrix} with the interaction functional
  \[ V_g (\varphi) \assign \int_{\mathbb{R}^2} \mathd \mu_z \cos (a \varphi
     (z)), \]
  with $0<a < 4 \pi \hbar^{- 1}$. Then, setting $\alpha
  \assign \frac{a^2 \hbar}{4 \pi}$, there exist positive constants $\tilde{C}$, $C_Q(\mu)$ and $K$
  such that
  \begin{equation}
    | [\Gamma_Q (S (\lambda V)) (\varphi)]_n | \leqslant \frac{2 (2 \mu)^{n
    \alpha} (C_Q)^{n^2}}{(n!)^{1 - 1 / p}} \left( \frac{2 \lambda e^{2^{- 1}
    a^2 K}}{\hbar} \right)^n \|g\|^n_{L^q} C^{n / p},\qquad \forall \varphi\in\mathcal{E} (\mathbb{R}^2),
    \label{Eq:bound-n-th-order-QS}
  \end{equation}
  for any $p \in [1, \alpha^{- 1})$ such that $1 / q + 1 / p = 1$.
\end{theorem}

\begin{proof}
  As highlighted by Equation~\eqref{Eq:n-th-order-bound}, it suffices to exhibit
  a suitable bound for the $n$-th term in the expansion of 
  \[ \tmop{ev}_0 \left[ \exp_{\hbar H + Q} \left( \frac{\lambda}{\hbar} V_g
  \right) \right] , \]
  namely
  \[ \frac{1}{n!} \left( \frac{\lambda}{\hbar} \right)^n \tmop{ev}_0
     [\mathcal{T}_n^{\hbar H + Q} (V_g \otimes \ldots \otimes V_g)]. \]
  On account of the discussion about the positive and negative parts of the integral kernels under scrutiny as per Theorem \ref{Thm:conditioning} and as per the following discussion, we can conclude that
  \[ \tmop{ev}_0 \left[ \exp_{\hbar H + Q} \left( \frac{\lambda}{\hbar} V_g
     \right) \right] \leqslant 2 \tmop{ev}_0 \left[ \exp_{\hbar H_0 + Q}
     \left( \frac{2 \lambda e^{2^{- 1} a^2 K}}{\hbar} V_g \right) \right], \]
  with $K$ a suitable constant, whose dependence on the underlying data has no relevance in the current analysis. Hence, thanks to conditioning, it suffices to control the series associated to the massless theory in order to bound the massive one. Following the strategy outlined in \cite{BPR23-Sine-Gordon-Massive}, we can exploit the Cauchy-Schwartz inequality, whose validity is guaranteed by the positivity of $\hbar H_0 + Q$ on the support of $g^{\otimes 2}$. More precisely, for any functionals $A$ and $B$ in the algebra $\mathcal{A}_{\hbar H_0 + Q}^V$, see Equation \eqref{Eq: algebra 2}, it holds that
  \[ | \tmop{ev}_0 (A \cdot_{\hbar H_0 + Q} B) | \leqslant \sqrt{| \tmop{ev}_0
     (A^{\ast} \cdot_{\hbar H_0 + Q} A) |} \sqrt{| \tmop{ev}_0 (B^{\ast}
     \cdot_{\hbar H_0 + Q} B) |}, \]
  which yields, setting $B=1$,
  \begin{equation}\label{Eq: CS}
  	| \tmop{ev}_0 (A) |^2 \leqslant | \tmop{ev}_0
  	(A^{\ast} \cdot_{\hbar H_0 + Q} A) | . 
  \end{equation}
  As a consequence, observing that $V_{a, g}^{\ast} = V_{- a, g}$, a direct application of Equation \eqref{Eq: CS} implies
  \[ \tmop{ev}_0 (V_g \cdot_{\hbar H_0 + Q} \ldots \cdot_{\hbar H_0 + Q} V_g)
     = \qquad \qquad \qquad \qquad \qquad \qquad \qquad \qquad \qquad \qquad
     \qquad \qquad  \]
\begin{align*}
       \qquad\qquad= & \frac{1}{2^n} \sum_{k = 0}^n \binom{n}{k} \tmop{ev}_0 \left(
       \underbrace{V_{a, g} \cdot_{\hbar H_0 + Q} \ldots \cdot_{\hbar H_0 + Q}
       V_{a, g}}_k \cdot_{\hbar H_0 + Q} \underbrace{V_{- a, g} \cdot_{\hbar H_0
       + Q} \ldots \cdot_{\hbar H_0 + Q} V_{- a, g}}_{n-k} \right)\\
       \leqslant & \frac{1}{2^n} \sum_{k = 0}^n \binom{n}{k} \left[
       \tmop{ev}_0 \left( \underbrace{V_{a, g} \cdot_{\hbar H_0 + Q} \ldots
       \cdot_{\hbar H_0 + Q} V_{a, g}}_n \cdot_{\hbar H_0 + Q} \underbrace{V_{-
       a, g} \cdot_{\hbar H_0 + Q} \ldots \cdot_{\hbar H_0 + Q} V_{- a, g}}_n
       \right) \right]^{\frac{1}{2}}\\
       = & \left[ \tmop{ev}_0 \left( \underbrace{V_{a, g} \cdot_{\hbar H_0 +
       Q} \ldots \cdot_{\hbar H_0 + Q} V_{a, g}}_n \cdot_{\hbar H_0 + Q}
       \underbrace{V_{- a, g} \cdot_{\hbar H_0 + Q} \ldots \cdot_{\hbar H_0 +
       Q} V_{- a, g}}_n \right) \right]^{\frac{1}{2}} .\\
     \end{align*}
  Since all hypotheses of Lemma~\ref{Lem:Cauchy-determinant} are met, we apply the estimate in Equation \eqref{Eq: estimate O} to obtain
  \[ \tmop{ev}_0 (V_g \cdot_{\hbar H_0 + Q} \ldots \cdot_{\hbar H_0 + Q} V_g)
     \leqslant (C_Q)^{n^2} (2 \mu)^{n \alpha} \|g\|^n_{L^q} (C^n n!)^{1 / p}.
  \]
  As a consequence, it descends that
  \begin{align*}
       | [\Gamma_Q (S (\lambda V)) (\varphi)]_n |  \leqslant & \frac{2}{n!}
       \left( \frac{2 \lambda e^{2^{- 1} a^2 K}}{\hbar} \right)^n (C_Q)^{n^2}
       (2 \mu)^{n \alpha} \|g\|^n_{L^q} (C^n n!)^{1 / p}\\
       = & \frac{2 (2 \mu)^{n \alpha} (C_Q)^{n^2}}{(n!)^{1 - 1 / p}} \left(
       \frac{2 \lambda e^{2^{- 1} a^2 K}}{\hbar} \right)^n \|g\|^n_{L^q} C^{n
       / p},
     \end{align*}
 which concludes the proof.
\end{proof}

\noindent As a consequence Theorem \ref{Thm:estimates-Q-S-matrix-I}, we state the following crucial corollary.
\begin{corollary}\label{Cor: convergence S matrix}
	For any, but fixed $p\in[1,\alpha^{-1})$ with $\alpha$ defined as in Theorem~\ref{Thm:estimates-Q-S-matrix-I}, the series
	\[ \Gamma_Q (S (\lambda V)) (\varphi) = \sum_{n \geqslant 0} [\Gamma_Q (S
	(\lambda V)) (\varphi)]_n, \]
	is uniformly convergent for any $\varphi \in \mathcal{E} (\mathbb{M})$.
\end{corollary}

\begin{proof}
	Absolute convergence directly follows from the observation that, as inferred from Stirling formula, the factorial $(n!)^{1 - 1 / p}$ at the denominator of
	Equation~\eqref{Eq:bound-n-th-order-QS} for $p > 1$, tames the exponential growth due to the other factors.
\end{proof}

\begin{remark}
  We observe that the argument yielding absolute convergence of the $Q -
  S$-matrix holds true also for the case of $\Gamma_Q ((S (\lambda
  V))^{\star_{\hbar \omega} - 1})$. As we have shown in
  Equation~\eqref{Eq:inverse-Q-S-matrix},
  \[ \Gamma_Q (S (\lambda V)^{\star_{\hbar \omega} - 1}) = \sum_{n \geqslant
     0} \frac{1}{n!} \left( \frac{- i \lambda}{\hbar} \right)^n
     \mathcal{T}_n^{\hbar \Delta_{\tmop{AF}} + Q} (\Gamma_Q (V) \otimes \ldots
     \otimes \Gamma_Q (V)). \]
  The only difference with respect to $\Gamma_Q (S (\lambda
  V))$ lies in the fact that we have to work with the
  anti-Feynman rather than the Feynman propagator. Nonetheless, since
  $\tmop{Re} (\Delta_{\tmop{AF}}) = \tmop{Re} (\Delta_F) = H$, the convergence result follows suit.
\end{remark}

\begin{remark}
Observe that the convergence of the $Q-S$-matrix proven in Theorem \ref{Thm:estimates-Q-S-matrix-I} is guaranteed only if we restrict the attention to the diamond $D_\mu$ as per Equation \eqref{Eq: spacetime diamond}, where $\mu$ is arbitrary but fixed. In \cite{BPR23-Sine-Gordon-Massive} this limitation has been removed, though the relevant proof strongly relies on the underlying kernels being translation invariant both in time and in space. Alas, we cannot extend this result to the case under scrutiny since the covariance $Q$ as per Equation \eqref{Eq:def-of-Q} is defined in terms of cut-off along the time directions which inhibits translation invariance. This limitation is not a major hurdle as far as the construction of solutions for the stochastic sine-Gordon model is concerned as we discuss in Section \ref{Sec:classical-limit}.
\end{remark}

\subsection{Convergence of the interacting field}\label{Sec:conv-interacting-field}
Goal of this section is to elaborate on Theorem \ref{Thm:estimates-Q-S-matrix-I} deriving a convergence result for the formal power series in Equation~\eqref{Eq:Q-Bogoliubov}. We recall that the $Q$-deformed version of the Bogoliubov map, for $F\in\mathcal{F}_{\textrm{loc}}^{\otimes m}(\mathbb{R}^2)$:
\begin{equation}\label{Eq:Bog-Map}
  \Gamma_Q [R_{\lambda V} (F)] = \Gamma_Q ((S (\lambda V_g))^{\star_{\hbar
  \omega} - 1}) \star_{Q + \hbar \omega} [(\Gamma_Q (S (\lambda V_g)) \star_{Q
  + \hbar \Delta_F} \Gamma_Q (F))] . 
\end{equation}
In this section we focus on the functional $F_f = \Phi_f$, $f \in\mathcal{D} (\mathbb{R}^2)$, see Example \ref{Ex: basic functionals}, whose interacting version reads
\[ \Gamma_Q [R_{\lambda V} (\Phi_f)] = \Gamma_Q ((S (\lambda
   V_g))^{\star_{\hbar \omega} - 1}) \star_{Q + \hbar \omega} [(\Gamma_Q (S
   (\lambda V_g)) \star_{Q + \hbar \Delta_F} \Gamma_Q (\Phi_f))] .
   \label{Eq:Q-interacting field} \]
In order to simplify the notation, we write
$$\Phi_{I, f} \assign \Gamma_Q [R_{\lambda V} (\Phi_f)].$$
The starting point consists of exhibiting a series expansion for
the interacting field which is suitable for studying its convergence. On account of the linearity of the field functional, it descends that, since only the first order functional derivative of $\Phi_f$ is
non vanishing, 
\[ \begin{array}{lll}
     \Gamma_Q (S (\lambda V_g)) \star_{Q + \hbar \Delta_F} \Gamma_Q (\Phi_f) &
     = & \Gamma_Q (S (\lambda V_g)) \star_{Q + \hbar \Delta_F} \Phi_f\\
     & = & \Gamma_Q (S (\lambda V_g)) \Phi_f + \langle \Gamma_Q (S (\lambda
     V_g))^{(1)}, (Q + \hbar \Delta_F) \Phi_f^{(1)} \rangle,
   \end{array} \]
where, in the first step, we used that $\Gamma_Q (\Phi_f) = \Phi_f$. As a consequence
\begin{equation}\label{Eq: interacting field decomposition}
	 \begin{array}{lll}
		\Phi_{I, f} & = & \Gamma_Q ((S (\lambda V_g))^{\star_{\hbar \omega} - 1})
		\star_{Q + \hbar \omega} (\Gamma_Q (S (\lambda V_g)) \Phi_f) +\\
		&  & + \Gamma_Q ((S (\lambda V_g))^{\star_{\hbar \omega} - 1}) \star_{Q
			+ \hbar \omega} (\langle \Gamma_Q (S (\lambda V_g))^{(1)}, (Q + \hbar
		\Delta_F) \Phi_f^{(1)} \rangle)\\
		& \backassign & J (V_g, \Phi_f) + M (V_g, \Phi_f) .
	\end{array} 
\end{equation}
We study separately $J (V_g, \Phi_f)$ and $M (V_g,\Phi_f)$. Starting with the first term, we shall exploit the following lemma, which is an immediate consequence of the Leibniz rule. In the next section we shall prove Lemma \ref{Lem:Leibniz-pro} as a generalization.
\begin{lemma}[Lemma 9, {\cite{Bahns-Rejzner_Sine-Gordon}}]
  \label{Lemma:Leibniz}Let $K$ be a generic integral kernel and let $A, B, C$ be smooth functionals as per Definition \ref{Def:functionals}, with $C$ linear. Then
  \begin{equation}\label{Eq: Leibniz}
  	 \mathcal{M} \circ e^{D_K} [A \otimes (BC)] = [\mathcal{M} \circ e^{D_K} (A \otimes B)] C + \mathcal{M}
  	\circ e^{D_K} [\langle A^{(1)}, KC^{(1)} \rangle \otimes B] ,
  \end{equation}
  where $A^{(1)}$, $C^{(1)}$ are first-order functional derivatives.
\end{lemma}
Thanks to the defining properties of the deformation map $\Gamma_Q$ as per Equation \eqref{Eq: algebra homomorphism} and applying Lemma~\ref{Lemma:Leibniz}, it descends that
\begin{equation}
  \begin{array}{lll}
    J (V_g, \Phi_f) & = & \Gamma_Q ((S (\lambda V_g))^{\star_{\hbar \omega} -
    1}) \cdot_{Q + \hbar \omega} (\Gamma_Q (S (\lambda V_g)) \Phi_f)\\
    & = & \mathcal{M} \circ e^{D_{Q + \hbar \omega}} [\Gamma_Q ((S (\lambda
    V_g))^{\star_{\hbar \omega} - 1}) \otimes (\Gamma_Q (S (\lambda V_g))
    \Phi_f)]\\
    & = & \mathcal{M} \circ e^{D_{Q + \hbar \omega}} [\Gamma_Q ((S (\lambda
    V_g))^{\star_{\hbar \omega} - 1}) \otimes (\Gamma_Q (S (\lambda V_g)
    \nobracket)] \Phi_f +\\
    &  & + \mathcal{M} \circ e^{D_{Q + \hbar \omega}} [\langle \Gamma_Q ((S (\lambda
    V_g))^{\star_{\hbar \omega} - 1})^{(1)}, (Q + \hbar \omega) \Phi_f^{(1)}
    \rangle \otimes \Gamma_Q (S (\lambda V_g))].
  \end{array} \label{Eq:decomposition-of-J}
\end{equation}
Exploiting once again that the deformation map acts as an homomorphism on the deformed algebra, see Remark \ref{Rem:deformation-map}, the first term in
Equation~\eqref{Eq:decomposition-of-J} can be further expanded as
\begin{align*}
	\mathcal{M} \circ e^{D_{Q + \hbar \omega}} [\Gamma_Q &((S (\lambda V_g))^{\star_{\hbar
			\omega} - 1}) \otimes (\Gamma_Q (S (\lambda V_g) \nobracket)] \Phi_f = \Gamma_Q
	((S (\lambda V_g))^{\star_{\hbar \omega} - 1}) \star_{Q + \hbar \omega}
	\Gamma_Q (S (\lambda V_g)) \Phi_f \\
	&=\Gamma_{Q + \hbar \omega} [\Gamma_{Q + \hbar \omega}^{- 1} (\Gamma_Q ((S
	(\lambda V_g))^{\star_{\hbar \omega} - 1})) \Gamma_{Q + \hbar \omega}^{- 1}
	(\Gamma_Q (S (\lambda V_g)))]\Phi_f \\
	&=\Gamma_Q \Gamma_{\hbar \omega} [\Gamma_{\hbar \omega}^{- 1} ((S (\lambda
	V_g))^{\star_{\hbar \omega} - 1}) \Gamma_{\hbar \omega}^{- 1} (S (\lambda
	V_g))]\Phi_f \\
	& = \Gamma_Q [(S (\lambda V_g))^{\star_{\hbar \omega} - 1}
	\star_{\hbar \omega} S (\lambda V_g)] \Phi_f = \Gamma_Q (1) \Phi_f = \Phi_f .
\end{align*}
Therefore the first term on the right hand side of Equation~\eqref{Eq:decomposition-of-J} coincides with $\Phi_f$ itself. Focusing on the second contribution, we observe that
\[ \mathcal{M} \circ e^{D_{Q + \hbar \omega}} [\langle \Gamma_Q ((S (\lambda
   V_g))^{\star_{\hbar \omega} - 1})^{(1)}, (Q + \hbar \omega) \Phi_f^{(1)}
   \rangle \otimes \Gamma_Q (S (\lambda V_g))] = \qquad \qquad \qquad \]
\[ \qquad \qquad \qquad = \langle \Gamma_Q ((S (\lambda V_g))^{\star_{\hbar
   \omega} - 1})^{(1)}, (Q + \hbar \omega) \Phi_f^{(1)} \rangle \star_{Q +
   \hbar \omega} \Gamma_Q (S (\lambda V_g)).\]
Replacing the $Q-S$ matrix and its $\star_{\hbar\omega}$ inverse with their formal power series in $\lambda$ as per
Equations~\eqref{Eq:Q-S-matrix} and~\eqref{Eq:inverse-Q-S-matrix} we obtain
  \begin{align}\label{Eq: J expansion}
     J (V_g, \Phi_f) = \Phi_f + \sum_{n \geqslant 0} \left( \frac{i
     \lambda}{\hbar} \right)^n \sum_{\ell = 0}^n &\binom{n}{\ell} (- 1)^{\ell}
     \ell \times\nonumber\\
     \times \, \langle \mathcal{T}_{\ell}^{\hbar \Delta_{\tmop{AF}} + Q}
     (\Gamma_Q (V_g) \otimes &\ldots \otimes \Gamma_Q (V_g)^{(1)} \otimes
     \ldots \otimes \Gamma_Q (V_g)), (Q + \hbar \omega) \Phi_f^{(1)} \rangle
     \star_{Q + \hbar \omega} \mathcal{T}_{n - \ell}^{\hbar \Delta_F + Q},
  \end{align}
where the symmetry property of the  map $\mathcal{T}_{\ell}^{\hbar
\Delta_{\tmop{AF}} + Q}$ in its arguments has been used. Observe that, on account of Example \ref{Ex: basic functionals}, it holds by direct inspection that
\[ \qquad ((Q + \hbar \omega) \Phi_f^{(1)}) (y) = \int_{\mathbb{R}^2} \mathd
   \mu_z (Q + \hbar \omega) (y - z) f (z). \]
Focusing on the interacting vertex functional
\[ V_{g } = \int_{\mathbb{R}^2} \mathd \mu_x  \frac{e^{ia \varphi (x)} + e^{- ia
   \varphi (x)}}{2} g (x), \]
a formal computation entails that 
\begin{equation}\label{Eq: derivative vertex functional}
	 V_g^{(1)} (y) = \int_{\mathbb{R}^2} \mathd \mu_x  \frac{e^{ia \varphi (x)} -
		e^{- ia \varphi (x)}}{2} ia \delta (x - y) g (x) = \frac{e^{ia \varphi (x)}
		- e^{- ia \varphi (x)}}{2} iag (y) = V' (y) g (y),	
\end{equation}
where we introduced
\[ V' (y) \assign - a \sin (a \varphi (y)) = - a \frac{e^{ia \varphi (y)} -
   e^{- ia \varphi (y)}}{2 i} . \]
Since any deformation map acts as the identity on local functionals, Equation~\eqref{Eq:action-of-gamma-Q-on-interaction} allows us to write the $n$-th order contribution to Equation \eqref{Eq: J expansion} as
\begin{align*}
\langle \mathcal{T}_{\ell}^{\hbar \Delta_{\tmop{AF}} + Q} (\Gamma_Q (V_g)
   \otimes \ldots \otimes \Gamma_Q (V_g)^{(1)} &\otimes \ldots \otimes \Gamma_Q
   (V_g)), (Q + \hbar \omega) \Phi_f^{(1)} \rangle =\\
   =&\mathcal{T}_{\ell}^{\hbar
   \Delta_{\tmop{AF}} + Q} (V_{g_Q} \otimes \ldots \otimes \langle V' g_Q, (Q
   + \hbar \omega) f \rangle \otimes \ldots \otimes V_{g_Q}), 
   \end{align*}
yielding
\begin{align}\label{Eq:final-J-expansion}
\nonumber
 J (V_g, \Phi_f) = \Phi_f + \sum_{n \geqslant 0} \left( \frac{i
    \lambda}{\hbar} \right)^n \sum_{\ell = 0}^n& \binom{n}{\ell} (- 1)^{\ell}
    \ell \\
    \times \mathcal{T}_{\ell}^{\hbar \Delta_{\tmop{AF}} + Q} (V_{g_Q} \otimes
    &\ldots \otimes \langle V' g_Q, (Q + \hbar \omega) f \rangle \otimes \ldots
    \otimes V_{g_Q}) \star_{Q + \hbar \omega} \mathcal{T}_{n - \ell}^{\hbar
    \Delta_F + Q} (V_{g_Q} \otimes \ldots \otimes V_{g_Q}) .
\end{align}
Focusing on $M (V_g, \Phi_f)$ in Equation \eqref{Eq: interacting field decomposition}, its analysis is similar to that of $J(V_g, \Phi_f)$. It can be expressed as a formal power series, since
\[ \langle \Gamma_Q (V_g)^{(1)}, (Q + \hbar \Delta_F) \Phi_f^{(1)} \rangle =
   \langle V' g_Q, (Q + \hbar \Delta_F) f \rangle, \]
and thus, \tmtextit{mutatis mutandis},
\begin{align*}
     M (V_g,& \Phi_f) = \sum_{n \geqslant 0} \left( \frac{i \lambda}{\hbar}
     \right)^n \sum_{\ell = 0}^n \binom{n}{\ell} (- 1)^{\ell} (n - \ell)
     \times\\
     &\times \, \mathcal{T}_{\ell}^{\hbar \Delta_{\tmop{AF}} + Q}
     (V_{g_Q} \otimes \ldots \otimes V_{g_Q}) \star_{Q + \hbar \omega}
     [\mathcal{T}_{n - \ell}^{\hbar \Delta_F + Q} (V_{g_Q} \otimes \ldots
     \otimes \langle V' g_Q, (Q + \hbar \Delta_F) f \rangle \otimes  \ldots
     \otimes V_{g_Q})] .
\end{align*}
Overall, we proved the following result.
\begin{theorem}\label{Thm:expansion-of-the-interacting-field}
	Let
  \[ \Phi_{I, f} \assign \Gamma_Q [R_{\lambda V} (\Phi_f)] = \Gamma_Q ((S
     (\lambda V))^{\star_{\hbar \omega} - 1}) \star_{Q + \hbar \omega}
     [(\Gamma_Q (S (\lambda V)) \star_{Q + \hbar \Delta_F} \Gamma_Q (\Phi_f))]
     . \]
It holds that
  \begin{equation}\label{Eq: expansion interacting 2}
   \begin{array}{lll}
  	\Phi_{I, f} & = & J (V_g, \Phi_f) + M (V_g, \Phi_f)\\
  	& = & \Phi_f + \sum_{n \geqslant 0} \left( \frac{i \lambda}{\hbar}
  	\right)^n J_n (V_g^{\otimes n}, \Phi_f) + \sum_{n \geqslant 0} \left(
  	\frac{i \lambda}{\hbar} \right)^n M_n (V_g^{\otimes n}, \Phi_f),
  \end{array}
  \end{equation}
  where
\begin{align}\label{Eq: nth order J}
       J_n (V_g^{\otimes n},& \Phi_f)  =  \sum_{\ell = 0}^n \binom{n}{\ell}
       (- 1)^{\ell} \ell \times\nonumber\\
        & \times \mathcal{T}_{\ell}^{\hbar \Delta_{\tmop{AF}} + Q} (V_{g_Q}
       \otimes \ldots \otimes \langle V' g_Q, (Q + \hbar \omega) f \rangle
       \otimes \ldots \otimes V_{g_Q}) \star_{Q + \hbar \omega} \mathcal{T}_{n
       - \ell}^{\hbar \Delta_F + Q} (V_{g_Q} \otimes \ldots \otimes V_{g_Q}),
\end{align}
  and
\begin{align}\label{Eq: nth order M}
       M_n (V_g^{\otimes n},& \Phi_f)  =  \sum_{\ell = 0}^n \binom{n}{\ell}
       (- 1)^{\ell} (n - \ell) \times\nonumber\\
         & \times \mathcal{T}_{\ell}^{\hbar \Delta_{\tmop{AF}} + Q} (V_{g_Q}
       \otimes \ldots \otimes V_{g_Q}) \star_{Q + \hbar \omega} \mathcal{T}_{n
       - \ell}^{\hbar \Delta_F + Q} (V_{g_Q} \otimes \ldots \otimes \langle V'
       g_Q, (Q + \hbar \Delta_F) f \rangle \otimes  \ldots \otimes V_{g_Q}) .
\end{align}
\end{theorem}

The rest of this section is devoted to proving that the series
expansion of the interacting field introduced in
Theorem~\ref{Thm:expansion-of-the-interacting-field} is absolutely
convergent. Eventually, we shall show that the series in Equation \eqref{Eq: expansion interacting 2} shares the same form as the one appearing in the construction of the $Q - S$-matrix. Once this correspondence has been settled, the proof is analogous to the one of Theorem \ref{Thm:estimates-Q-S-matrix-I} and therefore we shall omit it.

Since the perturbative expansion of the interacting field is split in two separate contributions, see Equation \eqref{Eq: expansion interacting 2}, we shall start by further manipulating the $n$-th perturbative order of $J (V_g, \Phi_f)$ as per Equation \eqref{Eq: nth order J}. Discarding the combinatorial coefficients, we consider a generic element of the sum, namely
\begin{equation}\label{Eq: generic contribution J_n}
\mathcal{T}_{\ell}^{\hbar \Delta_{\tmop{AF}} + Q} (V_{g_Q} \otimes \ldots
\otimes \langle V' g_Q, (Q + \hbar \omega) f \rangle \otimes \ldots \otimes
V_{g_Q}) \star_{Q + \hbar \omega} \mathcal{T}_{n - \ell}^{\hbar \Delta_F +
	Q} (V_{g_Q} \otimes \ldots \otimes V_{g_Q}) . 
\end{equation}
For later convenience, let us introduce the modified test function
\begin{equation}
  \tilde{g} (y) \assign g_Q (y) [(Q + \hbar \omega) f] (y),
  \label{Eq:modified-test-funtion}
\end{equation}
which inherits from $g_Q$ the property of being smooth and compactly supported. Hence Equation \eqref{Eq: generic contribution J_n} can be rewritten in a concise form as
\begin{align}\label{Eq: generic contribution J_n 2 }
	&\mathcal{T}_{\ell}^{\hbar \Delta_{\tmop{AF}} + Q} (V_{g_Q} \otimes \ldots
	\otimes V_{\tilde{g}}' \otimes \ldots \otimes V_{g_Q}) \star_{Q + \hbar
		\omega} \mathcal{T}_{n - \ell}^{\hbar \Delta_F + Q} (V_{g_Q} \otimes \ldots
	\otimes V_{g_Q})\nonumber\\
	&=\mathcal{T}_{\ell}^{\hbar \Delta_{\tmop{AF}} + Q} (V_{\tilde{g}}' \otimes
	\ldots \otimes V_{g_Q} \otimes \ldots \otimes V_{g_Q}) \star_{Q + \hbar
		\omega} \mathcal{T}_{n - \ell}^{\hbar \Delta_F + Q} (V_{g_Q} \otimes \ldots
	\otimes V_{g_Q}),  
\end{align}
where we used that the time-ordered product is commutative. As in the analysis of the S-matrix in Section \ref{Sec: Convergence of the Q-S matrix} and recalling the explicit form of the vertex functional $V_{g_Q} = 2^{- 1}(V_{a, g_Q} + V_{- a, g_Q})$, a direct computation shows that
\begin{align*}
	\mathcal{T}_{n - \ell}^{\hbar \Delta_F + Q} (V_{g_Q} \otimes \ldots
	\otimes V_{g_Q}) =  \frac{1}{2^{n - \ell}} \sum_{p = 0}^{n - \ell}
	\binom{n - \ell}{p} \mathcal{T}_{n - \ell}^{\hbar \Delta_F + Q} (V_{a,
		g_Q}^{\otimes p} \otimes V_{- a, g_Q}^{\otimes (n - \ell - p)}),
\end{align*}
while
\begin{align*}
	\mathcal{T}_{n - \ell}^{\hbar \Delta_F + Q} (V_{a, g_Q}^{\otimes p}
	\otimes V_{- a, g_Q}^{\otimes (n - \ell - p)}) = & \mathcal{T}_{n -
		\ell}^{\hbar \Delta_F + Q} \left[ \int_{\mathbb{R}^{2(n - \ell)}} e^{i
		\sum_{j = 1}^{n - \ell} a_j^{(p)} \varphi (y_j)} g_Q (y_1) \ldots g_Q
	(y_{n - \ell}) \mathd \mu_{Y_{n - \ell}} \right]\\
	= & \int_{\mathbb{R}^{2(n - \ell)}} e^{i \sum_{j = 1}^{n - \ell}
		a^{(p)}_j \varphi (y_j)} e^{- \sum_{1 \leqslant m < j \leqslant n - \ell}
		a^{(p)}_m a^{(p)}_j (Q + \hbar \Delta_F) (y_m, y_j)} g_Q (Y_{n - \ell})
	\mathd \mu_{Y_{n - \ell}}.
\end{align*}
Starting from the first factor $\mathcal{T}_{\ell}^{\hbar \Delta_{\tmop{AF}} + Q} (V_{\tilde{g}}' \otimes
\ldots \otimes V_{g_Q} \otimes \ldots \otimes V_{g_Q}) $ it holds that
\begin{align}\label{Eq: generic contribution J_n 3}
     \mathcal{T}_{\ell}^{\hbar \Delta_{\tmop{AF}} + Q} (V_{\tilde{g}}' \otimes
     \ldots \otimes V_{g_Q} \otimes \ldots \otimes V_{g_Q})  = &
     \frac{1}{2^{\ell} i} \sum_{k = 0}^{\ell - 1} \binom{\ell - 1}{k}
     \mathcal{T}_{\ell}^{\hbar \Delta_{\tmop{AF}} + Q} ( (V_{a, \tilde{g}} -
     V_{- a, \tilde{g}}) \otimes V_{a, g_Q}^{\otimes k} \otimes V_{- a,
     g_Q}^{\otimes (\ell - 1 - k)})\nonumber\\
     = & \frac{1}{2^{\ell} i} \sum_{k = 0}^{\ell - 1} \binom{\ell - 1}{k}
     [\mathcal{T}_{\ell}^{\hbar \Delta_{\tmop{AF}} + Q} (V_{a, \tilde{g}}
     \otimes V_{a, g_Q}^{\otimes k} \otimes V_{- a, g_Q}^{\otimes (\ell - 1 -
     k)}) + \nobracket\nonumber\\
      & \qquad \qquad \qquad \qquad - \nobracket \mathcal{T}_{\ell}^{\hbar
     \Delta_{\tmop{AF}} + Q} (V_{- a, \tilde{g}} \otimes V_{a, g_Q}^{\otimes
     k} \otimes V_{- a, g_Q}^{\otimes (\ell - 1 - k)})] .
\end{align}
Considering separately the two terms on the right hand side of Equation \eqref{Eq: generic contribution J_n 3}, the first one reads
\begin{align*}
     \mathcal{T}_{\ell}^{\hbar \Delta_{\tmop{AF}} + Q} (V_{a, \tilde{g}}
     \otimes V_{a, g_Q}^{\otimes k} \otimes V_{- a, g_Q}^{\otimes (\ell - 1 -
     k)}&) =  \mathcal{T}_{\ell}^{\hbar \Delta_{\tmop{AF}} + Q} \left[
     \int_{\mathbb{R}^{2\ell}} e^{ia \varphi (x_1) + i \sum_{j = 2}^{\ell}
     a^{(k)}_j \varphi (x_j)} \tilde{g} (x_1) \ldots g_Q (x_{\ell}) \mathd
     \mu_{X_{\ell}} \right]\\
     &=  \int_{\mathbb{R}^{2\ell}} \mathd \mu_{X_{\ell}} \tilde{G}_Q
     (X_{\ell}) e^{ i \left[ a \varphi (x_1) +\sum_{j = 2}^{\ell} a_j^{(k)}
     \varphi (x_j) \right]} \times\\
      & \times e^{- \sum_{2 \leqslant j \leqslant \ell} a  a_j^{(k)} (Q +
     \hbar \Delta_{\tmop{AF}}) (x_1, x_j)} e^{- \sum_{2 \leqslant m < j
     \leqslant \ell} a_m^{(k)} a_j^{(k)} (Q + \hbar \Delta_{\tmop{AF}}) (x_m,
     x_j)},
\end{align*}
where $\mathd \mu_{X_{\ell}} \assign \mathd \mu_{x_1} \ldots \mathd
\mu_{x_{\ell}}$, $\tilde{G}_1 (X_{\ell}) \assign \tilde{g} (x_1) \ldots g_Q
(x_{\ell})$ and where we adopted the notation
\begin{align*}
     a^{(k)}  =  \{ a_2^{(k)}, \ldots, a_{\ell}^{(k)} \}  =  \{
     \underbrace{a, \ldots, a}_k, \underbrace{- a, \ldots, - a}_{\ell-1-k} \}.
\end{align*}
Focusing on the second term in Equation \eqref{Eq: generic contribution J_n 3}, we obtain
\begin{align*}
     \mathcal{T}_{\ell}^{\hbar \Delta_{\tmop{AF}} + Q} (V_{- a, \tilde{g}}
     \otimes V_{a, g_Q}^{\otimes k} \otimes V_{- a, g_Q}^{\otimes (\ell - 1 -
     k)})  &=  \mathcal{T}_{\ell}^{\hbar \Delta_{\tmop{AF}} + Q} \left[
     \int_{\mathbb{R}^{2\ell}} e^{- ia \varphi (x_1) + i \sum_{j = 2}^{\ell}
     a_j^{(k)} \varphi (x_j)} \tilde{g} (x_1) \ldots g_Q (x_{\ell}) \mathd
     \mu_{X_{\ell}} \right]\\
     &=  \int_{\mathbb{R}^{2\ell}} \mathd \mu_{X_{\ell}} \tilde{G}_Q
     (X_{\ell}) e^{- ia \varphi (x_1)+i\sum_{j = 2}^{\ell} a_j^{(k)}
     \varphi (x_j)} \times\\
     & \times e^{\sum_{2 \leqslant j \leqslant \ell} a  a_j^{(k)}
     (Q + \hbar \Delta_{\tmop{AF}}) (x_1, x_j)} e^{- \sum_{2 \leqslant m < j
     \leqslant \ell} a_m^{(k)} a_j^{(k)} (Q + \hbar \Delta_{\tmop{AF}}) (x_m,
     x_j)}.
\end{align*}
As a last step we consider the $\star_{Q + \hbar \omega}$-product of the last two identities computed above. This amounts to \begin{align}\label{Eq:two-terms-int-field}
\nonumber
    \mathcal{T}_{\ell}^{\hbar \Delta_{\tmop{AF}} + Q} (V_{g_Q} \otimes \ldots
    \otimes V_{\tilde{g}}' \otimes \ldots \otimes V_{g_Q}) \star_{Q + \hbar
    \omega} \mathcal{T}_{n - \ell}^{\hbar \Delta_F + Q} (V_g \otimes \ldots
    \otimes V_g) = \frac{1}{2^n i} \sum_{k = 0}^{\ell - 1} \sum_{p = 0}^{n -
    \ell} \binom{\ell - 1}{k} \binom{n - \ell}{p} \times\\
    \nonumber
 \times
    [\mathcal{T}_{\ell}^{\hbar \Delta_{\tmop{AF}} + Q} (V_{a, \tilde{g}}
    \otimes V_{a, g_Q}^{\otimes k} \otimes V_{- a, g_Q}^{\otimes (\ell - 1 -
    k)}) \star_{Q + \hbar \omega} \mathcal{T}_{n - \ell}^{\hbar \Delta_F + Q}
    (V_{a, g_Q}^{\otimes p} \otimes V_{- a, g_Q}^{\otimes (n - \ell - p)}) +
    \nobracket\\
    \left.  
    -\mathcal{T}_{\ell}^{\hbar \Delta_{\tmop{AF}} + Q} (V_{- a, \tilde{g}}
    \otimes V_{a, g_Q}^{\otimes k} \otimes V_{- a, g_Q}^{\otimes (\ell - 1 -
    k)}) \star_{Q + \hbar \omega} \mathcal{T}_{n - \ell}^{\hbar \Delta_F + Q}
    (V_{a, g_Q}^{\otimes p} \otimes V_{- a, g_Q}^{\otimes (n - \ell - p)})
    \right] . 
\end{align}
From the point of view of convergence, the reader can easily convince him/herself that the two terms in Equation~\eqref{Eq:two-terms-int-field} behave analogously. Thus we shall focus only on one of them, the convergence of the other following suit. A direct application of Equation \eqref{Eq:deform-quantisation} entails that
\begin{align}\label{Eq:two-terms-int-field-2}
     \mathcal{T}_{\ell}^{\hbar \Delta_{\tmop{AF}} + Q} (V_{a, \tilde{g}}
     \otimes V_{a, g_Q}^{\otimes k} &\otimes V_{- a, g_Q}^{\otimes (\ell - 1 -
     k)}) \star_{Q + \hbar \omega} \mathcal{T}_{n - \ell}^{\hbar \Delta_F + Q}
     (V_{a, g_Q}^{\otimes p} \otimes V_{- a, g_Q}^{\otimes (n - \ell - p)})
     =\nonumber\\
     = \int_{\mathbb{R}^{2n}} \mathd \mu_{X_n} &e^{ia \varphi (x_1) + i \sum_{j =
     2}^{\ell} a_j^{(k)} \varphi (x_j)} e^{i \sum_{j = \ell + 1}^n a^{(p)}_j
     \varphi (x_j)} \tilde{G}_Q (X_n) e^{- \sum_{2 \leqslant j \leqslant \ell}
     a  a_j^{(k)} (Q + \hbar \Delta_{\tmop{AF}}) (x_1, x_j)}\times\nonumber\\\qquad&\times e^{- \sum_{2
     \leqslant m < j \leqslant \ell} a_m^{(k)} a_j^{(k)} (Q + \hbar
     \Delta_{\tmop{AF}}) (x_m, x_j)}
     e^{- \sum_{\ell + 1 \leqslant m < j \leqslant
     n} a^{(p)}_m a^{(p)}_j (Q + \hbar \Delta_F) (x_m, x_j)} \nonumber\\
     &\times e^{- \sum_{2
     \leqslant m \leqslant \ell, 1 \leqslant j \leqslant n - \ell} a_m^{(k)}
     a^{(p)}_j (Q + \hbar \omega) (x_m, x_{\ell + j})} e^{- \sum_{1 \leqslant
     j \leqslant n - \ell} aa^{(p)}_j (Q + \hbar \omega) (x_1, x_{\ell + j})}
     .
\end{align}
We shall estimate the absolute value of the expression in Equation \eqref{Eq:two-terms-int-field-2} as follows
\begin{align}\label{Eq: estimate absolute value}
     |\mathcal{T}_{\ell}^{\hbar \Delta_{\tmop{AF}} + Q}& (V_{a, \tilde{g}}
     \otimes V_{a, g_Q}^{\otimes k} \otimes V_{- a, g_Q}^{\otimes (\ell - 1 -
     k)}) \star_{Q + \hbar \omega} \mathcal{T}_{n - \ell}^{\hbar \Delta_F + Q}
     (V_{a, g_Q}^{\otimes p} \otimes V_{- a, g_Q}^{\otimes (n - \ell - p)}) |
     \leqslant\nonumber\\
     \leqslant \int_{\mathbb{R}^{2n}} &\mathd \mu_{X_n} | \tilde{G} (X_n) | e^{-
     \sum_{2 \leqslant j \leqslant \ell} a  a_j^{(k)} \tmop{Re} (Q + \hbar
     \Delta_{\tmop{AF}}) (x_1, x_j)} e^{- \sum_{2 \leqslant u < j \leqslant
     \ell} a_u^{(k)} a_j^{(k)} \tmop{Re} (Q + \hbar \Delta_{\tmop{AF}}) (x_u,
     x_j)} \times\nonumber\\
    \times& e^{- \sum_{\ell + 1 \leqslant u < j \leqslant n} a^{(p)}_u
     a^{(p)}_j \tmop{Re} (Q + \hbar \Delta_F) (x_u, x_j)} e^{- \sum_{2
     \leqslant u \leqslant \ell, 1 \leqslant j \leqslant n - \ell} a_u^{(k)}
     a^{(p)}_j \tmop{Re} (Q + \hbar \omega) (x_u, x_{\ell + j})} \times\nonumber\\
    \times&e^{- \sum_{1
     \leqslant j \leqslant n - \ell} aa^{(p)}_j \tmop{Re} (Q + \hbar \omega)
     (x_1, x_{\ell + j})} .
\end{align}
Recalling that $Q$ is real-valued and that
\[ \tmop{Re} (\omega) \backassign H \qquad \tmop{Re} (\Delta_F) = \tmop{Re}
   \left( \frac{i}{2} (\Delta^R + \Delta^A) + H \right) = H, \qquad
   \tmop{Re} (\Delta_{\tmop{AF}}) = \tmop{Re} (\omega - i \Delta^R) =
   \tmop{Re} (\omega) = H, \]
Equation \eqref{Eq: estimate absolute value} can be further improved as follows
\begin{align*}
     &|\mathcal{T}_{\ell}^{\hbar \Delta_{\tmop{AF}} + Q} (V_{a, \tilde{g}}
     \otimes V_{a, g_Q}^{\otimes k} \otimes V_{- a, g_Q}^{\otimes (\ell - 1 -
     k)}) \cdot_{Q + \hbar \omega} \mathcal{T}_{n - \ell}^{\hbar \Delta_F + Q}
     (V_{a, g_Q}^{\otimes p} \otimes V_{- a, g_Q}^{\otimes (n - \ell - p)}) |
     \leqslant\\
      &\leqslant \int_{\mathbb{M}^{\ell}} \mathd \mu_{X_n} | \tilde{G}
     (X_n) | e^{- \sum_{2 \leqslant j \leqslant \ell} a  a_j^{(k)} (Q + \hbar
     H) (x_1, x_j)} e^{- \sum_{2 \leqslant u < j \leqslant \ell} a_u^{(k)}
     a_j^{(k)} (Q + \hbar H) (x_u, x_j)} \times\\
  &\times e^{- \sum_{\ell + 1 \leqslant u < j \leqslant n} a^{(p)}_u
     a^{(p)}_j (Q + \hbar H) (x_u, x_j)} e^{- \sum_{2 \leqslant u \leqslant
     \ell, 1 \leqslant j \leqslant n - \ell} a_u^{(k)} a^{(p)}_j (Q + \hbar
     H) (x_u, x_{\ell + j})} e^{- \sum_{1 \leqslant j \leqslant n - \ell}
     aa^{(p)}_j (Q + \hbar H) (x_1, x_{\ell + j})}\\
     & = \tmop{ev}_0 (\mathcal{T}_n^{\hbar H + Q} (V_{a, | \tilde{g} |}
     \otimes V_{a, g}^{\otimes k} \otimes V_{- a, g}^{\otimes (\ell - 1 - k)}
     \otimes V_{a, g}^{\otimes p} \otimes V_{- a, g}^{\otimes (n - \ell -
     p)})) .
\end{align*}
Considering the second term in Equation~\eqref{Eq:two-terms-int-field}, an
analogous procedure entails
\begin{align*}
     | \mathcal{T}_{\ell}^{\hbar \Delta_{\tmop{AF}} + Q} (V_{- a, \tilde{g}}
     \otimes V_{a, g_Q}^{\otimes k} &\otimes V_{- a, g_Q}^{\otimes (\ell - 1 -
     k)}) \star_{Q + \hbar \omega} \mathcal{T}_{n - \ell}^{\hbar \Delta_F + Q}
     (V_{a, g_Q}^{\otimes p} \otimes V_{- a, g_Q}^{\otimes (n - \ell - p)}) |
     \leqslant\\
      &\leqslant \tmop{ev}_0 (\mathcal{T}_n^{\hbar H + Q} (V_{-
     a, | \tilde{g} |} \otimes V_{a, g}^{\otimes k} \otimes V_{- a,
     g}^{\otimes (\ell - 1 - k)} \otimes V_{a, g}^{\otimes p} \otimes V_{- a,
     g}^{\otimes (n - \ell - p)})) .
\end{align*}
Using Equation \eqref{Eq:two-terms-int-field} as well as the estimates above, we obtain
\begin{align*}
     | \mathcal{T}_{\ell}^{\hbar \Delta_{\tmop{AF}} + Q} (V_{g_Q} \otimes
     \ldots \otimes V_{\tilde{g}}' \otimes \ldots \otimes V_{g_Q})& \star_{Q +
     \hbar \omega} \mathcal{T}_{n - \ell}^{\hbar \Delta_F + Q} (V_{g_Q}
     \otimes \ldots \otimes V_{g_Q}) | \leqslant \frac{1}{2^n} \sum_{k =
     0}^{\ell - 1} \sum_{p = 0}^{n - \ell} \binom{\ell - 1}{k} \binom{n -
     \ell}{p} \times\\
  &\times [\tmop{ev}_0
     (\mathcal{T}_n^{\hbar H + Q} (V_{a, | \tilde{g} |} \otimes V_{a,
     g}^{\otimes k} \otimes V_{- a, g}^{\otimes (\ell - 1 - k)} \otimes V_{a,
     g}^{\otimes p} \otimes V_{- a, g}^{\otimes (n - \ell - p)})) +
     \nobracket\\
     &\left.  + \tmop{ev}_0
     (\mathcal{T}_n^{\hbar H + Q} (V_{- a, | \tilde{g} |} \otimes V_{a,
     g}^{\otimes k} \otimes V_{- a, g}^{\otimes (\ell - 1 - k)} \otimes V_{a,
     g}^{\otimes p} \otimes V_{- a, g}^{\otimes (n - \ell - p)})) \right]\\
 &= \frac{1}{2} [\tmop{ev}_0
 (\mathcal{T}_n^{\hbar H + Q} (V_{a, | \tilde{g} |} \otimes V_g^{\otimes
 	n - 1})) + \tmop{ev}_0 (\mathcal{T}_n^{\hbar H + Q} (V_{- a, |
 	\tilde{g} |} \otimes V_g^{\otimes n - 1}))].
\end{align*}
Eventually, using Equation~\eqref{Eq: nth order J} we obtain
\begin{align}\label{Eq:bound-for-Jn}
\nonumber
    | J_n (V_g^{\otimes n}, \Phi_f) | & \leqslant  \frac{1}{2} \sum_{\ell =
    0}^n \binom{n}{\ell} \ell [\tmop{ev}_0 (\mathcal{T}_n^{\hbar H + Q}
    (V_{a, | \tilde{g} |} \otimes V_g^{\otimes n - 1})) + \tmop{ev}_0
    (\mathcal{T}_n^{\hbar H + Q} (V_{- a, | \tilde{g} |} \otimes
    V_g^{\otimes n - 1}))]\\
    & =  n 2^{n - 2} [\tmop{ev}_0 (\mathcal{T}_n^{\hbar H + Q} (V_{a, |
    \tilde{g} |} \otimes V_g^{\otimes n - 1})) + \tmop{ev}_0
    (\mathcal{T}_n^{\hbar H + Q} (V_{- a, | \tilde{g} |} \otimes
    V_g^{\otimes n - 1}))] .
\end{align}
With this estimate we are in a position to prove a bound strictly related to the absolute convergence of the series defining $J(V_g,\Phi_f)$. 
\begin{theorem}\label{Thm:bound-Jn}
  Under the assumptions of Theorem~\ref{Thm:estimates-Q-S-matrix-I} and recalling $J_n (V_g^{\otimes n}, \Phi_f)$ as per Equation \eqref{Eq: nth order J}, the following bound holds true:
  \[ | J_n (V_g^{\otimes n}, \Phi_f) | \leqslant \frac{n 2^n (2 \mu)^{n
     \alpha} (C_Q)^{n^2}}{2 (n!)^{1 - 1 / p}} \left( \frac{2 \lambda e^{2^{-
     1} a^2 K}}{\hbar} \right)^n \|g\|^{n - 1}_{L^q} \| \tilde{g} \|_{L^q}
     C^{n / p} . \]
\end{theorem}
\begin{proof}
  The proof of this result is closely related to the one of Theorem~\ref{Thm:estimates-Q-S-matrix-I}. Since we are going to adopt the same strategy, we discuss only the main steps of the proof. On account of the analysis performed in this section, in particular of Equation~\eqref{Eq:bound-for-Jn}, it suffices to exhibit a
  suitable bound for $\tmop{ev}_0 (\mathcal{T}_n^{\hbar H + Q} (V_{a, |
  \tilde{g} |} \otimes V_g^{\otimes n - 1}))$ and for $\tmop{ev}_0
  (\mathcal{T}_n^{\hbar H + Q} (V_{- a, | \tilde{g} |} \otimes V_g^{\otimes
  n - 1}))$. Since the analysis of these two contributions is analogous, we limit ourselves to discussing the first one, as an identical argument applies to the other.
  
 We observe that, as in the proof of Theorem~\ref{Thm:estimates-Q-S-matrix-I}, the conditioning and inverse conditioning results stated in Theorem~\ref{Thm:conditioning} can be applied, guaranteeing the validity of the following inequality:
 \begin{equation}\label{Eq: conditioning interacting field}
 	\tmop{ev}_0 (\mathcal{T}_n^{\hbar H + Q} (V_{a, | \tilde{g} |} \otimes
 	V_g^{\otimes n - 1})) \leqslant 2 \tmop{ev}_0 \left\{
 	\mathcal{T}_n^{\hbar H + Q} \left[ \left( \frac{2 \lambda e^{2^{- 1}
 			a^2 K}}{\hbar} V_{a, | \tilde{g} |} \right) \otimes \left( \frac{2
 		\lambda e^{2^{- 1} a^2 K}}{\hbar} V_g  \right)^{\otimes n - 1} \right]
 	\right\} . 
 \end{equation} 
As a matter of fact, Equation \eqref{Eq: conditioning interacting field} suggests that controlling the massless case entails an analogous control of the massive counterpart. Applying {\tmem{mutatis mutandis}} Lemma~\ref{Lem:Cauchy-determinant}, the sought after estimate descends. This concludes the proof. 
\end{proof}

Starting from Equation \eqref{Eq: nth order M} and iterating the analysis applied to $J_n(V_g^{\otimes n},\Phi_f)$, an analogous expression can be obtained for $M_n$, the only
difference being the explicit form of the effective test-function  $\tilde{\tilde{g}}$, which is defined in this case in terms of the Feynman propagator as
\begin{equation}
  \tilde{\tilde{g}} (y) \assign g_Q (y) [(Q + \hbar \Delta_F) f] (y).
  \label{Eq:new-test-funct-for-M}
\end{equation}
Without going through the detailed procedure once more, we directly state the resulting estimate on the absolute value of $M_n(V_g^{\otimes n},\Phi_f)$:
\begin{align*}
     | M_n (V_g^{\otimes n}, \Phi_f) | & \leqslant  \frac{1}{2} \sum_{\ell =
     0}^n \binom{n}{\ell} (n - \ell) [\tmop{ev}_0 (\mathcal{T}_n^{\hbar H +
     Q} (V_{a, | \tilde{\tilde{g}} |} \otimes V_g^{\otimes n - 1})) +
     \tmop{ev}_0 (\mathcal{T}_n^{\hbar H + Q} (V_{- a, |
     \tilde{\tilde{g}} |} \otimes V_g^{\otimes n - 1}))]\\
     & =  n 2^{n - 2} [\tmop{ev}_0 (\mathcal{T}_n^{\hbar H + Q} (V_{a, |
     \tilde{\tilde{g}} |} \otimes V_g^{\otimes n - 1})) + \tmop{ev}_0
     (\mathcal{T}_n^{\hbar H + Q} (V_{- a, | \tilde{\tilde{g}} |}
     \otimes V_g^{\otimes n - 1}))],
\end{align*}
where the second line descends from the identity
\[ \sum_{\ell = 0}^n \binom{n}{\ell} (n - \ell) = n 2^{n - 1} . \]
\begin{theorem}\label{Thm::bound-Mn}
Under the assumptions of Theorem~\ref{Thm:estimates-Q-S-matrix-I} and recalling $M_n (V_g^{\otimes n}, \Phi_f)$ as per Equation \eqref{Eq: nth order M}, the following bound is satisfied:
  \[ | M_n (V_g^{\otimes n}, \Phi_f) | \leqslant \frac{n 2^n (2 \mu)^{n
     \alpha} (C_Q)^{n^2}}{2 (n!)^{1 - 1 / p}} \left( \frac{2 \lambda e^{2^{-
     1} a^2 K}}{\hbar} \right)^n \|g\|^{n - 1}_{L^q} \| \tilde{\tilde{g}}
     \|_{L^q} C^{n / p} . \]
\end{theorem}
From Theorems~\ref{Thm:bound-Jn} and~\ref{Thm::bound-Mn}, the following corollary follows.
\begin{corollary}
Under the assumptions of Theorem~\ref{Thm:expansion-of-the-interacting-field},
  the power series defining the interacting field $\Phi_{I, f}(\varphi)$ as per Equation \eqref{Eq: expansion interacting 2} is absolutely
  convergent, for all field configurations $\varphi \in \mathcal{E}(\mathbb{R}^2)$.
\end{corollary}
\begin{proof}
	In complete analogy with the proof of Corollary \ref{Cor: convergence S matrix}, convergence is a direct consequence of the estimates in Theorems \ref{Thm:bound-Jn} and \ref{Thm::bound-Mn}, combined with the decomposition in Equation\eqref{Eq: expansion interacting 2}.
\end{proof}
\subsubsection{Convergence of the $n$-point correlation functions}\label{Sec:convergence-corr-functions}
In this section we shall extend the convergence result proven in Section~\ref{Sec:conv-interacting-field} at the level of the $n$-point correlation functions of the interacting field. To wit, in this section we shall prove convergence of the power series defining the interacting counterpart of algebra elements of the form
$$ \Phi_{f_1} \ldots \Phi_{f_p}, \quad 
 p \in \mathbb{N}, \quad 
   f_i \in \mathcal{D} (\mathbb{R}^2), \quad \forall i \in \{ 1, \ldots, p \}.$$
Throughout this section we adopt the notation $\Phi_1 \ldots \Phi_p \assign \Phi_{f_1}\ldots \Phi_{f_p}$. On account of Equation~\eqref{Eq:Bog-Map}, the $Q$-deformed interacting version of the product of fields $F = \Phi_1 \ldots \Phi_p$ reads
\[ \Gamma_Q [R_{\lambda V} (\Phi_1 \ldots \Phi_p)] = \Gamma_Q ((S (\lambda
   V_g))^{\star_{\hbar \omega} - 1}) \star_{Q + \hbar \omega} [(\Gamma_Q (S
   (\lambda V_g)) \star_{Q + \hbar \Delta_F} \Gamma_Q (\Phi_1 \ldots \Phi_p))]
   . \]
As a first step let us observe that, splitting the deformed product $\star_{Q + \hbar \Delta_F}$ into the sum between the pointwise counterpart and higher order contracted contributions, we obtain
\[ \begin{array}{lll}
     \Gamma_Q (S (\lambda V_g)) \star_{Q + \hbar \Delta_F} \Gamma_Q (\Phi_1
     \ldots \Phi_p) & = & \Gamma_Q (S (\lambda V_g)) \Gamma_Q (\Phi_1 \ldots
     \Phi_p) +\\
     &  & + \sum_{n = 1}^p \langle \Gamma_Q (S (\lambda V_g))^{(n)}, (Q +
     \hbar \Delta_F)^{\otimes n} \Gamma_Q (\Phi_1 \ldots \Phi_p)^{(n)} \rangle.
   \end{array} \]
As a consequence,
\begin{align}\label{Eq:decomp-correlation}
\nonumber
    \Gamma_Q [R_{\lambda V} (\Phi_1 \ldots \Phi_p)] & =  \Gamma_Q ((S
    (\lambda V_g))^{\star_{\hbar \omega} - 1}) \star_{Q + \hbar \omega}
    [\Gamma_Q (S (\lambda V_g)) \Gamma_Q (\Phi_1 \ldots \Phi_p)]\\
    \nonumber
    &   + \sum_{n = 1}^p \Gamma_Q ((S (\lambda V_g))^{\star_{\hbar \omega} -
    1}) \star_{Q + \hbar \omega} \langle \Gamma_Q (S (\lambda V_g))^{(n)}, (Q
    + \hbar \Delta_F)^{\otimes n} \Gamma_Q (\Phi_1 \ldots \Phi_p)^{(n)}
    \rangle\\
    & \backassign  J (V_g, \Phi_1 \ldots \Phi_p) + \sum_{n = 1}^p M^n (V_g,
    \Phi_1 \ldots \Phi_p),
\end{align}
where the arguments of $J$ and $M^n$ remind us of the specific functional under scrutiny. Mirroring the analysis of the interacting field, let us start working with $J (V_g, \Phi_1 \ldots \Phi_p)$, namely
\begin{equation}
  J (V_g, \Phi_1 \ldots \Phi_p) = \Gamma_Q ((S (\lambda V_g))^{\star_{\hbar
  \omega} - 1}) \star_{Q + \hbar \omega} [\Gamma_Q (S (\lambda V_g)) \Gamma_Q
  (\Phi_1 \ldots \Phi_p)] . \label{Eq:relevant-corr}
\end{equation}
Note that, setting $I = \{ 1, \ldots, p \}$, $I_1 = I
\setminus \{ i_1, j_1 \}$ and, iteratively,  $I_{\ell} = I_{\ell - 1} \setminus
\{ i_{\ell}, j_{\ell} \}$, the deformed product $\Gamma_Q (\Phi_1 \ldots \Phi_p)$ can be written as
\begin{align}\label{Eq:Wick}
\nonumber
    \Gamma_Q (\Phi_1 \ldots \Phi_p) & =  \Phi_1 \ldots \Phi_p +\\
    \nonumber
    &   + \sum_{i_1 < j_1 ; i_1, j_1 \in I} Q (f_{i_1}, f_{j_1}) \left[
    \prod_{k_1 \in I_1} \Phi_{k_1} + \sum_{i_2 < j_2 ; i_2, j_2 \in I_1} Q
    (f_{i_2}, f_{j_2}) \left[ \prod_{k_2 \in I_2} \Phi_{k_2} + \ldots +
    \right. \right.\\
    &  \qquad \qquad \qquad \qquad + \left. \left. \sum_{i_{\lfloor p / 2
    \rfloor} < j_{\lfloor p / 2 \rfloor} ; i_{\lfloor p / 2 \rfloor},
    j_{\lfloor p / 2 \rfloor} \in I_{\lfloor p / 2 \rfloor}} Q (f_{i_{\lfloor
    p / 2 \rfloor}}, f_{j_{\lfloor p / 2 \rfloor}}) \Phi_{\lfloor p / 2
    \rfloor} \right] \right],
\end{align}
which is nothing but Wick's theorem adapted to the map $\Gamma_Q$, see Equation \eqref{Eq:GammaQ}. In Equation \eqref{Eq:Wick}, with the notation $\Phi_{\lfloor p / 2 \rfloor} = 1$ we mean that, whenever $p$ is even, then $\Phi_{\lfloor p / 2 \rfloor} = 1$.
	At this level, Equations~\eqref{Eq:relevant-corr} and~\eqref{Eq:Wick} suggest that, in order to prove the convergence of the
	series defining $J (V_g, \Phi_1 \ldots \Phi_p)$, it suffices to check convergence of the products
	\begin{equation}
		\Gamma_Q ((S (\lambda V_g))^{\star_{\hbar \omega} - 1}) \star_{Q + \hbar
			\omega} [\Gamma_Q (S (\lambda V_g)) \Phi_1 \ldots \Phi_m],
		\label{Eq:form-to-handle}
	\end{equation}
	for $m \leqslant p$. Indeed, this statement stems from replacing Equation~\eqref{Eq:Wick} into Equation~\eqref{Eq:relevant-corr},
	observing that the factors $Q (f_{i_{\ell}}, f_{j_{\ell}})$ are constant
	functionals. Hence they are not affected by the contractions defining the product $\star_{Q + \hbar\omega}$. 
	
To better grasp the underlying rationale, an example is in due order.
\begin{example}
  The simplest contribution to Equation \eqref{Eq:Wick} is the one where a single contraction occurs, namely
  \begin{align*}
  \sum_{i_1 < j_1 ; i_1, j_1 \in I} &\Gamma_Q ((S (\lambda
  V_g))^{\star_{\hbar \omega} - 1}) \star_{Q + \hbar \omega} \left[ Q
  (f_{i_1}, f_{j_1}) \prod_{k_1 \in I_1} \Phi_{k_1} \right]\\
  &=\sum_{i_1 < j_1 ; i_1, j_1 \in I} Q (f_{i_1}, f_{j_1}) \left[ \Gamma_Q
  ((S (\lambda V_g))^{\star_{\hbar \omega} - 1}) \star_{Q + \hbar \omega}
  \prod_{k_1 \in I_1} \Phi_{k_1} \right].
  \end{align*}
  The only term to be tamed is the one under square brackets which falls in the class of those considered in Equation~\eqref{Eq:form-to-handle}.
\end{example}
As a consequence, the analysis of $J (V_g, \Phi_1 \ldots \Phi_p)$ boils down to that of terms of the form
\begin{equation}
  \Gamma_Q ((S (\lambda V_g))^{\star_{\hbar \omega} - 1}) \star_{Q + \hbar
  \omega} [\Gamma_Q (S (\lambda V_g)) \Phi_1 \ldots \Phi_m],
  \label{Eq:case-to-study}
\end{equation}
for $m \leqslant p$. Observe that $J (V_g, \Phi_1 \ldots \Phi_p)$ is a finite sum of contributions of this kind, see Equations \eqref{Eq:relevant-corr} and \eqref{Eq:Wick}. Hence, we need a generalization of Lemma~\ref{Lemma:Leibniz} to deal with this scenario.
\begin{lemma}\label{Lem:Leibniz-pro}
Let $K$ be an integral kernel and let $A,B\in\mathcal{F}_{loc}$ be regular
  functionals, $m \in \mathbb{N}$ and $I = \{ 1, \ldots, m \}$. Then
  \begin{align}
  \nonumber
    A \star_K (B \Phi_1 \ldots \Phi_m) =& (A \star_K B) \Phi_1 \ldots \Phi_m \\+&
    \sum_{\ell = 1}^m \frac{1}{\ell !} \sum_{i_1, \ldots, i_{\ell} \in I, i_1
    < \ldots < i_{\ell}} \left[ \left\langle A^{(\ell)}, K^{\otimes \ell}
    \Phi_{i_1}^{(1)} {\ldots \Phi_{i_{\ell}}^{(1)}}  \right\rangle \star_K B
    \right] \prod_{j \in I \setminus \{ i_1, \ldots, i_{\ell} \}} \Phi_j,
    \label{Eq:Leibniz-pro}
  \end{align}
  where  \[ \prod_{j \in\emptyset} \Phi_j = 1. \]
\end{lemma}

\begin{proof}
  The statement descends directly from an iterative application of the Leibniz rule. It suffices to recall the identity
  \[ (B \Phi_1 \ldots \Phi_m)^{(n)} = \sum_{k = 0}^n \binom{n}{k} B^{(n - k)}
     (\Phi_1 \ldots \Phi_m)^{(k)}, \]
  observing that it is non-vanishing only for $k \leqslant m$. On account of the explicit expression of $\star_K$ as per Equation \eqref{Eq:deform-quantisation}, it holds that
\begin{align}\label{Eq: AB}
    A \star_K (B \Phi_1 \ldots \Phi_m)  = & \sum_{n \geqslant 0}
       \frac{1}{n!} \langle A^{(n)}, K^{\otimes n} (B \Phi_1 \ldots
       \Phi_m)^{(n)} \rangle\\
        = & \sum_{n \geqslant 0} \frac{1}{n!} \sum_{k = 0}^n \binom{n}{k}
       \langle A^{(n)}, K^{\otimes n} B^{(n - k)} (\Phi_1 \ldots \Phi_m)^{(k)}
       \rangle
\end{align} 
We observe that the case $k = 0$ in Equation \eqref{Eq: AB} coincides with the first contribution on the right hand side of Equation~\eqref{Eq:Leibniz-pro}, obtained when all the functional derivatives act on $B$. Considering instead a generic $k$, such tat $1 \leqslant k \leqslant m$, and exchanging the sums in Equation \eqref{Eq: AB} we write the inner series as
\begin{align*}
       \sum_{n \geqslant k} \frac{1}{n!} \binom{n}{k} \langle A^{(n)},
       K^{\otimes n} B^{(n - k)} &(\Phi_1 \ldots \Phi_m)^{(k)} \rangle =\\
       =  \sum_{n \geqslant k} \frac{1}{k ! (n - k) !}& \sum_{\substack{i_1,
       \ldots, i_{k} \in I\\ i_1 < \ldots < i_{k}}} \left\langle A^{(n)},
       K^{\otimes n} \left( \Phi_{i_1}^{(1)} {\ldots \Phi_{i_{k}}^{(1)}} 
       B^{(n - k)} \right) \right\rangle \prod_{j
       \in I \setminus \{ i_1, \ldots, i_{k} \}} \Phi_j\\
        =  \sum_{h \geqslant 0} \frac{1}{k !h!}& \sum_{\substack{i_1, \ldots,
       i_{k} \in I\\ i_1 < \ldots < i_{k}}} \left\langle A^{(h + k)},
       K^{\otimes h + k} \left( \Phi_{i_1}^{(1)} {\ldots
       \Phi_{i_{k}}^{(1)}}  B^{(h)} \right) \right\rangle\prod_{j
       \in I \setminus \{ i_1, \ldots, i_{k} \}} \Phi_j\\
        =  \frac{1}{k !} &\sum_{\substack{i_1, \ldots, i_{k} \in I\\ i_1 < \ldots
       < i_{k}}} \left[ \left\langle A^{(k)}, K^{\otimes k}
       \Phi_{i_1}^{(1)} {\ldots \Phi_{i_{k}}^{(1)}}  \right\rangle \star_K
       B \right] \prod_{j \in I \setminus \{ i_1, \ldots, i_{k} \}} \Phi_j.
  \end{align*}
  Summing over all values of $k$, the proof is concluded.
\end{proof}

We can apply Lemma~\ref{Lem:Leibniz-pro} directly to Equation~\eqref{Eq:case-to-study} by choosing $A = \Gamma_Q ((S (\lambda V_g))^{\star_{\hbar \omega} -
1})$, $B = \Gamma_Q (S (\lambda V_g))$ and $K = Q + \hbar \omega$. As a result, the deformed products of interest assume the form
\begin{align}
     &\Gamma_Q ((S (\lambda V_g))^{\star_{\hbar \omega} - 1}) \star_{Q + \hbar
     \omega} [\Gamma_Q (S (\lambda V_g)) \Phi_1 \ldots \Phi_m] =  \Phi_1
     \ldots \Phi_m+\notag \\+ &\sum_{\ell = 1}^m \frac{1}{\ell !} \sum\limits_{\substack{i_1, \ldots,
     i_{\ell} \in I\\ i_1 < \ldots < i_{\ell}}} \prod_{j \in I \setminus \{ i_1,
     \ldots, i_{\ell} \}} \Phi_j  \underbrace{\left\langle \Gamma_Q ((S (\lambda V_g))^{\star_{\hbar
     \omega} - 1})^{(\ell)}, (Q + \hbar \omega)^{\otimes \ell}
     \Phi_{i_1}^{(1)} {\ldots \Phi_{i_{\ell}}^{(1)}}  \right\rangle \star_{Q +
     \hbar \omega} \Gamma_Q S (\lambda V_g)}_{\ell-\textrm{th term}} .\label{Eq: Complicated}
\end{align}

\noindent Considering the $\ell$-th term in Equation \eqref{Eq: Complicated}, we can rewrite it as
\begin{align}\label{Eq:term-to-handle}
\nonumber
    \left\langle \Gamma_Q ((S (\lambda V_g))^{\star_{\hbar \omega} -
    1})^{(\ell)}, (Q + \hbar \omega)^{\otimes \ell} \Phi_{i_1}^{(1)} {\ldots
    \Phi_{i_{\ell}}^{(1)}}  \right\rangle \star_{Q + \hbar \omega} \Gamma_Q S
    (\lambda V_g) = \sum_{n \geqslant 0} \left( \frac{i \lambda}{\hbar}
    \right)^n \sum_{j = 0}^n \binom{n}{j} (- 1)^j \times\\
  \times \left\langle
    \mathcal{T}_j^{\hbar \Delta_{\tmop{AF}} + Q} (V_{g_Q}^{\otimes
    j})^{(\ell)}, (Q + \hbar \omega)^{\otimes \ell} \Phi_{i_1}^{(1)} {\ldots
    \Phi_{i_{\ell}}^{(1)}}  \right\rangle \star_{Q + \hbar \omega}
    \mathcal{T}_{n - j}^{\hbar \Delta_F + Q} (V_{g_Q}^{\otimes (n - j)}),
\end{align} 
where we have exploited Equation~\eqref{Eq:action-of-gamma-Q-on-interaction}. Observing that the $j$-th functional derivatives acting on $V_{g_Q}^{\otimes j}$
are distributed on the factors according to the Leibniz rule, the $\ell$-th derivative of the modified time ordered product of the vertex functionals reads
\begin{align*}
     &\mathcal{T}_j^{\hbar \Delta_{\tmop{AF}} + Q} (V_{g_Q}^{\otimes
     j})^{(\ell)}  =  \mathcal{T}_j^{\hbar \Delta_{\tmop{AF}} + Q} (V_{g_Q} 
     \otimes \ldots \otimes V_{g_Q})^{(\ell)}\\
     &=  \sum_{i_1 = 0}^{\ell} \ldots
     \sum_{i_{j - 1} = 0}^{\ell - i_1 - \ldots - i_{j - 2}} \binom{\ell}{i_1}
     \binom{\ell - i_1}{i_2} \ldots \binom{\ell - i_1 - \ldots - i_{j -
     2}}{i_{j - 1}} \mathcal{T}_j^{\hbar \Delta_{\tmop{AF}} + Q}
     (V_{g_Q}^{(i_1)} \otimes \ldots \otimes V_{g_Q}^{(\ell - i_1 - \ldots -
     i_{j - 1})}),
\end{align*}
giving rise to a finite number of contributions. These ought to be contracted with $(Q + \hbar \omega)^{\otimes \ell} \Phi_{i_1}^{(1)}
{\ldots \Phi_{i_{\ell}}^{(1)}}$, yielding
\begin{align}\label{Eq: final form to be estimated}
     \eqref{Eq:term-to-handle} = \sum_{n \geqslant 0} \left( \frac{i
     \lambda}{\hbar} \right)^n \sum_{\ell = 0}^n \binom{n}{j} (- 1)^j
     \sum_{i_1 = 0}^{\ell} \sum_{i_2 = 0}^{\ell - i_1} \ldots \sum_{i_{j - 1}
     = 0}^{\ell - i_1 - \ldots - i_{j - 2}} \binom{\ell}{i_1} \binom{\ell -
     i_1}{i_2} \ldots \binom{\ell - i_1 - \ldots - i_{j - 2}}{i_{j - 1}}
     \times\nonumber\\
     \times \left\langle
     \mathcal{T}_j^{\hbar \Delta_{\tmop{AF}} + Q} \left(
     V_{\tilde{g}_{i_1}}^{(i_1)} \otimes \ldots \otimes
     V_{\tilde{g}_{i_{\ell}}}^{(\ell - i_1 - \ldots - i_{j - 1})} \right)
     \right\rangle \star_{Q + \hbar \omega} \mathcal{T}_{n - j}^{\hbar
     \Delta_F + Q} (V_{g_Q}^{\otimes n - j}),
\end{align}
where, for simplicity of the notation, we introduced
\[ \tilde{g}_{i_k} (y) \assign g_Q (y) [(Q + \hbar \omega) f_{i_k}]^{^{i_k}}
   (y)\in\mathcal{D}(\mathbb{R}^2) \,. \]
\begin{remark}
   	 At this point, convergence of the series in Equation \eqref{Eq: final form to be estimated} descends by mirroring the analysis of the interacting field in Section \ref{Sec:conv-interacting-field}. This can be seen by noticing that, apart from an irrelevant coefficient, functional derivatives of $V_{g_Q} $ are either cosine or sine functionals. Hence we have to cope with a finite number of terms whose structure closely resembles the one of the interacting field.
\end{remark}
This concludes our investigation on the absolute convergence of the series
defining $J (V_g, \Phi_1 \ldots \Phi_p)$. As in the case of the interacting field, an identical procedure applies to the
analysis of $M^n$ introduced in Equation~\eqref{Eq:decomp-correlation}. Indeed, by means of the same
argument, discussed after Equation~\eqref{Eq:Wick}, one can restrict the attention to studying
\[ \Gamma_Q ((S (\lambda V_g))^{\star_{\hbar \omega} - 1}) \star_{Q + \hbar
   \omega} \langle \Gamma_Q (S (\lambda V_g))^{(n)}, (Q + \hbar
   \Delta_F)^{\otimes n} (\Phi_1 \ldots \Phi_p)^{(n)} \rangle, \]
for $n \leqslant p$. As highlighted by Equation~\eqref{Eq:new-test-funct-for-M},
the only difference with respect to the analysis of $J$ lies in the fact
that the modified test-functions are built out of the Feynman
propagator, namely
\[ \tilde{g}^F_{i_k} (y) \assign g_Q (y) [(Q + \hbar \mathLaplace_F)
   f_{i_k}]^{^{i_k}} (y)\in\mathcal{D}(\mathbb{R}^2) . \]
The content of this section culminates in the following result.
\begin{corollary}
  The power series defining the interacting
  observable $\Gamma_Q [R_{\lambda V} (\Phi_1 \ldots \Phi_p)]$ is absolutely
  convergent for any field configuration $\varphi \in \mathcal{E}(\mathbb{R}^2)$ and for any $p \in \mathbb{N}$.
\end{corollary}

\section{The classical limit $\hbar \rightarrow 0^+$}\label{Sec:classical-limit}
The previous section has been entirely devoted to showing absolute convergence of the formal power series defining the interacting field and the associated correlation functions within the perturbative quantum approach to the stochastic sine-Gordon model. Our next objective is to retrieve the information concerning the classical stochastic Sine-Gordon model via a limit procedure. In this endeavor we shall adapt to our setting the approach discussed in \cite{Dutsch-Fredenhagen-AQFT-perturb-th-loop-exp}.

As anticipated in Section~\ref{Sec:Strategy}, our ultimate goal is to compute
the expectation value of the solution of the stochastic sine-Gordon equation on the two-dimensional Minkowski space-time. As extensively discussed in \cite{BCDR23, BDR23,  DDR20} this is achieved by applying
the deformation map $\Gamma_Q$, whose effect is to codify the probabilistic information of the free stochastic equation at the level of the algebra of classical interacting observables. In turn, these can be seen as the classical limit of their quantum interacting counterparts. Having proven absolute convergence of the stochastic, quantum interacting field, existence of the classical limit would imply the absolute convergence of the formal power series representing the expectation value of the perturbative solution of the stochastic sine-Gordon model.

Before investigating classical limit of Equation~\eqref{Eq:Bog-Map}, one must check that such limit is well-defined. In the following we shall prove that the only non-vanishing contributions to the series defining Equation~\eqref{Eq:Bog-Map} involve non-negative powers of $\hbar$ ruling out possible blow-ups as $\hbar \rightarrow 0^+$. To keep the discussion as general as possible, in this section we shall consider arbitrary multi-local observables $F \in
\mathcal{F}^{\otimes m}_{\tmop{loc}}(\mathbb{R}^2)$, which encompasses the relevant examples of the interacting field and of its correlation functions.

As discussed in Section \ref{Sec:interacting-AQFT}, the starting point is the power series expansion in $\lambda>0$ of the quantum interacting
observable $F \in \mathcal{F}^{\otimes m}_{\tmop{loc}}(\mathbb{R}^2)$, $m \in
\mathbb{N}$,
\begin{equation}\label{Eq:expansion-in-ret-prod}
  R_{\lambda V_g} (F) = \sum_{n \geqslant 0} \frac{\lambda^n}{n!} R_{n, m}
  (V_g^{\otimes n}, F), 
\end{equation}
where $R_n (V_g^{\otimes n}, F)$ are the so-called {\tmem{retarded products}},
namely
\begin{equation}\label{Eq:retarded-products}
  R_{n, m} (V_g^{\otimes n}, F_f) = \left( \frac{i}{\hbar} \right)^n
  \sum_{\ell = 0}^n \binom{n}{\ell} (- 1)^{\ell} \mathcal{T}_{\ell}^{\hbar
  \Delta_{\tmop{AF}}} (V_g \otimes \ldots \otimes V_g) \star_{\hbar \omega}
  \mathcal{T}_{n - \ell, m}^{\hbar \Delta_F} (V_g \otimes \ldots \otimes V_g
  \otimes F_f), 
\end{equation}
where with $\mathcal{T}_{n - \ell, m}^{\hbar \Delta_F} (V_g
\otimes \ldots \otimes V_g \otimes F)$ we mean that the argument of
$\mathcal{T}_{n - \ell, m}^{\hbar \Delta_F}$ is given by $n - \ell$ copies of
$V_g$, so to keep track of the fact that $F$ depends on $m$
points. Accordingly, per consistency $f$ must lie in $\mathcal{D} (\mathbb{R}^{2m})$. In the following we prove that all contributions to Equation~\eqref{Eq:retarded-products} are of order $O(\hbar^0)$, which is tantamount to 
the existence and finiteness for any $n \geqslant 0$ of the limit
\[ \lim_{\hbar \rightarrow 0^+} R_{n, m} (V_g^{\otimes n}, F_f). \]
\begin{remark}
If we consider in the previous discussion the functionals codifying the interacting field and the $n$-point correlation functions, we can combine the existence of the limit as $\hbar\to 0$ to the proof that the underlying power series in the coupling constant $\lambda$ is absolutely convergent is a suitable regime. This two ingredients are the cornerstone of our construction of the solutions and of the correlation functions of the stochastic sine-Gordon equation on two-dimensional Minkowski spacetime.
\end{remark}
A central notion in this section is that of connected products. Hereafter, a connection is represented by a contraction of fields by means of a suitable integral kernel.

\begin{definition}\label{Def: connected products}
	Let $F_1, \ldots, F_n \in \mathcal{F}_{\tmop{loc}}(\mathbb{R}^2)$ and let $\star_{\hbar\omega}$ be the deformed product introduced in Equation \eqref{Eq:deform-quantisation}. In addition, for any $n>1$, we denote by $\sigma_P$ the collection of all partitions $P\equiv \bigcup_{j=1}^k P_j$ of the set $\{1,\ldots,n\}$ into $k$ disjoint, non-empty subsets $P_j$ with $1<k<n$.	We call connected product between the functionals $F_1, \ldots, F_n$  
\begin{equation}
(F_1 \star_{\hbar \omega} \ldots \star_{\hbar \omega} F_n)^c = F_1
   \star_{\hbar \omega} \ldots \star_{\hbar \omega} F_n - \sum_{| P |
   \geqslant 2} \prod_{p \in P} (F_{p_1} \star_{\hbar \omega} \ldots
   \star_{\hbar \omega} F_{p_{| p |}})^c,
\end{equation}
where $p = (p_1, \ldots, p_{| p |})$, $p_1 < \ldots < p_{| p |}$, the sum runs
over the set of all partitions $P$ of the set $\{ 1, \ldots, n \}$ containing
at least two sets and where $\prod$ denotes the classical pointwise product.
Similarly, given the time-ordering map $\mathcal{T}_n$ as per Equation \eqref{Eq: time-ordered product}, we define the connected time-ordered and anti-time-ordered products as
\[ \mathcal{T}_n^K (F_1 \otimes \ldots \otimes F_n)^c =\mathcal{T}_n^K (F_1
   \otimes \ldots \otimes F_n) - \sum_{| P | \geqslant 2} \prod_{p \in P}
   \mathcal{T}_{| p |}^K (F_{p_1} \otimes \ldots \otimes F_{p_{| p |}})^c, \]
   for $K = \hbar \Delta_F$ and $K = \hbar \Delta_{\tmop{AF}}$.
\end{definition}

As proven in {\cite[Prop.~1]{Dutsch-Fredenhagen-AQFT-perturb-th-loop-exp}} and in the following comments, whenever one considers local functionals $F_1, \ldots, F_n$ of order $\hbar^0$, then
\begin{align}
  &(F_1 \star_{\hbar \omega} \ldots \star_{\hbar \omega} F_n)^c =\mathcal{O}
  (\hbar^{n - 1}), \label{Eq:homogeneity-star}\\
  & \mathcal{T}_n^K (F_1 \otimes \ldots \otimes F_n)^c =\mathcal{O} (\hbar^{n -
  	1}), \label{Eq:homogeneity-T}
\end{align}
regardless whether $K = \hbar \Delta_F$ or $K = \hbar \Delta_{\tmop{AF}}$. Heuristically, this statement relies on the observation that a fully connected product of functionals accounts for at least $n - 1$ contractions, each of which carries a factor $\hbar$.
\begin{proposition}\label{Prop:classical-limit}
For any $F\in\mathcal{F}^{\otimes m}_{\textrm{loc}}(\mathbb{R}^2)$ it holds that
  \begin{enumerate}
    \item any non-vanishing contribution to  $R_{n, m} (V_g^{\otimes n}, F)$ in Equation \eqref{Eq:retarded-products} is such that each one among the $n$ functionals $V_g$ is contracted at least once with one of the entries of
    $F$;
    
    \item for any $n \geqslant0$,
    \[ R_{n, m} (V_g^{\otimes n}, F) =\mathcal{O} (\hbar^0) . \]
  \end{enumerate}
\end{proposition}

\begin{proof} 
  Starting from $1.$, observe that it suffices to prove that, if one of the $n$ factors $V_g$ in $R_{n, m} (V_g^{\otimes n}, F)$ is not contracted with all the other factors $V_g$ and with $F$, then its contribution to $R_{n, m} (V_g^{\otimes n}, F)$ vanishes. 
  Denoting the non-contracted vertex functional by $\tilde{V}_g$ to distinguish it from the others, Equation \eqref{Eq:retarded-products}, takes the form
\begin{align}\label{Eq: non-contracted vertex}
 R_{n, m} (V_g^{\otimes n}, F_f) = \left( \frac{i}{\hbar} \right)^n
  \sum_{\ell = 0}^n \sum_{(p_1,p_2)\in\mathcal{P}_{\ell,n-\ell}}(- 1)^{\ell} \mathcal{T}_{\ell}^{\hbar
  \Delta_{\tmop{AF}}} (V_{g}^{\otimes p_1}) \star_{\hbar \omega}
  \mathcal{T}_{n - \ell, m}^{\hbar \Delta_F} (V_{g}^{\otimes p_2}
  \otimes F_f), 
\end{align}
where $\mathcal{P}_{\ell,n-\ell}$ is the collection of partitions of the set $\{1,\ldots,n\}$ into subsets of size $\ell$ and $n-\ell$, respectively.
We stress that, if all the factors $V_g$ are identical, then Equation \eqref{Eq: non-contracted vertex} turns into Equation \eqref{Eq:retarded-products} due to the relation
$$
\sum_{(p_1,p_2)\in\mathcal{P}_{\ell,n-\ell}}=\binom{n}{\ell}.
$$
Nonetheless, the expression in Equation \eqref{Eq: non-contracted vertex} is more convenient since it allows to keep track of the fact that one of the factor $V_g$ is different from the other ones as it is the only one not connected with the other ones. 
For a fixed $n\geq 1$ and $0<\ell<n$, the case $n=0$ being trivial, we consider
\begin{align}\label{Eq:neq-combinatoric}
\sum_{(p_1,p_2)\in\mathcal{P}_{\ell,n-\ell}}(- 1)^{\ell} \mathcal{T}_{\ell}^{\hbar
  \Delta_{\tmop{AF}}} (V_{g}^{\otimes p_1}) \star_{\hbar \omega}
  \mathcal{T}_{n - \ell, m}^{\hbar \Delta_F} (V_{g}^{\otimes p_2}
  \otimes F_f).
\end{align}
Notice that there are precisely $\binom{n-1}{\ell-1}$ cases where the factor $\tilde{V}_g$ lies in the argument of $\mathcal{T}_{\ell}^{\hbar\Delta_{\tmop{AF}}}$ and
\begin{align*}
\binom{n}{\ell}-\binom{n-1}{\ell-1}=\binom{n-1}{\ell},
\end{align*}
cases where $\tilde{V}_g$ is acted upon by $\mathcal{T}_{\ell}^{\hbar\Delta_{\tmop{F}}}$. This combinatorial argument is direct consequence of the property of the $n-1$ factors $V_g$ of being indistinguishable. As a result it holds that
\begin{align*}
\eqref{Eq:neq-combinatoric}=(-1)^\ell\bigg[\binom{n-1}{\ell-1}\mathcal{T}_{\ell}^{\hbar
  \Delta_{\tmop{AF}}} (\tilde{V}_g\otimes &V_{g}^{\otimes(\ell-1)}) \star_{\hbar \omega}
  \mathcal{T}_{n - \ell, m}^{\hbar \Delta_F}(V_{g}^{\otimes(n-\ell)}
  \otimes F_f)\\
  &+\binom{n-1}{\ell}\mathcal{T}_{\ell}^{\hbar
  \Delta_{\tmop{AF}}} (V_{g}^{\otimes\ell}) \star_{\hbar \omega}
  \mathcal{T}_{n - \ell, m}^{\hbar \Delta_F} (\tilde{V}_g\otimes V_{g}^{\otimes(n-\ell-1)}
  \otimes F_f)\bigg]\\
=(-1)^\ell\tilde{V}_g\bigg[\binom{n-1}{\ell-1}\mathcal{T}_{\ell-1}^{\hbar
  \Delta_{\tmop{AF}}} (&V_{g}^{\otimes(\ell-1)}) \star_{\hbar \omega}
  \mathcal{T}_{n - \ell, m}^{\hbar \Delta_F}(V_{g}^{\otimes(n-\ell)}
  \otimes F_f)\\
  &+\binom{n-1}{\ell}\mathcal{T}_{\ell}^{\hbar
  \Delta_{\tmop{AF}}} (V_{g}^{\otimes\ell}) \star_{\hbar \omega}
  \mathcal{T}_{n - \ell-1, m}^{\hbar \Delta_F} (V_{g}^{\otimes(n-\ell-1)}
  \otimes F_f)\bigg],
\end{align*}
where we exploited that the (anti)time-ordered products are symmetric and that $\tilde{V}_g$ is not contracted with the other vertices, \textit{i.e.}, it multiplies them via the pointwise product.
Focusing once more on the retarded products, Equation \eqref{Eq: non-contracted vertex} reads
\begin{align}\label{Eq:telescopic-sum}
\nonumber
R_{n, m} (V_g^{\otimes n}, F_f) =\left(\frac{i}{\hbar}\right)^n&\tilde{V}_g
 \bigg[ \mathcal{T}_{n-1,m}^{\hbar \Delta_F}(V_{g}^{\otimes(n-1)}
  \otimes F_f)+(-1)^n\mathcal{T}_{n-1}^{\hbar
  \Delta_{\tmop{AF}}} (V_{g}^{\otimes(n-1)})\star_{\hbar\omega}
  \mathcal{T}_{0,m}^{\hbar \Delta_F}(F_f)\\\nonumber
  &\quad+\sum_{\ell = 1}^{n-1}\Bigl((-1)^\ell\binom{n-1}{\ell-1}\mathcal{T}_{\ell-1}^{\hbar\Delta_{\tmop{AF}}} (V_{g}^{\otimes(\ell-1)}) \star_{\hbar \omega}
  \mathcal{T}_{n - \ell, m}^{\hbar \Delta_F}(V_{g}^{\otimes(n-\ell)}
  \otimes F_f)\\
  &\qquad+(-1)^\ell\binom{n-1}{\ell}\mathcal{T}_{\ell}^{\hbar
  \Delta_{\tmop{AF}}} (V_{g}^{\otimes\ell}) \star_{\hbar \omega}
  \mathcal{T}_{n - \ell-1, m}^{\hbar \Delta_F} (V_{g}^{\otimes(n-\ell-1)}
  \otimes F_f)\Bigr) \bigg],
\end{align}
where, for later convenience, in the first line of Equation \eqref{Eq:telescopic-sum} we have isolated the terms due to $\ell=0$ and $\ell=n$. 
Observe that the sum in Equation \eqref{Eq:telescopic-sum} is telescopic: the contribution from $\ell=j$ in the first term under the sum is
$$
(-1)^j\binom{n-1}{j-1}\mathcal{T}_{j-1}^{\hbar\Delta_{\tmop{AF}}} (V_{g}^{\otimes(j-1)}) \star_{\hbar \omega}
  \mathcal{T}_{n - j, m}^{\hbar \Delta_F}(V_{g}^{\otimes(n-j)}
  \otimes F_f),
$$
while the case $\ell=j-1$ in the second term under the sum reads
$$
(-1)^{j-1}\binom{n-1}{j-1}\mathcal{T}_{j-1}^{\hbar\Delta_{\tmop{AF}}} (V_{g}^{\otimes(j-1)}) \star_{\hbar \omega}
  \mathcal{T}_{n - j, m}^{\hbar \Delta_F}(V_{g}^{\otimes(n-j)}
  \otimes F_f).
$$
A direct inspection entails that these two contributions cancel each other. As a consequence, only the case $\ell=1$ in the first term and the case $\ell=n-1$ in the second term survive yielding
\begin{align*}
R_{n, m} (V_g^{\otimes n}, F_f) &=\left(\frac{i}{\hbar}\right)^n\tilde{V}_g
 \bigg[ \mathcal{T}_{n-1,m}^{\hbar \Delta_F}(V_{g}^{\otimes(n-1)}
  \otimes F_f)+(-1)^n\mathcal{T}_{n-1}^{\hbar
  \Delta_{\tmop{AF}}} (V_{g}^{\otimes(n-1)})\star_{\hbar\omega}
  \mathcal{T}_{0,m}^{\hbar \Delta_F}(F_f)\\
  &\qquad\qquad\qquad-\mathcal{T}_{n - 1, m}^{\hbar \Delta_F}(V_{g}^{\otimes(n-1)}
  \otimes F_f)+(-1)^{n-1}\mathcal{T}_{n-1}^{\hbar
  \Delta_{\tmop{AF}}} (V_{g}^{\otimes(n-1)}) \star_{\hbar \omega}
  \mathcal{T}_{0, m}^{\hbar \Delta_F}(F_f) \bigg]=0\,.
\end{align*}
This concludes the proof of item $1.$

\vskip .2cm

\noindent Focusing instead on item $2.$, the rationale at the heart of the proof consists of writing
  \[ \mathcal{T}_{\ell}^{\hbar \Delta_{\tmop{AF}}} (V_g \otimes \ldots
     \otimes V_g) \star_{\hbar \omega} \mathcal{T}_{n - \ell, m}^{\hbar
     \Delta_F} (V_g \otimes \ldots \otimes V_g \otimes F), \]
in terms of connected time-ordered products, see Definition \ref{Def: connected products}, and to use the ensuing expression in Equation \eqref{Eq:retarded-products} obtaining
\begin{align*}
      R_{n, m} (V_g^{\otimes n}, F)  & =  \left( \frac{i}{\hbar} \right)^n\sum_{I
       \subset \{ 1, \ldots, n\}} (- 1)^{| I |} \sum_{P \subset
       \tmop{Part} \{ I \}} \sum_{Q \subset \tmop{Part} \left\{ I^c \bigsqcup
       \{ 1, \ldots, m \} \right\}}\\
       &   \qquad\qquad \qquad \left( \prod_{p \in P} \mathcal{T}_{| p |}^{\hbar
       \Delta_{\tmop{AF}}} (p)^c \right) \star_{\hbar \omega} \left( \prod_{q
       \in Q} \mathcal{T}_{| q |}^{\hbar \Delta_F} (q)^c \right) .
 \end{align*}
  On account of the dependence of (anti)time-ordered products on $\hbar$ as per Equation~\eqref{Eq:homogeneity-T}, it follows that
  \begin{equation}
    \prod_{p \in P} \mathcal{T}_{| p |}^{\hbar \Delta_{\tmop{AF}}} (p)^c
    =\mathcal{O} (\hbar^{| I | - | P |}), \qquad \qquad \prod_{q \in Q}
    \mathcal{T}_{| q |}^{\hbar \Delta_F} (q)^c =\mathcal{O} (\hbar^{| I^c | +
    m - | Q |}) . \label{Eq:connect-comp}
  \end{equation}
  Yet, resorting to point $1.$ of this proposition, the only non-vanishing
  contributions are those where all the vertices associated with $V_g$ are
  contracted with at least one of $F$. Since $F$ has $m$ vertices,
  among these contributions the one with lowest order in $\hbar$ are those with $m$ disconnected components, each of which encompasses precisely one vertex of $F$. Hence we can infer that, overall, there exist at least $n \geqslant | P | + | Q | -
  m$ contractions. This is due to the fact that, denoting by $(X_j, Y_j, y_j)_{j = 1}^m$ the vertices associated with such connected components, there are at least $| X_j | + | Y_j | - 1$ contractions for any $j=1,\ldots,m$. As a consequence, since $\sum_{j = 1}^m | X_j | = | P |$ and $\sum_{j = 1}^m| Y_j | = | Q |$, we conclude that the minimum number of contractions is $n
  \geqslant | P | + | Q | - m$. 
  
  To conclude, on account of Equation~\eqref{Eq:connect-comp}, the lowest order in $\hbar$ contributing to the sum is 
  \[ (| I | - | P |) + (| I^c | + m - | Q |) + (| P | + | Q | - m) = | I | +
     | I^c | = n, \]
  compensating exactly the overall factor $\left( \frac{i}{\hbar} \right)^n$.
\end{proof}

\begin{remark}\label{Rem: existence classical limit}
	 It is important to observe that Proposition \ref{Prop:classical-limit}, and in particular point $2.$ entails that, at any perturbative order in $\lambda>0$, the classical limit exists.
\end{remark}

On top of Remark \ref{Rem: existence classical limit}, being the classical limit of $R_{\lambda V_g} (F)$ well defined, the remaining task is to prove that it coincides with $r_{\lambda V_g}(F)$ obtained via the classical M\"oller map as per Equation \eqref{Eq:first-classical-obs}. To this end, we recall that the algebra of classical observables is endowed with the Poisson brackets
\[ \{ \varphi (x), \varphi (y) \} = \Delta (x - y) . \]
Classical interacting observables are obtained as formal
power series in the coupling parameter $\lambda$ having coefficients in
$\mathcal{F}^V(\mathbb{R}^2)$, see Equation \eqref{Eq:first-classical-obs}. For the sake of readability we recall the formula here, at the level of integral kernel and adopting the notation $Y =
(y_1, \ldots, y_m)$,
\begin{equation}
  r_{\lambda V_g} (F) (Y) = \sum_{n \geqslant 0} \lambda^n \int_{t_1 \leqslant
  \ldots \leqslant t_n \leqslant t} \mathd \mu_{x_1} \ldots \mathd \mu_{x_n} g
  (x_1) \ldots g (x_n) \{ V  (x_1), \{ V (x_2), \ldots \{ V (x_n), F (Y) \}
  \ldots \} \} . \label{Eq:perturbative-solution}
\end{equation}
As already stated above, we focus on proving that the quantum interacting observables built in terms of retarded products via the Bogoliubov map, see Equation~\eqref{Eq:retarded-products}, converge to their classical counterpart
represented by Equation~\eqref{Eq:perturbative-solution}, see \cite[Sec. 5.3]{Dutsch-Fredenhagen-AQFT-perturb-th-loop-exp} for further details. The main ingredient that we employ is the following lemma.
\begin{lemma}
  \label{Lem:retarded-property-commutator}
  Let $h,f,g\in \mathcal{D}'(\mathbb{R}^2)$ be such that $\tmop{supp} (h)$ is contained in the past of a fixed
  Cauchy surface $\Sigma$, while $\tmop{supp} (f)$ and $\tmop{supp} (g)$ are contained in its future. Then, for any $n\in\mathbb{N}$, if holds that
  \[ R_{n + 1, m} (V_h \otimes V_g^{\otimes n}, F_f) = - \frac{i}{\hbar} [V_h,
     R_{n, m} (V_g^{\otimes n}, F_f)]_{\hbar \omega} . \]
\end{lemma}

\begin{proof}
  A key r\^ole in this proof is played by the combinatorial argument in the proof of item $1.$ of Proposition \ref{Prop:classical-limit}, the factor $\tilde{V}_g$ being replaced by $V_h$.
  \textit{Mutatis mutandis}, it holds that
  \begin{align}\label{Eq:new-combinatoric-star}
  \nonumber
R_{n+1, m} (V_h\otimes V_g^{\otimes n}, F_f)&=\left(\frac{i}{\hbar}\right)^{n+1}\bigg[\mathcal{T}_{n+1,m}^{\hbar\Delta_F}(V_h\otimes V_{g}^{\otimes n}\otimes F_f)+(-1)^{n+1}\mathcal{T}_{n+1}^{\hbar
  \Delta_{\tmop{AF}}} (V_h\otimes V_{g}^{\otimes n})\star_{\hbar\omega}
  \mathcal{T}_{0,m}^{\hbar\Delta_F}(F_f)\\\nonumber
  &\quad+\sum_{\ell = 1}^{n}(-1)^\ell\binom{n}{\ell-1}\mathcal{T}_{\ell}^{\hbar\Delta_{\tmop{AF}}} (V_h\otimes V_{g}^{\otimes(\ell-1)}) \star_{\hbar \omega}
  \mathcal{T}_{n+1-\ell, m}^{\hbar \Delta_F}(V_{g}^{\otimes(n+1-\ell)}
  \otimes F_f)\\
  &\qquad+(-1)^\ell\binom{n}{\ell}\mathcal{T}_{\ell}^{\hbar
  \Delta_{\tmop{AF}}} (V_{g}^{\otimes\ell}) \star_{\hbar \omega}
  \mathcal{T}_{n+1-\ell, m}^{\hbar \Delta_F} (V_h\otimes V_{g}^{\otimes(n-\ell-1)}\otimes F_f) \bigg].
  \end{align}
Resorting now to Equation \eqref{Eq:factoriz-time-ordering} and to its anti-time ordered counterpart, we can reformulate Equation \eqref{Eq:new-combinatoric-star} as
 \begin{gather*}
R_{n+1, m} (V_h\otimes V_g^{\otimes n}, F_f)=\\
=\left(\frac{i}{\hbar}\right)^{n+1}\bigg[\mathcal{T}_{n,m}^{\hbar\Delta_F}(V_{g}^{\otimes n}\otimes F_f)\star_{\hbar\omega}V_h+(-1)^{n+1}V_h\star_{\hbar\omega}\mathcal{T}_{n}^{\hbar\Delta_{\tmop{AF}}} (V_{g}^{\otimes n})\star_{\hbar\omega} \mathcal{T}_{0,m}^{\hbar\Delta_F}(F_f)\\
  \quad+\sum_{\ell = 1}^{n}(-1)^\ell\binom{n}{\ell-1}V_h\star_{\hbar\omega}\mathcal{T}_{\ell}^{\hbar\Delta_{\tmop{AF}}} (V_{g}^{\otimes(\ell-1)}) \star_{\hbar \omega}
  \mathcal{T}_{n+1-\ell, m}^{\hbar \Delta_F}(V_{g}^{\otimes(n+1-\ell)}
  \otimes F_f)\\
  \qquad+(-1)^\ell\binom{n}{\ell}\mathcal{T}_{\ell}^{\hbar
  \Delta_{\tmop{AF}}} (V_{g}^{\otimes\ell}) \star_{\hbar \omega}
  \mathcal{T}_{n+1-\ell, m}^{\hbar \Delta_F} (V_{g}^{\otimes(n-\ell-1)}\otimes F_f)\star_{\hbar\omega}V_h \bigg]=\\
  =\frac{i}{\hbar}\bigg[R_{n, m} (V_g^{\otimes n}, F_f)\bigg]\star_{\hbar\omega}V_h
  -\frac{i}{\hbar}V_h\star_{\hbar\omega}\bigg[R_{n, m} (V_g^{\otimes n}, F_f)\bigg]\\
  =- \frac{i}{\hbar} [V_h, R_{n, m} (V_g^{\otimes n}, F_f)]_{\hbar\omega},
  \end{gather*}
which entails the sought statement.
\end{proof}

Let us consider a family of points $(x_1, \ldots, x_n) \in \mathbb{R}^{2n}$
such that $x_i \neq x_j$ if $i \neq j$. At the level of integral kernels
and, out of an iterative application of Lemma~\ref{Lem:retarded-property-commutator}, {\tmem{cf.}} {\cite[Sec. 5.3]{Dutsch-Fredenhagen-AQFT-perturb-th-loop-exp}, denoting by $X' =
(x'_1, \ldots, x'_m)$ and by $t' \assign \min \{ t'_1, \ldots, t_m' \}$ we obtain
\begin{align}\label{Eq:commutator-Bogoliubov}
\nonumber
    R_{n, m} (V (x_1) \ldots V (x_n), F (X')) & =  \left( - \frac{i}{\hbar}
    \right)^n \sum_{\pi \in \mathcal{S}_n} \vartheta (t' - t_{\pi (n)})
    \vartheta (t_{\pi (n)} - t_{\pi (n - 1)}) \ldots \vartheta (t_{\pi (2)} -
    t_{\pi (1)}) \times\\
    &   \qquad \times [V (x_{\pi (1)}), [V (x_{\pi (2)}), \ldots [V (x_{\pi
    (n)}), F (X')]_{\hbar \omega} \ldots]_{\hbar \omega}]_{\hbar \omega},
\end{align}
where $\mathcal{S}_n$ is the group of permutations of $n$-indices.
Equation \eqref{Eq:commutator-Bogoliubov}, together with the convergence result for the series defining the quantum interacting observables, see Remark \ref{Rem: existence classical limit}, entails that \cite{Dutsch-Fredenhagen-AQFT-perturb-th-loop-exp}
\begin{equation}
  \lim_{\hbar \rightarrow 0^+} R_{\lambda V_g} (F) = r_{\lambda V_g} (F),
  \label{Eq:quantum-classical-limit}
\end{equation}
since the quantisation procedure is designed in such  a way that
\[ - \frac{i}{\hbar} [\cdot, \cdot]_{\hbar \omega} \rightarrow \{ \cdot, \cdot
   \}, \qquad \qquad \tmop{as} \quad \hbar \rightarrow 0^+ . \]

\begin{remark}
  This argument confirms the result of Proposition~\ref{Prop:classical-limit}
  since it shows that the above is a Taylor series expansion in $\hbar$
  and not a Laurent one as for the $S$-matrix. We stress that,
  for a finite value of $\hbar$, Equation~\eqref{Eq:commutator-Bogoliubov}
  involves a number of contributions which is strictly larger than that  of those involved in its classical counterpart, Equation~\eqref{Eq:perturbative-solution}. This is due to the fact that the quantum commutators $[\cdot, \cdot]_{\hbar\omega}$ are constructed with respect to an exponential product, while this is not the case for the classical Poisson bracket $\{ \cdot \comma \cdot\}$. Using a language typical of quantum field theory, we could rephrase the statement as a consequence of the fact that, in an interacting quantum theory, at a perturbative level, one has to take into account loop diagrams, which do not have a counterpart in the classical scenario.
\end{remark}
As a last step the existence of the classical limit must be translated to the study of the underlying stochastic sine-Gordon model.
\begin{theorem}\label{Thm: hbar limit}
	Let $F = \Phi^{\otimes m}\in\mathcal{F}_{\mu c}(\mathbb{R}^2)$, $m \in \mathbb{N}$. Denoting by $\Gamma_Q$ the deformation map as per Equation \eqref{Eq:GammaQ}, by $R_{\lambda V_g} (F)$ the interacting version of the observable $F\in\mathcal{F}_{\mu c}(\mathbb{R}^2)$, see Equation \eqref{Eq:Bogoliubov map} and by $r_{\lambda V_g} (F)$ the corresponding classical interacting observable obtained via the classical M\"oller map introduced in Equation \eqref{Eq:first-classical-obs}, it holds that
  \[ \lim_{\hbar \rightarrow 0^+} \mathGamma_Q [R_{\lambda V_g} (F)] =
     \mathGamma_Q [r_{\lambda V_g} (F)].\]
\end{theorem}

\begin{proof}
  This result is a direct consequence of Proposition \ref{Prop:classical-limit}. As a matter of fact we observe that
  \[ \lim_{\hbar \rightarrow 0^+} \mathGamma_Q [R_{\lambda V_g} (F)] =
     \lim_{\hbar \rightarrow 0^+} \lim_{k \rightarrow \infty} \sum_{n = 0}^k
     \mathGamma_Q [R_{\lambda V_g} (F)]_n = \lim_{\hbar \rightarrow 0^+}
     \lim_{k \rightarrow \infty} \sum_{n = 0}^k \mathGamma_Q [R_{n, m} (F)],
  \]
  and we recall that $R_{n, m} (F)$ can be decomposed in the sum of a term of
  order zero in $\hbar$, which coincides with $r_{n, m} (F)$, and of additional terms all of order
  $\mathcal{O} (\hbar)$, denoted by $\mathfrak{O}_n$:
  \[ R_{n, m} (F) = r_{n, m} (F) +\mathfrak{O}_{n, m} (F). \]
  Thus, 
  \[ \lim_{\hbar \rightarrow 0^+} \lim_{k \rightarrow \infty} \sum_{n = 0}^k
     \mathGamma_Q [r_{n, m} (F)] = \lim_{k \rightarrow \infty} \sum_{n = 0}^k
     \mathGamma_Q [r_{n, m} (F)] = \mathGamma_Q [r_{\lambda V_g} (F)], \]
  where, in the first step, we exploited that $\lim_{k \rightarrow \infty}
  \sum_{n = 0}^k \mathGamma_Q [r_{n, m} (F)]$ does not depend on $\hbar$ while
  the second one is a direct consequence of Equation~\eqref{Eq:quantum-classical-limit} together with the convergence results proven in Section~\ref{Sec:interplay}. Then, the following chain of identities is satisfied
\begin{align*}
       \lim_{\hbar \rightarrow 0^+} \mathGamma_Q [R_{\lambda V_g} (F)] & = 
       \mathGamma_Q [r_{\lambda V_g} (F)] + \lim_{\hbar \rightarrow 0^+}
       \lim_{k \rightarrow \infty} \sum_{n = 0}^k \mathGamma_Q
       [\mathfrak{O}_n]\\
       & =  \mathGamma_Q [r_{\lambda V_g} (F)] + \lim_{\hbar \rightarrow
       0^+} \lim_{k \rightarrow \infty} \sum_{n = 0}^k \hbar \mathGamma_Q
       [\hbar^{- 1} \mathfrak{O}_n]\\
       & =  \mathGamma_Q [r_{\lambda V_g} (F)] + \lim_{\hbar \rightarrow
       0^+} \hbar \lim_{k \rightarrow \infty} \sum_{n = 0}^k \mathGamma_Q
       [\hbar^{- 1} \mathfrak{O}_n] .
\end{align*}
  We observe that, on account of Proposition \ref{Prop:classical-limit}, $\lim_{k \rightarrow
  \infty} \sum_{n = 0}^k \mathGamma_Q [\hbar^{- 1} \mathfrak{O}_n]$ is
  absolutely convergent and bounded for $\hbar \rightarrow 0^+$ and therefore $\lim_{\hbar \rightarrow 0^+} \hbar
  \lim_{k \rightarrow \infty} \sum_{n = 0}^k \mathGamma_Q [\hbar^{- 1}
  \mathfrak{O}_n] = 0$.
  As a consequence
  \[ \lim_{\hbar \rightarrow 0^+} \mathGamma_Q [R_{\lambda V_g} (F)] =
     \mathGamma_Q [r_{\lambda V_g} (F)] , \]
  which is the sought after identity.
\end{proof}

\section*{Acknowledgments}
The authors are grateful to Nicola Pinamonti for the enlightening discussions about the topic. We are grateful to Beatrice Costeri for the comments on this manuscript. The work of A.B is supported by a PhD fellowship of the University of Pavia and partly by the GNFM-Indam Progetto Giovani {\em Feynman propagator for Dirac fields: a microlocal analytic approach}, CUP E53C22001930001, whose support is gratefully acknowledged. A.B. is also grateful to the Sorbonne Universit\'e for the kind hospitality during the realization of part of this project. The research of P.R. is supported by a postdoc fellowship of the University of Bonn, which is gratefully acknowledged. 

\appendix
\section{Comparison with a perturbative approach}

In this short appendix, we compare the non-perturbative results concerning the stochastic sine-Gordon model constructed in this work with the one that  would be obtained by means of a purely perturbative approach to the solution theory of the stochastic Sine-Gordon equation \eqref{Eq: sine-Gordon equation}. We observe that Proposition \ref{Prop:classical-limit} provides a rationale for computing at a perturbative level the coefficients $\mathGamma_Q [r_{n,m} (F)] $ of the expectation value of any observable $F$. Yet, for the sake of simplicity, we shall discuss here only the case of an interacting field $\mathGamma_Q [r_{\lambda V_g} (\Phi_f)]$, \textit{i.e.}, a solution to Equation \eqref{Eq: sine-Gordon equation}. Our goal is to confirm that the results in this paper are in agreement with a perturbative approach at any order in the coupling parameter. Hence we first compute the lower order contributions to the solution of the functional version of Equation \eqref{Eq: sine-Gordon equation}, {\it i.e.},
\begin{equation}
	\Psi=\lambda a\Delta^R\circledast g\sin(a\Psi)+\Phi_f=-\lambda\Delta^R\circledast  V^{(1)}_{g}(\Psi)+\Phi_f,\qquad\Psi\llbracket\lambda\rrbracket=\sum_{n=0}^\infty\lambda^n\Psi_n,\qquad \Psi_n\in\mathcal{F}_{loc},
\end{equation}
where $\circledast$ denotes the stochastic convolution, see {\it e.g.} \cite{DDRZ21}, $\Phi_f$ is the functional introduced in Example \ref{Ex: basic functionals} with $f\in\mathcal{D}(\mathbb{R}^2)$, while $V^{(1)}_{g}$ is the first order functional derivative of the vertex functional introduced in Equation \eqref{Eq: vertex functional}. The coefficients of the perturbative expansion $\Psi\llbracket\lambda\rrbracket$ can be computed via the formula 
\begin{equation}\label{Eq: derivatives perturbatives coefficients}
	\Psi_n:=\frac{1}{n!}\frac{d^n}{d\lambda^n}\Psi\llbracket\lambda\rrbracket\Big\vert_{\lambda=0}.
\end{equation}
Focusing the attention on to the lower order contributions, Equation \eqref{Eq: derivatives perturbatives coefficients}  yields
\begin{align}\label{Eq: lower order terms}	
	&\Psi_0=\Phi,\nonumber\\
	&\Psi_1=\Delta^R\circledast V^{(1)}_{g}\\
	&\Psi_2=\Delta^R\circledast(V^{(2)}_g\Delta^R\circledast V^{(1)}_g)\nonumber,
\end{align}
where $V^{(1)}_g$ is as per Equation \eqref{Eq: derivative vertex functional}, while the second order functional derivative $V^{(2)}_g$ can be defined accordingly. 
We claim that these coefficients coincide with those obtained via the classical M\"oller map as per Equation \eqref{Eq: r-map}. To wit, we calculate the quantum counterpart of the interacting field at a given order and, subsequently, we take the classical limit which has been proven to coincide with $r_{\lambda V_g}(\psi)$, see \cite{Dutsch-Fredenhagen-AQFT-perturb-th-loop-exp}. Specializing Equations \eqref{Eq:expansion-in-ret-prod} and \eqref{Eq:retarded-products} to $F=\Phi_f$, it holds that
\begin{equation*}
  R_{\lambda V_g} (\Phi_f) = \sum_{n \geqslant 0} \frac{\lambda^n}{n!} R_{n, m}
  (V_g^{\otimes n}, \Phi_f), 
\end{equation*}
where 
\begin{equation*}
  R_{n, m} (V_g^{\otimes n}, \Phi_f) = \left( \frac{i}{\hbar} \right)^n
  \sum_{\ell = 0}^n \binom{n}{\ell} (- 1)^{\ell} \mathcal{T}_{\ell}^{\hbar
  \Delta_{\tmop{AF}}} (V_g \otimes \ldots \otimes V_g) \star_{\hbar \omega}
  \mathcal{T}_{n - \ell, m}^{\hbar \Delta_F} (V_g \otimes \ldots \otimes V_g
  \otimes \Phi_f), 
\end{equation*}
is the $n$-th perturbative contribution to the interacting field, up to a factor $n!$ at the denominator. Proposition \ref{Prop:classical-limit} highlights how the only non-vanishing contributions are those for which all the $n$ terms $V_g$ are contracted at least once with the linear field $\Phi_f$. This implies that $ R_{n, m} (V_g^{\otimes n}, F)=\mathcal{O}(\hbar^0)$. The only contributions surviving the limiting procedure are those with \emph{exactly} $n$-contractions, the terms with more being multiplied by a positive power of $\hbar$, in agreement with the absence of loops in the classical framework. The trivial case $n=0$ yields $R_{0,1}(\Phi_f)=\Phi_f$. Considering the case $n=1$, it holds that
\begin{align*}
R_{1, 1} (V_g, \Phi_f) &= \frac{i}{\hbar}[ \mathcal{T}_{1, 1}^{\hbar \Delta_F} (V_g \otimes \Phi_f)-
  \mathcal{T}_{1}^{\hbar
  \Delta_{\tmop{AF}}} (V_g) \star_{\hbar \omega}
  \mathcal{T}_{0, 1}^{\hbar \Delta_F} (\Phi_f)]\\
  &=\frac{i}{\hbar}[V_g \star_{\hbar \Delta_F} \Phi_f-
V_g \star_{\hbar \omega} \Phi_f].
\end{align*}
Since only the contributions with exactly one contraction contribute to the classical case, it descends that
\begin{align*}
r_{1, 1} (V_g, \Phi_f) =\frac{i}{\hbar} \langle V^{(1)}_g, \hbar(\Delta_F-\omega) f\rangle=-\langle V^{(1)}_g,\Delta^A f\rangle,
\end{align*}
where, in the last identity, we used that $\Delta_F-\omega=i\Delta^A$, see Equation \eqref{Eq: feynman propagator}.
We conclude that 
\begin{align}\label{Eq: r_1}
r_{1, 1} (V_g, \Phi_f) =-\langle\Delta^RV^{(1)}_g,f\rangle.
\end{align}
To better grasp the structure of a perturbative solution of Equation \eqref{Eq: sine-Gordon equation}, the analysis of the second order is way more enlightening. At the quantum level it reads
\begin{align*}
  R_{2, 1} (V_g\otimes V_g, \Phi_f) =-&\frac{1}{\hbar^2}
 [\mathcal{T}_{2, 1}^{\hbar \Delta_F} (V_g \otimes V_g
  \otimes \Phi_f)-
  2\mathcal{T}_{1}^{\hbar
  \Delta_{\tmop{AF}}} (V_g) \star_{\hbar \omega}
  \mathcal{T}_{1, 1}^{\hbar \Delta_F} (V_g \otimes \Phi_f)\\
  &+\mathcal{T}_{2}^{\hbar
  \Delta_{\tmop{AF}}} (V_g\otimes V_g) \star_{\hbar \omega}
  \mathcal{T}_{0,1}^{\hbar \Delta_F} (\Phi_f)]\\
  =-&\frac{1}{\hbar^2}[V_g \star_{\hbar\Delta_F} V_g\star_{\hbar\Delta_F}  \Phi_f-2V_g\star_{\hbar\omega}(V_G\star_{\hbar\Delta_F}\Phi_f )+V_g \star_{\hbar\Delta_{AF}}(V_g\star_{\hbar\omega}\Phi_f)].
\end{align*}
As proven, the only non-vanishing contributions to the classical limit are those where there are exactly two contractions and where all the interacting vertices are contracted with $\Phi_f$. To analyze this scenario, we can resort to a graphical representation. The rationale is to encode the information carried by the underlying kernels within graphs subordinated to the following rules:
\begin{itemize}
	\item black dots represent occurrences of the vertex functionals $V_g$, while purple dots indicate the linear field $\Phi_f$, 
	\item the occurrence of a kernel is denoted by a segment joining the contracted functionals. To distinguish all possible bi-distributions involved in the computation, we associate black edges to $\hbar\Delta_F$, green edges to $\hbar\omega$ and red ones to $\hbar\Delta_{AF}$.
\end{itemize}
In view of these rules, taking into account the relevant multiplicities and the Leibniz rule, we obtain
\begin{equation}\label{Eq: graphs}
\frac{1}{2!}r_{2, 1} (V_g\otimes V_g, \Phi_f) =-
\begin{tikzpicture}[thick,scale=2]
\filldraw (.25,0.35)circle (1pt);
\filldraw (-.25,0.35)circle (1pt);
\filldraw[purple] (0,0) circle (1pt);
\draw (-.25,0.35) -- (.25,0.35);
\draw (0,0) -- (.25,0.35);
\end{tikzpicture}\,+
\begin{tikzpicture}[thick,scale=2]
\filldraw (.25,0.35)circle (1pt);
\filldraw (-.25,0.35)circle (1pt);
\filldraw[purple] (0,0)circle (1pt);
\draw[green] (-.25,0.35) -- (.25,0.35);
\draw (0,0) -- (.25,0.35);
\end{tikzpicture}\,+
\begin{tikzpicture}[thick,scale=2]
\filldraw (.25,0.35)circle (1pt);
\filldraw (-.25,0.35)circle (1pt);
\filldraw[purple] (0,0)circle (1pt);
\draw[green] (-.25,0.35) -- (0,0);
\draw[green] (-0.25,0.35) -- (.25,0.35);
\end{tikzpicture}\,-
\begin{tikzpicture}[thick,scale=2]
\filldraw (.25,0.35)circle (1pt);
\filldraw (-.25,0.35)circle (1pt);
\filldraw[purple] (0,0)circle (1pt);
\draw[red] (-.25,0.35) -- (.25,0.35);
\draw[green] (0,0) -- (.25,0.35);
\end{tikzpicture}\,.
\end{equation}
At the level of integral kernels the expression for the classical solution at second order in perturbation theory reads
\begin{gather}\label{Eq: r_2}
\frac{1}{2!}r_{2,1}(V_g\otimes V_g,\Phi_f)=\notag\\
=-\langle V^{(1)}_g,\Delta_F\langle V^{(2)}_g,\Delta_F f\rangle\rangle+\langle V^{(1)}_g,\omega\langle V^{(2)}_g,\Delta_F f\rangle\rangle+\langle\langle V^{(2)}_g\rangle,\omega f,\omega V^{(1)}_g\rangle\nonumber\\
-\langle V^{(1)}_g,\Delta_{AF}\langle V^{(2)}_g,\omega f\rangle\rangle\nonumber\\
=-\langle V^{(1)}_g,(\Delta_F-\omega)\langle V^{(2)}_g,\Delta_F f\rangle\rangle+\langle\langle V^{(2)}_g\rangle,\omega f,\omega V^{(1)}_g\rangle-\langle V^{(1)}_g,\Delta_{AF}\langle V^{(2)}_g,\omega f\rangle\rangle\nonumber\\
=-\langle V^{(1)}_g,i\Delta^A\langle V^{(2)}_g,\Delta_F f\rangle\rangle+\langle V^{(1)}_g,(\Delta_{AF}-i\Delta^A)\langle V^{(2)}_g,\omega f\rangle\rangle-\langle V^{(1)}_g,\Delta_{AF}\langle V^{(2)}_g,\omega f\rangle\rangle\nonumber\\
=-\langle V^{(1)}_g,i\Delta^A\langle V^{(2)}_g,\Delta_F f\rangle\rangle+\langle V^{(1)}_g,i\Delta^A\langle V^{(2)}_g,\omega f\rangle\rangle\nonumber\\
=\langle V^{(1)}_g,\Delta^A\langle V^{(2)}_g,\Delta^A f\rangle\rangle,
\end{gather}
where, in the third equality, we used once more the identity $\Delta_F-\omega=i\Delta^A$ as well as 
\begin{align*}
\langle\langle V^{(2)}_g,\omega f\rangle,\omega V^{(1)}_g\rangle&=\langle\omega V^{(1)}_g,\langle V^{(2)}_g,\omega f\rangle\rangle=\langle(\Delta_{AF}+i\Delta^R) V^{(1)}_g,\langle V^{(2)}_g,\omega f\rangle\rangle
\\
&=\langle V^{(1)}_g,(\Delta_{AF}+i\Delta^A)\langle V^{(2)}_g,\omega f\rangle\rangle.
\end{align*}
A direct inspection shows that, up to tje second order in $\lambda$, the perturbative coefficients calculated both via the perturbative expansion as per Equation \eqref{Eq: lower order terms} and via a classical limiting procedure, as per Equations \eqref{Eq: r_1} and \eqref{Eq: r_2}, coincide.

Eventually, an analogous comparison can be carried over at the level of expectation values, the only difference being the action of the map $\Gamma_Q$ implementing additional contractions. Even if from the standpoint of a perturbative analysis of the SPDE this boils down to an algorithmic procedure, its field theoretical counterpart relies heavily on Theorem \ref{Thm: hbar limit}.
\begin{remark}
	This graphical representation of the contributions to the classical version of the interacting field has the net advantage of allowing for an immediate extension to the computation of expectation values. As explained in Section \ref{Sec:microlocal-approach-to-spdes}, such information is encoded at the algebraic level by contracting pairs of fields into a kernel $Q\in\mathcal{D}'(\mathbb{R}^2\times \mathbb{R}^2)$ which encodes the stochastic properties of the free random field. Hence the graphical counterpart of this operation amounts to adding a colored edge representing $Q$. At a practical level the problem of computing expectation values of products of fields boils down to finding all admissible maximally connected graphs having a fixed number of vertices and edges, including the two-point function of the free random field. This viewpoint, which sheds light on how the stochastic information is handled at the same level of the quantum features of the problem, simplifies a two-steps procedure within the perturbative study of the solution, namely building the desired perturbative coefficients and applying the deformation map $\Gamma_{\cdot_Q}$ evaluated at the zero configuration.  
\end{remark}

\end{document}